\documentclass[onefignum,onetabnum]{siamart190516}

\newif\ifIEEE
\IEEEfalse

\newif\ifSIAM
\SIAMtrue

\ifSIAM
\else 

\OneAndAHalfSpacedXI

\usepackage{natbib}
\bibpunct[, ]{(}{)}{,}{a}{}{,}%
\fi

\usepackage{mathrsfs,amssymb,amsfonts,url,bbm,enumerate, amsbsy}
\usepackage{graphicx, subfigure, color}
\usepackage[draft]{fixme}
\usepackage{algpseudocode}
\usepackage{algorithm}
\usepackage[normalem]{ulem}
\usepackage{mathtools}

\ifIEEE
\usepackage{amsmath,amsthm}

\theoremstyle{plain}
\newtheorem{theorem}{Theorem}
\newtheorem{lemma}{Lemma}

\newtheorem{corollary}{Corollary}
\newtheorem{conjecture}{Conjecture}
\newtheorem{claim}{Claim}

\theoremstyle{definition}
\newtheorem{definition}{Definition}

\author{ 
	Arsalan Sharifnassab$^\dagger{} ^*$
	\thanks{$\dagger$ A. Sharifnassab is with the Department of Computing Science, University of Alberta, Edmonton, Canada; email: sharifna@ualberta.ca 
	}
	and John N. Tsitsiklis$^\ddagger$
	\thanks{$\ddagger$ J. N. Tsitsiklis is with the Laboratory for Information and Decision Systems, Dept.\ of Electrical Engineering and Computer Science, MIT, Cambridge MA, 02139, USA; email: jnt@mit.edu}
	\thanks{$*$ This work was partially carried out while A.\ Sharifnassab was a visitor at the Laboratory for Information and Decision Systems, MIT, Cambridge MA, 02139, USA; and partially when he was a postdoctoral associate at the Department of Electrical Engineering, Sharif University of Technology, Tehran, Iran.}
}

\fi

\ifSIAM 
\newsiamremark{remark}{Remark}
\newsiamremark{hypothesis}{Hypothesis}
\newsiamremark{example}{Example}
\newsiamthm{claim}{Claim}
\newsiamthm{conjecture}{Conjecture}

\author{
	Arsalan Sharifnassab\thanks{Department of Computing Science, University of Alberta, Edmonton, Canada (\email{sharifna@ualberta.ca}) 
}
	\and John N. Tsitsiklis\thanks{Laboratory for Information and Decision Systems, Electrical Engineering and Computer Science Department, MIT, Cambridge MA, 02139, USA 
		(\email{jnt@mit.edu}).}
	}

\usepackage{amsopn}

\title{Jumping Fluid Models and Delay Stability of Max-Weight Dynamics under Heavy-Tailed Traffic}

\else 

\renewenvironment{proof}{\emph{Proof}:\,}{$\quad\square$}

\TheoremsNumberedThrough     
\ECRepeatTheorems
\EquationsNumberedThrough    
\MANUSCRIPTNO{ }

\fi

\newif\ifCommentsOn
\CommentsOntrue    

\ifCommentsOn  
\fxsetup{ status=draft,   author=,    layout=inline,    theme=color}
\definecolor{fxwarning}{rgb}{0.8000,0.0000,0.0000} 
\renewcommand{\comment}[1]{ \textcolor{red}{\bf{[[}}\fxwarning{#1}\textcolor{red}{\bf{]]}} }
\else 
\renewcommand{\comment}[1]{}
\fi

\ifSIAM
\newcommand{\citep}[1]{\cite{#1}}
\newcommand{\argmax}[1]{\underset{#1}{\operatorname{argmax}}}
\newcommand{\argmin}[1]{\underset{#1}{\operatorname{argmin}}}

\else 
\RequirePackage[colorlinks, linkcolor=black, citecolor=blue]{hyperref}
\renewcommand{\argmax}[1]{\underset{#1}{\operatorname{argmax}}}
\renewcommand{\argmin}[1]{\underset{#1}{\operatorname{argmin}}}
\fi

\newcommand{\hide}[1]{}

\renewcommand{\hat}[1]{\widehat{#1}}
\renewcommand{\bar}[1]{\overline{#1}}
\renewcommand{\tilde}[1]{\widetilde{#1}}
\newcommand{\E}[1]{\mathbb{E}\left[#1\right]} 
\newcommand{\Exp}{\mathbb{E}}

\newcommand{\one}{\mathbbm{1}}
\newcommand{\R}{\mathbb{R}}
\newcommand{\Z}{\mathbb{Z}}

\newcommand{\nn}{\ell}
\newcommand{\gammast}{\gamma}
\newcommand{\Au}{\bar{A}}
\newcommand{\lamax}{\|\lamvec^*\|+\epsilon}
\newcommand{\mumax}{\bar\mu}

\newcommand{\lambarmax}{\bar{\lambda}}

\newcommand{\Mthree}{M_3}

\newcommand{\Mb}{\bar{M}}

\newcommand{\Drift}{\bar{\mathcal{D}}}  
\newcommand{\conv}{\mathrm{conv}}
\newcommand{\event}{\mathcal{E}}
\renewcommand{\S}{{\mathcal{M}}}
\newcommand{\Sbar}{\bar{\mathcal{M}}}  
\def\N{{\mathcal{N}}}
\newcommand{\vvec}{\mathbf{v}}
\newcommand{\tet}{\theta}
\newcommand{\Tet}{\Theta}
\newcommand{\Bvec}{\mathbf{B}}
\newcommand{\bvec}{\mathbf{b}}
\newcommand{\Avec}{\mathbf{A}}
\newcommand{\Qvec}{\mathbf{Q}}
\newcommand{\qvec}{\mathbf{q}}

\newcommand{\xvec}{\mathbf{x}}

\newcommand{\nvec}{\mathbf{n}}
\newcommand{\Jvec}{\mathbf{J}}
\newcommand{\yvec}{\mathbf{y}}
\newcommand{\uvec}{\mathbf{u}}
\newcommand{\Uvec}{\mathbf{U}}
\newcommand{\pvec}{\mathbf{p}}
\newcommand{\Nvec}{\mathbf{N}}
\newcommand{\lamvec}{{\boldsymbol{\lambda}}}
\newcommand{\muvec}{{\boldsymbol{\mu}}}
\newcommand{\nuvec}{{\boldsymbol{\nu}}}
\newcommand{\evec}{{\boldsymbol{e}}}
\newcommand{\xivec}{{\boldsymbol{\xi}}}
\newcommand{\gamvec}{{\boldsymbol{\gamma}}}
\newcommand{\equals}{=}
\newcommand{\dd}{d}
\newcommand{\Bev}{\mathcal{B}}
\newcommand{\notimply}{\makebox[3.7pt][l]{$\implies$}\not \phantom{aA}}

\newcommand{\F}{{\rm F}}
\newcommand{\rr}{r}
\def\Aclass{{\mathcal{A}}}

\renewcommand{\kappa}{k}

\newcommand\redsout{\bgroup\markoverwith{\textcolor{red}{\rule[0.5ex]{2pt}{.5pt}}}\ULon}
\newcommand\bluesout{\bgroup\markoverwith{\textcolor{blue}{\rule[0.5ex]{2pt}{.5pt}}}\ULon}

\def\jP{\mathbb{P}}
\def\vn{{\boldsymbol n}}
\def\vN{{\boldsymbol N}}
\def\vg{{\boldsymbol \gamma}}
\renewcommand{\Pr}{\mathbb{P}}

\begin{document}

\ifSIAM
\vspace{5pt}

\maketitle

\vspace{10pt}
\noindent
{\bf Abstract.}\ 
We say that a random variable is  \emph{light-tailed} if moments of order $2+\epsilon$ are finite for some $\epsilon>0$; otherwise, we say that it is \emph{heavy-tailed}. We study queueing networks that operate under the Max-Weight scheduling policy, for the case where some queues receive heavy-tailed and some receive light-tailed traffic.
Queues with light-tailed arrivals are  often delay stable (that is, expected queue sizes are uniformly bounded over time) but can also become delay unstable because of resource-sharing with other queues that receive heavy-tailed arrivals. 

Within this context, and for any given ``tail exponents'' of the input traffic, we develop a necessary and sufficient condition under which a queue is robustly delay stable, in terms of \emph{jumping fluid}  models---an extension of  traditional fluid models that allows for jumps along coordinates associated with heavy-tailed flows. Our result elucidates the precise mechanism that leads to delay instability, through a coordination of multiple abnormally large arrivals at possibly different times and queues and settles an earlier open question on the sufficiency of a particular fluid-based criterion.  Finally, we explore the power of Lyapunov functions in the study of delay stability.

\vspace{20pt}
\else 

\TITLE{Jumping Fluid Models and Delay Stability of Max-Weight Dynamics under Heavy-Tailed Traffic}

\ARTICLEAUTHORS{%
	\AUTHOR{Arsalan Sharifnassab}
	\AFF{Department of Electrical Engineering, Sharif University of Technology, Tehran, Iran, \EMAIL{a.sharifnassab@gmail.com}} 
\AUTHOR{John N. Tsitsiklis}
\AFF{Laboratory for Information and Decision Systems, Electrical Engineering and Computer Science Department, MIT, Cambridge MA, 02139, USA , \EMAIL{jnt@mit.edu}}
} 

\ABSTRACT{
	We say that a random variable is  \emph{light-tailed} if moments of order $2+\epsilon$ are finite for some $\epsilon>0$; otherwise, we say that it is \emph{heavy-tailed}. We study queueing networks that operate under the Max-Weight scheduling policy, for the case where some queues receive heavy-tailed and some receive light-tailed traffic.
	Queues with light-tailed arrivals are  often delay stable (that is, expected queue sizes are uniformly bounded over time) but can also become delay unstable because of resource-sharing with other queues that receive heavy-tailed arrivals. 
	
	Within this context, and for any given ``tail exponents'' of the input traffic, we develop a necessary and sufficient condition under which a queue is robustly delay stable, in terms of \emph{jumping fluid}  models---an extension of  traditional fluid models that allows for jumps along coordinates associated with heavy-tailed flows. Our result elucidates the precise mechanism that leads to delay instability, through a coordination of multiple abnormally large arrivals at possibly different times and queues and settles an earlier open question on the sufficiency of a particular fluid-based criterion.  Finally, we explore the power of Lyapunov functions in the study of delay stability.
	}


\maketitle

\fi

\section{Introduction} 
We say that a random variable is  \emph{light-tailed} if moments of order $2+\epsilon$ are finite for some $\epsilon>0$; otherwise, we say that it is
\emph{heavy-tailed}.
 We study queueing networks that  operate under the Max-Weight scheduling policy, for the case where some queues receive heavy-tailed traffic, while some other queues receive light-tailed traffic; our motivation stems from  the fact that heavy-tailed processes are often natural models of the inputs to computer and communications networks \citep{FossKZ11}. Queues that receive heavy-tailed traffic are naturally \emph{delay unstable,} that is, they incur infinite expected delay, as an immediate consequence of the Pollaczek–-Khintchine formula.  
 However, it is also known that due to the relatively complex Max-Weight dynamics, some of the queues that receive light-tailed traffic may also end up delay unstable\footnote{This phenomenon can arise under other scheduling policies as well, e.g., the generalized processor sharing policy \citep{BorsBJ03}.}. 
Our aim is to develop conditions that 
determine whether any particular queue  is delay stable or not.

This problem has been  studied extensively  \citep{Mark13, MarkMT14,  MarkMT18, NairJW15}, and 
a necessary condition for delay stability was given in \citep{MarkMT14, MarkMT18}. In particular,
 \cite{MarkMT18} considered the associated fluid model, initialized at zero, except for a positive initial condition at some queue that receives heavy-tailed arrivals.  If another queue happens to eventually become positive under that fluid model, one can then conclude  that the latter queue is delay unstable. 
This result led to the natural question whether this condition is also sufficient, i.e., whether
delay stability is guaranteed when any such fluid trajectory (with a positive initialization at 
any single one of the queues that receive heavy-tailed traffic) keeps the queue of interest at zero level.
Such a sufficiency result might appear plausible because
in models involving heavy-tailed random variables,  large fluctuations are often the consequence of a \emph{single} abnormally large value  in the underlying heavy-tailed random variables \citep{Pake75, Vera77, Anan9l, Asmu96, Durr80, FossKZ11}.

\subsection{{Our Contributions}}\label{ss:results}
We start in Section~\ref{sec:example} by showing that the above mentioned 
possible sufficiency result does not hold. We accomplish this by providing a fairly simple example in which a large arrival at any single heavy-tailed queue does not cause a certain queue of interest to grow, but a combination of two large arrivals, at two different heavy-tailed queues, can result in large backlogs at the queue of interest. 
We also provide necessary and sufficient conditions for delay stability in that particular example, which provide intuition for a possible more general result. Interestingly, delay instability manifests itself only when the
 tail exponents (to be defined precisely in Section~\ref{subsec:Heavy-tailed and exponential type}) 
 of the heavy-tailed  arrival processes lie in a specific range. 
 As a consequence,  we need to take these exponents explicitly into account, something that traditional fluid models cannot do. 
 
We then generalize, by developing general and tight (necessary and sufficient) conditions for delay stability, in terms of deterministic fluid-like models  in which there can be multiple jumps (in different queues and possibly at different times) at the heavy-tailed queues. 

Our conditions are not easy to test computationally, but this seems unavoidable: since the conditions are necessary and sufficient, the complexity of testing  them 
reflects the intrinsic complexity of testing delay stability. On the positive side, our conditions:
\begin{itemize}
\item[(a)]
provide a  conceptual understanding of the  mechanism that results in delay instability; 
\item[(b)]  can be checked in special cases, e.g., for the example in Section~\ref{sec:example}; see Section~\ref{ss:revisit}.
\end{itemize}

On the technical side,  our general conditions 
involve a small, but technically crucial, reformulation
of the delay stability problem. To understand the underlying issue, note that fluid models do not always lead to definite conclusions when the underlying system is marginally stable, but are generally 
effective when used to analyze ``robust'' properties. For this reason, we consider a network with given nominal arrival rates  and focus on \emph{robust} stability. Namely, we ask whether a certain queue is delay stable 
for all arrival processes with given tail exponents and for all
(possibly time-varying) arrival rates that lie in some open ball around the nominal ones.
When the problem is framed that way, 
definitive necessary and sufficient conditions for (robust) delay stability become possible.
Furthermore, with this formulation, it is only the nominal arrival rates and the tail exponents that matter,
as opposed to the details of the distribution of the input traffic.

Finally, earlier works \citep{Mark13} and \citep{MarkMT18} had shown that Lyapunov functions with certain structural properties can be used to certify delay stability.
But it was not known whether this methodology is ``complete,'' that is, whether delay stability can  always be established through a suitable Lyapunov function. Our results make progress in this direction, for the case of ``very heavy'' tails, 
that is, when there exists some $\gamma>0$  for which arriving traffic moments of 
order $1+\gamma$  are finite, but $\gamma$ can be arbitrarily small. For this regime, we derive some necessary and some sufficient conditions for delay stability: we show that a Lyapunov function of a special kind can be used to certify delay stability, together with a partial converse.

\subsection{Related Works on  Multiple Big Jumps }
As already mentioned, 
in several systems that involve heavy-tailed random variables,  large fluctuations are  the consequence of a single large jump  in the underlying heavy-tailed  variables. 
However, this is not always the case. There are known systems that have small fluctuations under a single large jump, and large deviations can only arise as a consequence of multiple jumps in heavy-tailed  random variables with suitable timing.
Early examples were studied in  \citep{JeleM03} and \citep{ZwarBM04} for single buffer system with multiple on/off input processes with  heavy-tailed periods. 
Furthermore, \cite{FossK12} provided conditions on the finiteness of moments of queue sizes in a multiserver G/G/s queue and 
showed the special effects caused by multiple jumps. 
Perhaps the closest work to ours that demonstrates the necessity of multiple big jumps is  \citep{ChenBRZ19},  which proposed a rare-event simulation technique to estimate the probability of rare events in stochastic systems that receive heavy-tailed inputs. As an application, they considered a queuing network with fixed fluid services and fixed routing, in which each queue receives independent heavy or light tailed exogenous arrivals. 
They characterized the tail-exponent of the queues in terms of a  knapsack-type constraint that involves the sum of numbers of large jumps in the different inputs that is required to make the queue of interest large, where the sum is weighted by the tail-exponents of the corresponding inputs. 
Such a knapsack-type constraint also appears in our bounds, and the tail-exponent fully determines the delay stability/instability of the queue of interest. 
Thus, our results parallel those in \citep{ChenBRZ19}, but for a different setting.
The main differences are that our network  operates under the more complex Max-Weight scheduler, and  our assumptions allow for  non i.i.d.  arrival processes. 
Furthermore, our 
proof techniques  are very different from those in \citep{ChenBRZ19}; it would be interesting to explore whether the techniques in \citep{ChenBRZ19} can be used to obtain alternative proofs for our setting (Theorem~\ref{th:main}).

\subsection{Outline}
The rest of the paper is organized as follows. 
We start in the next section with the details of our model and some definitions.  
In Section~\ref{sec:example}, we discuss an example that provides insights on the ways that arrival rates and tail exponents 
affect delay stability, and demonstrate that criteria based on traditional fluid models
are inadequate for the purpose of deciding delay stability. 
In Section~\ref{sec:JF}, we introduce a fluid-like model, which we call   \emph{jumping fluid model} and underlies our main result, the necessary and sufficient conditions for (robust) delay stability  that we present in Section~\ref{sec:main}. 
In Section~\ref{sec:lyap}, we study the power of Lyapunov functions for our problem.
In Sections~\ref{sec:proof JF=>RDS},~\ref{sec:proof RDS=>JF}, and~\ref{sec:proof lyap},
we provide the proofs of our  results. As the proofs are quite involved, they are presented  as a sequence of lemmas, with the proofs of the lemmas provided in Appendices~\ref{app:proof JF=>RDS} and~\ref{app:proof RDS=>JF}.
We discuss the results and directions for future research in Section~\ref{sec:discuss}.
Finally, in  Appendix~\ref{app:new}, we explore  alternative definitions of robust stability, and corresponding variants of our jumping fluid  conditions, and 
explain why they are unlikely to yield sharp necessary and sufficient conditions. 

\subsection{Notation}\label{ss:notation}
We collect here some notational conventions to be used throughout the paper.
We use boldface symbols to denote vectors, and ordinary font to denote scalars. 
For any vector $\vvec$, we use $v_i$ to denote its $i$th component, and $|\vvec|$ to denote 
the sum $|v_1|+\cdots+|v_n|$. 
 We also use the notation $[\vvec]^+$ to denote the vector with components $\max\{0,v_i\}$. Finally, we let ${\bf e}_j$ stand for the $j$th unit vector.

We use $\R_+$ and $\Z_+$ to denote the sets of nonnegative reals and nonnegative integers, respectively. 
Furthermore, for a vector $\vvec$, we write $\vvec\succeq {\bf 0}$ (respectively, $\vvec\succ {\bf 0}$) to indicate that all components are nonnegative  (respectively, positive). 
For any set $S$, we denote its convex hull by $\conv(S)$.

Throughout, $\|\cdot\|$ will stand for the Euclidean norm. 
Sometimes, we use the alternative notation $d(\xvec,\yvec)$ in place of $\|\xvec-\yvec\|.$
We also let $d({\bf x},S)$ be the distance of a vector ${\bf x}$ from a set $S$, i.e., $d({\bf x},S)=\inf_{{\bf y}\in S} \| {\bf x} - {\bf y}\|$.
We finally  use $\one(\cdot)$ to denote the indicator function, and $\log$ to denote the natural logarithm.

For any time function ${\bf x}(\cdot)$ which is right-continuous with left limits,  $\big(d {\bf x}/dt \big) (t)$ or
$\dot{\bf x}(t)$ stand for the \emph{right} derivative of ${\bf x}(t)$, with the implicit assumption that it exists, and $\xvec(t^-)$ stands for $\lim_{\tau\uparrow t} \xvec(\tau)$.

\medskip

\section{The model}\label{sec:model}
\subsection{Network model and the Max-Weight policy}
We consider a switched network that operates in discrete time. For simplicity and ease of presentation, we  restrict ourselves to single-hop networks. However,  our results are easily generalized to multi-hop networks of the type considered in \citep{AlTG19ssc}.
 
The network consists of  $\nn$ queues that buffer incoming packets (or jobs). For any $t\in\Z_+$, 
we let $\Qvec(t)$ be a nonnegative vector whose $j$th component is the length of the $j$th queue at time $t$. 
Packets  arrive to the queues according to a nonnegative stochastic vector arrival process, $\Avec(\cdot)$.
In particular,  $A_{j}(t)$ stands for the exogenous arrival to the $j$th queue at time $t$.
We assume that  the 
random variables $A_{j}(t)$, for different $j$ and $t$, are independent.
We refer to $\E{\Avec(t)}$ as the arrival rate vector at time $t$.

At each time $t$, the amount of service received by the queues is a nonnegative vector $\muvec(t)$, which is chosen by a scheduler from a finite set $\S$  of possible service vectors. 
The queue lengths then evolve according to 
\begin{equation} \label{eq:evolution MW}
\Qvec(t+1) = \big[\Qvec(t) - \muvec(t)\big]^+ + \Avec(t).
\end{equation}
As in \citep{AlTG19ssc}, we assume throughout the paper that for any $\muvec\in\S$, the set $\S$ also contains all vectors that result from setting some entries of $\muvec$ to zero. This assumption is naturally valid in most contexts.

We focus  exclusively on the popular
 Max-Weight (MW) scheduling policy, 
 \begin{equation}\label{eq:mu}
\muvec(t) \in \argmax{\nuvec \in \S}\ \nuvec^T \Qvec(t),
\end{equation} 
which is known to have favorable stability properties \citep{TassE92}: whenever there exists a policy under which the queues remain stable, MW will result in stable queues.
More specifically,  let us consider the set $\Sbar$, defined as 
the convex hull, $\conv(\S)$, of the set of all possible service vectors, which is the so-called  \emph{capacity region} of the network. 
For the case of i.i.d.\ arrivals, and 
under common stochastic assumptions, it is known that if the arrival rate vector lies in the interior of the capacity region, then MW will result in stable queues; conversely, if  the arrival rate vector lies outside the capacity region, the queues will be unstable under every scheduling policy~\citep{TassE92}.

\subsection{Light-tailed and heavy-tailed arrivals} \label{subsec:Heavy-tailed and exponential type}
In this subsection, we present some definitions related to the tails of the arrival process distributions.

To any nonnegative random variable $X$, we associate a \emph{tail exponent}, 
defined as the value of $\gamma$ at which $\E{X^{1+\gamma}}$ switches from finite to infinite:
\begin{equation} \label{eq:def gamma}
\gamma^*  \equals \sup \left\{\gamma: \ \E{X^{1+\gamma}}<\infty\right\}.
\end{equation}
As an  example, consider a continuous random variable whose  probability density function $f(\cdot)$ satisfies
\begin{equation} \label{eq:heavy example distrib}
c \cdot x^{-(2+\gamma)} \leq f(x) \leq \log^k x \cdot  x^{-(2+\gamma)}, \qquad \forall \ x\geq x_0,
\end{equation}
where $c$, $\gamma$, $k$, and $x_0$ are positive constants.
Such a distribution has a tail exponent equal to $\gamma$.

For an i.i.d.~arrival process, a tail exponent is unambiguously defined as the tail exponent of the marginal distribution at an arbitrary time.
However, once we bring robustness into the picture, we are led to consider arrival processes with non-identically distributed $\Avec(t)$. 
To any arrival process $A_j(\cdot)$, we  associate a tail exponent, $\gamma_j$,
defined as the largest value of $\gamma$ such that $A_j(t)$ is dominated by some nonnegative random variable $X$ with tail exponent $\gamma$, for all times $t$;
\begin{equation} \label{eq:def gamma process}
	\begin{split}
	 \gamma_j  \equals \sup \Big\{\gamma: \, \mbox{there exists a r.v. }X\geq 0 \quad\mbox{ s.t.}\quad   &X \mbox{ dominates } A_j(t) \mbox{ for all }t\ge0,  \\
	&\mbox{and }\E{X^{1+\gamma}}<\infty\Big\}.
	\end{split}
\end{equation}
Here, the term ``dominates'' refers to \emph{stochastic dominance}: 
a random variables $X$ dominates a random variable $Y$ if $\Pr(X>a) \ge \Pr(Y>a)$, for all $a\in\R$.
We say that $A_j(\cdot)$ is \emph{heavy-tailed} if $\gamma_j\le 1$, and \emph{light-tailed} otherwise.
Note that the above definition of a heavy-tailed process is aimed to capture the boundedness of  variance of the input distribution, and  differs from the conventional definition of heavy-tailed random variables, which  requires subexponential decay of  tail probabilities \citep{NairWZ20}.
The tail behavior of the different arrival processes is summarized by the vector
$\gamvec=(\gamma_1,\ldots,\gamma_\nn)$.

We note that as long as  $\E{A_j(t)}\le \mumax$,  for  some constant $\mumax$ and for all times $t$, 
the tail exponent $\gamma_j$ is well defined and lies in the range $[0,\infty]$. 
Indeed, suppose that $\lambda_j(t)=\E{A_j(t)}\le \mumax$, for some finite constant $\mumax$ and for all $t$.
Consider a random variable $X$ with probability density function $f_X(x)=(\mumax/x^2)\,\one(x\ge\mumax)$.  
The Markov inequality implies that 
$$\Pr\big(A_j(t)>a\big)\le \min\big\{1,\E{A_j(t)}/a\big\}\le \min\big\{1,\mumax/a\big\} = \Pr\big(X>a\big),\qquad \forall \ a>0.$$ 
In particular, $X$ dominates $A_j(t)$, for all $t$. Furthermore, $\E{X^{1+\gamma}}<\infty$, for every $\gamma<0$. Thus,  $\gamma_j$ is at least as large as any negative number, which implies that $\gamma_j\geq 0$. 

Conversely, if $\gamma_j>0$, then $\E{A_j(t)}$ is finite, and bounded as a function of $t$. We finally note that $\gamma_j$ may be infinite; this will be the case for bounded, or more generally, exponential-type, distributions.

\subsection{Robust delay stability} \label{subsec:delay stability}
As already mentioned, we are interested in the question of  \emph{delay stability} under the MW policy, in the presence of heavy-tailed arrivals.
Given a set of arrival processes, we say that queue $m$ is  \emph{delay stable} if, starting from $\Qvec(0)=0$, 
we have $\sup_t  \E{Q_{m}(t)}<\infty$.

The  arrival rates and the tail exponents 
do not provide enough information to decide whether we have delay stability, or even stability. For example, there is sometimes indeterminacy on the boundary of the capacity region. 
Furthermore, a tail exponent of $\gamma_j=1$ is compatible with  $\E{A_j^2(t)}$ being either finite or infinite, and this may be critical as far as delay stability is concerned (cf.~the Pollaczek-–Khintchine formula).
These difficulties, all related to indeterminacy at certain boundaries, can be circumvented by focusing on
a robust version of delay stability that incorporates two distinct elements.
\begin{itemize}
\item[(a)] We  require delay stability for \emph{all arrival process distributions} with given tail exponents. This allows us to focus on conditions that involve the tail exponents, while ignoring other details of these distributions. 
\item[(b)]  We require delay stability for 
\emph{all arrival rates}, possibly time-varying, in the vicinity of a given nominal rate. This allows us to avoid issues of indeterminacy when the arrival rate lies at a threshold between delay stability and delay instability. 
\end{itemize}

\begin{definition}[Robust Delay Stability] \label{def:RDS}
Let us fix a network with $\nn$ nodes,  and a set $\S$ of possible service vectors. 
We consider a tail  exponent vector $\gamvec$ with components in $[0,\infty]$, 
and a nominal arrival rate vector $\lamvec^*{\succeq {\bf 0}}$.
\begin{itemize}
\item[(a)] Given some  $\delta\geq 0$, an arrival process $\Avec(\cdot)$ belongs to the class $\Aclass_{\delta}(\gamvec; \lamvec^*)$ if:
\begin{itemize}
\item[(i)] The random variables $A_j(t)$, for different $j$ and $t$, 
are independent and have finite means. 
\item[(ii)] For every $j$,  the tail exponent of the process $A_j(\cdot)$ is $\gamma_j$. 
\item[(iii)] For every $t\geq 0$, we have $\big\|\E{\Avec(t)}-\lamvec^*\big\| \leq \delta$.
\end{itemize}
\item[(b)] For $m\in\{1,\ldots,\nn\}$, we say that  queue $m$ is \emph{Robustly Delay Stable (RDS)} if there exists some $\delta>0$ such that queue $m$ is delay stable for all arrival processes in the class $\Aclass_{\delta}(\gamvec; \lamvec^*)$.
\end{itemize}
\end{definition}

\section{A counterexample and the insufficiency of fluid models}\label{sec:example}
In this section, we discuss a simple example with the following  features:
\begin{itemize}
\item[(a)] Similar to existing examples, heavy-tailed arrivals to some queues can cause delay instability at a queue that receives light-tailed traffic.
\item[(b)] Tight criteria for delay instability need to take into account the values of the tail exponents. In particular, criteria that are based on traditional fluid models cannot be conclusive, because they do not involve the tail exponents.
\item[(c)]  Delay instability may emerge from coordinated large arrivals at multiple heavy-tailed queues.
\end{itemize}

Consider the three-queue network  in Fig.~\ref{fig:3 node example}. 
The set of possible service vectors $\S$ is such that at any time slot, up to two queues can be served, each with rate 1, but not all three queues can be served simultaneously. Thus,
at each time step, the MW policy chooses two queues with largest backlogs and serves each one  of them with rate~1. 

 The third queue 
receives deterministic arrivals, 
with $A_3(t)=\lambda<1$, for all $t\ge 0$, so that $\gamma_3=\infty$. 
The other two queues receive heavy-tailed traffic, with a density of the form \eqref{eq:heavy example distrib} 
and tail exponent $\gamma\in (0,1)$, 
and with  
rate $0.5$, i.e., $\E{A_1(t)} = \E{A_2(t)} = 0.5$.

\begin{figure} 
	\begin{center}
		{\includegraphics[width = .5\textwidth]{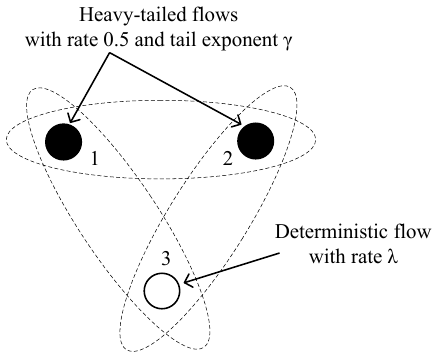}}
		\vspace{-.1cm}
	\end{center}
	\caption{A single-hop network 
	with three queues, two of which receive heavy-tailed traffic, while the third one receives deterministic traffic. 
	The dotted ellipses illustrate the different possible service vectors. 
	The third queue will be delay unstable for certain ranges of $\lambda$ and $\gamma$.}
	\label{fig:3 node example}
\end{figure}

 The total arrival rate is $1+\lambda$, which is less than 2, and the network is stable, in the conventional sense. 
The first two queues are automatically delay unstable because they receive heavy-tailed arrivals. However, the third queue can be either delay stable or delay unstable, depending on the values of $\lambda$ and $\gamma$. 
In what follows, we 
provide an informal discussion of the different cases.

\vspace{5pt}\noindent
{\bf Case 1.} ($0.5<\lambda<1$). 
In this case, queue 3 is delay unstable, through a scenario similar to those considered in earlier works \citep{MarkMT14}.
Intuitively, because of its heavy-tailed arrivals, $Q_1$  will occasionally receive large inputs. When that happens, with  $A_1(t_0)$ being large at some time $t_0>0$, 
$Q_1$ becomes and stays largest for some time. During that time,  the MW policy keeps serving $Q_1$, while the remaining service capacity is split between $Q_2$ and $Q_3$. Since $\lambda>0.5$, the sum of the arrival rates to $Q_2$ and $Q_3$ will exceed the aggregate service rate to these two queues over a time interval of duration $\Omega\big(A_1(t_0)\big)$. Thus, $Q_2$ and $Q_3$  build up to size $\Omega\big(A_1(t_0)\big)$. Since $\gamma<1$, we have  $\E{A_1^2(t)}=\infty$, 
and using this property, 
it can  be shown that $\E{Q_3(t)}$ grows unbounded.

The intuition behind the above argument is captured by an
existing criterion from \citep{MarkMT14} that examines certain trajectories $\qvec(\cdot)$ of a corresponding fluid model (also called fluid trajectories). 
In that fluid model, the arrival processes are replaced by deterministic flows with the same rates.  The initial conditions are    
$q_h(0)=1$ for some heavy-tailed queue (in our example, $h=1$ or $h=2$), and $q_j(0)=0$ for   $j\ne h$. (Because of symmetry, we only need to consider the case where $h=1$.) 
Let us say that  the ``zero fluid'' condition holds (ZF, for short) if the
the solution to the fluid model (known to be unique for the MW policy) keeps $q_3(\cdot)$ at zero. According to the criterion in  \citep{MarkMT14}, the 
failure of the ZF
condition certifies the delay instability of queue 3. This is indeed the case here: starting with the initial conditions $\qvec(0)=(1,0,0)$, it can be verified that for small positive times $t$, we have $q_3(t)=(\lambda-0.5)t/2>0$. 

In summary, for this particular case, we have delay instability, and this is correctly predicted by the failure of the ZF condition and available results.


\vspace{5pt}\noindent
{\bf Case 2.} ($0<\lambda<0.5$ and $\gamma>0.5$).
In this case,  queue 3 turns out to be delay stable (in fact, robustly delay stable). 
This is a consequence of
our general result in Section~\ref{sec:main}.

For this case, the fluid model initialized at $\qvec(0)=(1,0,0)$, satisfies $q_3(t)=0$ for all positive times, and the ZF condition holds. In particular, the ZF condition is ``aligned'' with delay stability. Observations of this nature,   led to the question whether   the ZF condition can be used as a certificate of delay stability~\citep{MarkMT18}. However,  this is not the case, as we discuss next.

\vspace{5pt}\noindent
{\bf Case 3.} ($0<\lambda<0.5$ and $\gamma\le 0.5$). 
In this case, the ZF condition holds, similar to Case 2. However, queue 3 turns out to be
delay unstable; this  follows from the proof\footnote{Strictly speaking, the results in Section~\ref{sec:main} only establish the absence of \emph{robust} delay stability,  not delay instability for the specific arrival process distributions of our example. However, a slight modification of the proof in Section~\ref{sec:proof RDS=>JF}   shows that queue 3 is indeed delay unstable.} of our main result, Theorem~\ref{th:main}.

We summarize here the underlying intuition. 
When $\gamma\le 0.5$, there is considerable probability that both queues 1 and 2 receive large inputs within a certain time interval. More concretely, for large values of $M$, 
there is probability  $\Omega(1/M)$ that both of $Q_1$ and $Q_2$  receive aggregate arrivals of size at least $3M$ within the time interval $[0,M]$.
 If this happens, both $Q_1$ and $Q_2$   become large enough  so that $Q_3$ receives no service during the interval $[M,2M]$. As a result, $Q_3(2M)\ge \lambda M$, with probability $\Omega(1/M)$. 
One can then use this fact to show that as $t$ increases,
$\E{Q_3(t)}$  grows unbounded, and queue 3 is delay unstable.

\vspace{5pt}
In summary, the presence or absence of delay stability depends on the tail exponents in a nontrivial manner. Furthermore, the ZF condition cannot discriminate between Cases 2 and 3, 
and thus cannot account for the different outcomes (delay stability in Case 2, delay instability in Case 3). 
In fact, the same obstacle arises with any other criterion that relies on  traditional fluid models, because fluid models do not take the tail exponents into account.  In order to make progress, we need to consider 
the probability that large inputs (or jumps) may arrive  within a certain time interval, as a function of the tail exponents (the probability is larger when the tail exponents are smaller).
Furthermore, as illustrated by Case 3, we may have to consider the effect of ``coordinated'' large jumps at more than one queue, within the same time interval.
This is accomplished by the model in the next section: it is still in the spirit of  traditional fluid models, except that it allows for jumps along the heavy-tailed flows, subject to a ``budget'' on allowed jumps, as determined by the tail exponents.


\section{Jumping fluid models} \label{sec:JF}
In this section,  we introduce a generalization of the fluid model, which allows for jumps along certain coordinates. We proceed by first defining a  traditional fluid model, and then modifying it.

\subsection{The fluid model}
A fluid model is a deterministic continuous-time dynamical system  that replaces the arrival process with a fluid stream of arrivals and  updates queue lengths along Max-Weight drifts. 
The literature provides a few, somewhat different but  equivalent, definitions of the fluid model \citep{ShahW12, MarkMT18}, which typically involve differential equations 
with boundary conditions. 
Here, we adopt an equivalent but somewhat simpler definition,\footnote{Recall our assumption in Section~\ref{sec:model} that for any $\muvec\in\S$, the set $\S$ also contains all vectors that result from setting some entries of $\muvec$ equal to zero. It was shown in Proposition~2 of~\citep{AlTG19ssc} that  under this assumption, the fluid model of  Definition~\ref{def:fluid model}, is equivalent to the more standard, albeit more complicated, definitions of  fluid models, based on differential equations 
with boundary conditions.}  from \citep{AlTG19ssc}.  

Recall the definition of $\S$ as the set of all possible service vectors. For any $\xvec\in\R_+^\nn$, we define 
	\begin{equation}\label{eq:def Sx}
		\S(\xvec)=\big\{ \muvec\in \S \mid\, \muvec^T \xvec \ge \nuvec^T \xvec, \ \forall \ \nuvec\in\S   \big\}
		=\argmax{\nuvec\in\S}\ {\nuvec^T \xvec},
	\end{equation}
which is the set of all possible service vectors that, for the given $\xvec$, attain the maximum in the definition of the MW policy; see~\eqref{eq:mu}.  We also let
	\begin{equation} \label{eq:def S(x)}
		\Sbar(\xvec)\equals \conv\big(	\S(\xvec)\big). 
	\end{equation}	
	Furthermore, given an arrival rate vector $\lamvec \succeq {\bf 0}$, and any $\xvec\in\R_+^\nn$, we let 
	\begin{equation}\label{eq:def S lambda}
		\Drift_\lamvec(\xvec)\,\equals\,  \lamvec-\Sbar(\xvec) \,=\, \conv\left(\left\{ \lamvec -\muvec\mid \muvec \in \argmax{\nuvec \in \S}\, \nuvec^T \xvec  \right\} \right),
\end{equation}
which is the set of candidate drifts when the queue length vector is $\xvec$.

For the definitions that follow, recall our convention that $\dot{\qvec}(t)$ denotes the \emph{right}  derivative of $\qvec(\cdot)$ with respect to time, at time $t$. 

\begin{definition}[Fluid Trajectories]\label{def:fluid model}
Let us fix a network with $\nn$ nodes,  a set $\S$ of possible service vectors, and  an arrival rate vector $\lamvec \succeq {\bf 0}$. 
	 A \emph{fluid trajectory} corresponding to $\lamvec$ is
	 a nonnegative, continuous,  
	 and right-differentiable $\ell$-dimensional function $\qvec(\cdot)$,  
that satisfies
	the  differential inclusion 
	\begin{equation}\label{eq:fluid evolution}
		\dot{\qvec}(t)\,\in\, \Drift_\lamvec\big(\qvec(t)\big), \qquad \forall\ t\geq 0.
	\end{equation}
\end{definition}

Given some $\lamvec\succeq {\bf 0}$ and $\qvec(0)\succeq {\bf 0}$, there always exists a unique (and  nonnegative) fluid trajectory $\qvec(\cdot)$ corresponding to $\lamvec$ and initialized at $\qvec(0)$ \citep{MarkMT18, AlTG19ssc}.

\subsection{Adding the jumps}
We now introduce jumping fluid (JF, for short) trajectories. 
Given a vector $\nvec=(n_1,\ldots,n_\nn)$ of nonnegative integers,
an $\epsilon$-JF$(\nvec)$ trajectory is a ``fluid trajectory with jumps,'' with $n_j$ \emph{positive} jumps in its $j$th component, 
while allowing for $\epsilon$-changes in the arrival rate $\lamvec$. More concretely:

\begin{definition}[$\epsilon$-Jumping Fluid Trajectories]\label{def:eJF}
Let us fix a network with $\nn$ nodes, and a set $\S$ of possible service vectors.  
We are given
an arrival rate vector $\lamvec^*\succeq {\bf 0}$, a nonnegative integer vector $\nvec$, 
and some $\epsilon\geq 0$.  An \emph{$\epsilon$-JF$(\nvec)$ trajectory\/} corresponding to $\lamvec^*$ is a nonnegative $\ell$-dimensional function $\qvec(\cdot)$, which is right-continuous with left limits, and right-differentiable, initialized with 
$\qvec(t)={\bf 0}$ for all $t<0$, 
and with the following properties:
\begin{itemize}
\item[(i)] Each component $q_j(\cdot)$ has $n_j$ points of discontinuity. \item[(ii)] If  $q_j(\cdot)$ has a discontinuity at some time $t$, then
$q_j(t)>q_j(t^-)$.
\item[(ii)] For every $t\geq 0$, 
\begin{equation}
\label{eq:eps jfm dyn} 
	\dot{\qvec}(t) \,\in\, \Drift_{\lamvec(t)}\big(\qvec(t)\big), 
	\end{equation}
for some nonnegative function $\lamvec(\cdot)$ which is right-continuous, piecewise constant, with finitely many points of discontinuity, and 
satisfies $\| \lamvec(t)- \lamvec^*\|\le\epsilon$, for all $t$.
\end{itemize}
Finally, an $\epsilon$-JF($\nvec$) trajectory with\footnote{If $\gamma_j=\infty$ and $n_j=0$, we use the convention $\infty\cdot 0=0$.} 
$\gamvec^T\nvec=\sum_{j=1}^{\nn}\gamma_j n_j \le 1$, is called   an $\epsilon$-JF($\gamvec$) trajectory. 

\end{definition}

Note that we allow jumps at time zero, in which case $\qvec(0)\neq {\bf 0}$. Note also that if $\epsilon=0$ and if jumps can only happen at time zero, then 
an $\epsilon$-JF trajectory is just a fluid trajectory, with the jumps of the JF trajectory determining the initial conditions of the fluid trajectory.

More generally, an $\epsilon$-JF trajectory consists of a concatenation of fluid trajectories, over the intervals where $\lamvec(\cdot)$ stays constant, together with a finite number of jumps. For this reason, once the jump times, jump sizes, and the function $\lamvec(\cdot)$ are specified, the results for fluid models extend and establish existence and uniqueness of the $\epsilon$-JF($\nvec$) trajectory.



Our next definition formalizes the requirement that a certain queue must stay at zero under all $\epsilon$-jumping fluid trajectories.

\begin{definition}[JF conditions] \label{def:eps jf cond} 
Let us fix a network with $\nn$ nodes, and a set $\S$ of possible service vectors.
We are	given an arrival rate vector $\lamvec^*\succeq {\bf 0}$, and a particular queue, $m$, of interest. 
\begin{itemize}

\item[(a)] Given a vector $\gamvec$ with components in $[0,\infty]$,   and some $\epsilon\ge0$,
 we say that the $\epsilon$-JF$(\gamvec)$ condition holds for queue $m$ and $\lamvec^*$  if for every  $\epsilon$-JF($\gamvec$) trajectory corresponding to $\lamvec^*$, and every $t\geq 0$, we have $q_m(t)=0$.

\item[(b)] Given a vector $\gamvec$ with components in $[0,\infty]$, we say that  the \emph{Robust Jumping Fluid} condition (RJF($\gamvec$), for short) holds for queue $m$ and 
$\lamvec^*$
 if there exists some $\epsilon>0$ 
such that  the $\epsilon$-JF$(\gamvec)$ condition holds for queue $m$ and $\lamvec^*$.

\end{itemize}
\end{definition}

Note the restriction on the number of jumps in terms of the tail exponents:
 the heavier the arrival processes (i.e., the smaller the tail exponents $\gamma_j$), the larger the number $n_j$ of jumps that we allow. As an example, if 
$\gamma_1\leq 1$ and
queue 1 is the only heavy-tailed queue, 
then we allow up to 
$\lfloor1/\gamma_1\rfloor$ jumps at queue 1, and no jumps at the other queues.


\section{Main result}\label{sec:main}

Our main result provides a necessary and sufficient condition for robust delay stability in terms of  JF trajectories. The proof is given in Sections~\ref{sec:proof JF=>RDS} and~\ref{sec:proof RDS=>JF}. 
 
\begin{theorem}\label{th:main}
Let us fix a network with $\nn$ nodes, a set $\S$ of possible service vectors, 
 an arrival rate vector $\lamvec^*\succeq {\bf 0}$, a particular queue, $m$, of interest, 
and a vector  $\gamvec$ of tail exponents with components in $(0,\infty]$. 
The queue $m$ 
 is robustly delay stable (RDS) if and only if the RJF$(\gamvec)$ condition holds for queue $m$ and $\lamvec^*$. \end{theorem}

\subsection{Some intuition}\label{ss:intuition}

We provide here a high-level explanation of our result. Some more refined intuition is provided by the proof outlines in Sections~\ref{s:outline-1} and~\ref{s:outline-2}. 

Let $M$ be a large constant. We say that the stochastic process $A_j(t)$ has a jump whenever it is larger than (approximately) $M$. Let $N_j$ be the number of jumps of $A_j(t)$ during the interval $[0,M]$,  and let
$\Nvec=(N_1,\ldots,N_{\nn})$. It turns out that for any nonnegative integer vector $\nvec$ the probability of the event $\Nvec=\nvec$ scales (approximately) like $M^{-\gamvec^T \nvec}$. The latter quantity is ``significant'' (in the sense that it makes an unbounded contribution to certain expected values) if and only if  $\gamvec^T \nvec \leq 1$. Thus, over an interval of length $M$, we can focus on sample paths for which the realized vector $\nvec$ of jump counts satisfies
$\gamvec^T\nvec \leq 1$, and examine whether such sample paths can cause the queue of interest to become large. 
We then argue that these sample paths are well-approximated by 
the $\epsilon$-JF($\nvec$) trajectories involved in the $\epsilon$-JF$(\gamvec)$ condition.

\subsection{Remarks} \label{ss:remarks}
We continue with some remarks on the scope of our result. 

\paragraph{Heavy-tailed queues} If queue $m$ is heavy-tailed, i.e., $\gamma_m\leq 1$, then the condition $\gamvec^T \nvec \leq 1$ allows $\nvec$ to be the $m$th unit vector.
With such a vector $\nvec$, we can have an $\epsilon$-JF($\nvec$) trajectory with a positive jump in the $m$th component, resulting in a positive value of $q_m(t)$. Thus, the RJF($\gamvec$) condition does not hold, and queue $m$ is not RDS. This is just a variation of the well-known fact that a queue with heavy-tailed arrivals is not delay stable.

\paragraph{Unstable systems}  Theorem~\ref{th:main} makes no stability assumptions. For unstable (or marginally stable) systems, some components of $\epsilon$-JF trajectories can grow arbitrarily large. On the other hand, these components do not necessarily have a substantial effect on  the queue, $m$, of interest. As long as the $m$th component of all $\epsilon$-JF trajectories stays at zero, queue $m$ will be RDS. 

\paragraph{Comparison with the ZF condition}
If $\gamma_j>1/2$ for all $j$, 
then an  $\epsilon$-JF trajectory can have at most one jump.  For stable systems, the RJF  condition boils down to a robust version of the ZF condition introduced
 in Section~\ref{sec:example}. In other words, for this case, a robust version of the ZF condition is a necessary and sufficient condition for robust delay stability. 
On the other hand, since the (robust version of) the ZF condition  is strictly weaker than the RJF condition, it does not provide necessary and sufficient conditions, for general $\gamvec$.

\paragraph{The light-tailed case}
Suppose that $\gamma_j>1$ for all $j$, so that all arrival processes are light-tailed. In this case, no jumps are allowed, and the RJF condition boils down to considering ordinary fluid trajectories, with  
slightly perturbed arrival rates. We have the following possibilities:
\begin{itemize}
	\item[(a)]
	If $\lamvec^*$ is in the interior of the capacity region $\Sbar$, then $\mathbf{0}$ is in the interior of $\Drift_{\lamvec^*}(\bf 0)=\lamvec^*-
	 \Sbar$. It then 
	turns out that $\mathbf{0}$ is an attracting fixed point of the fluid dynamics, the RJF condition holds, and we have RDS for all queues.  
	This is in line with existing results (e.g., see Theorem~4.5 of \citep{GeorNT06}).
	\item[(b)]
	If $\lamvec^*$ is on the boundary or outside the capacity region, and similar to our earlier discussion of unstable systems,
	queue $m$ could be either RDS or non-RDS, depending on whether (perturbed) fluid trajectories cause $q_m$ to become positive or not.
\end{itemize}

\paragraph{The case of
zero tail exponents}
Our definitions in Sections~\ref{sec:model} and \ref{sec:JF} are formulated for a nonnegative vector   $\gamvec$. However, our result is restricted to the case where this vector is positive. We comment on the reasons for this.

 
When $\gamma_j=0$, we are dealing with an arrival process for which $\E{A_j(t)}$ is finite, while $\E{A_j(t)^{1+\gamma}}$ may be infinite for every $\gamma>0$. Our proofs involve at certain places a division by $\gamma_j$, and break down if $\gamma_j=0$. It is not clear whether a similar result is possible when some of the tail exponents are zero.

\paragraph{Computational issues}
As already hinted in the Introduction, checking the RJF condition algorithmically appears to be a  hard computational problem, amenable only to impractical Tarski-like elimination algorithms. 
In one possible simplification, we might just consider $\epsilon$-JF trajectories with $\epsilon=0$,  so that $\lamvec(t)=\lamvec^*$, for all times $t$. 
 However, this would still leave the indeterminacy of the jump times and the jump sizes to be reckoned with. 
 Even worse, a restriction to this limited class of trajectories  does not seem to lead to necessary and sufficient conditions for any suitably modified notion of stability. 
See Appendix~\ref{app:relations} for further discussion. 
A related question is whether we could, without loss of generality, require all the jumps to occur at the same time, e.g.,  at time zero, thus eliminating the need to consider all possible values of the jump times. Unfortunately, this is not the case: there exist examples in which $\epsilon$-JF trajectories can drive a queue $m$ to a positive value, but this can happen only if we allow the jumps to occur at different times; see Appendix~\ref{app:timing} for an example.

\subsection{Our example, revisited}\label{ss:revisit}
On the positive side, Theorem \ref{th:main} elucidates the precise mechanism that leads to delay instability, through a coordination of multiple abnormally large arrival vectors, at possibly different times and queues. Furthermore, it allows us analyze simple problems, such as the one discussed in Section~\ref{sec:example}, which we do next.

Recall the network of three queues  in Section~\ref{sec:example}, in which $\gamvec=(\gamma,\gamma,\infty)$. 
For $\gamma$ in the range $(0,0.5]$, consider a jumping fluid trajectory $\qvec(\cdot)$, 
with $\nvec=(1,1,0)$, 
in which both $q_1$ and $q_2$ undergo  unit jumps at time $0$. With this trajectory, $q_3$ immediately starts to grow positive. 
Since, $\gamvec^T\nvec= 2\gamma\le1$, it follows that the RJF$(\gamvec)$ condition fails to hold. Theorem~\ref{th:main} then establishes that queue 3 is \emph{not} robustly delay stable. 

On the other hand, when $\gamma$ is in the range 
$(0.5,1]$, the constraint $\gamvec^T\nvec\leq 1$ allows 
 at most one jump, either in $q_1$ or $q_2$. Without loss of generality, we can assume that this jump takes place at time zero.  It turns out that  the fluid trajectories that start from either $\qvec(0)=(1,0,0)$ or $\qvec(0)=(0,1,0)$,  keep $q_3(\cdot)$  at zero if and only if $\lambda\le 0.5$. 
Therefore, for $\gamma$ in range $(0.5,1]$,  
 and also taking also robustness into account,
queue 3 is RDS if and only if $\lambda<0.5$.

\section{Robust delay stability via Lyapunov functions}\label{sec:lyap}

Lyapunov functions are a powerful  tool for the stability analysis of queueing networks \citep{TassE92, BertGT01, MaguBS16}, e.g., in throughput optimality proofs for the MW policy \citep{TassE92,Neel10}.  
\cite{Mark13} and \cite{MarkMT18} provided a sufficient condition for delay stability based on a class of piecewise linear Lyapunov functions, and used it to derive a sharp characterization of delay stability for a special class of networks, namely networks with disjoint schedules. Nevertheless, the Lyapunov approach  in \citep{Mark13,MarkMT18} has some drawbacks: (a) the condition provided therein is, in general, only sufficient for delay stability;  (b) it does not take  into account the tail exponents, even though they play an essential role in delay stability, as already discussed in Sections~\ref{sec:example} and~\ref{ss:revisit};  (c) the  Lyapunov functions considered were piecewise linear, which is perhaps  inadequate for the purpose of tight delay stability conditions.

In this section, we explore the power of Lyapunov functions, 
 for the case where the  tail exponents of the heavy-tailed queues may be arbitrarily close to zero,  so that the RJF  condition allows an arbitrarily large number of jumps.

For the remainder of this section, we assume that queues $1,\ldots,h$ can be heavy-tailed, where $h<\nn$, while the remaining queues are light-tailed.
Formally, we consider the set $\Gamma$ of tail coefficients, defined by
$$\Gamma=\{\gamvec \succ {\bf 0}: \gamma_j>1, \mbox{ for } j=h+1,\ldots,\nn\}.$$
We also fix an arrival rate vector $\lamvec^*\succeq {\bf 0}$ and 
a light-tailed queue $m>h$ of interest.

\begin{definition}[Special Lyapunov function]\label{def:nice lyap}
	For any $\epsilon>0$, we say that a function $V:\R_+^\nn\to\R_+$ is a \emph{special $\epsilon$-Lyapunov function} if:
	\begin{enumerate}
		\item $V$ is Lipschitz continuous, with Lipschitz constant $1$.
		\item $\dot{V}\big(\qvec(t)\big) \leq -\epsilon$ whenever $V\big(\qvec(t)\big)>0$, for all fluid trajectories $\qvec(\cdot)$   corresponding to arrival rate $\lamvec^*$.
		\item We have $V({\bf 0})=0$.  
		Furthermore, if $q_m>0$, then $V(q_m)>0$. 
		
		\item $V$ is nonincreasing along the coordinates associated with 
		heavy-tailed queues; that is, for $j=1,\ldots,h$, for any $\qvec\in\R_+^n$, and any $\alpha>0$, we have $V(\qvec+\alpha \evec_j)\le V(\qvec)$, where  $\evec_j$ is the $j$th unit vector.
	\end{enumerate}
\end{definition}

Special $\epsilon$-Lyapunov functions, as defined above, are quite similar to the functions considered in Theorem~2 of \citep{MarkMT18}. However, in contrast to  \citep{MarkMT18}, our special Lyapunov functions need not be piecewise linear. 
Our next result establishes a strong connection between the $\epsilon$-JF$(\gamvec)$   condition and the existence of special $\epsilon$-Lyapunov functions. The proof is given in Section~\ref{sec:proof lyap}.

\begin{theorem} \label{th:lyap}
	For any $\epsilon>0$, there exists a special $\epsilon$-Lyapunov function if and only if the  $\epsilon$-JF$(\gamvec)$ condition holds for every $\gamvec\in\Gamma$.
\end{theorem}
The special $\epsilon$-Lyapunov functions constructed in the proof of Theorem~\ref{th:lyap} are \emph{not} piecewise linear.  
We do not know whether Theorem~\ref{th:lyap} remains valid if we were restrict to piecewise linear function

Combining Theorems~\ref{th:main} and~\ref{th:lyap}, 
 we can establish a strong connection between robust delay stability and 
 special $\epsilon$-Lyapunov functions.
 In what follows, we say that queue $m$ is $\epsilon$-RDS$(\gamvec)$ if it is delay stable under all arrival processes in the class $\Aclass_{\epsilon}(\gamvec; \lamvec^*)$  in Definition~\ref{def:RDS}(a).

\begin{corollary}\label{cor:lyap}
Let us fix a network with $\nn$ nodes,  a set $\S$ of possible service vectors, an arrival vector $\lamvec^*\succeq {\bf 0}$, the number $h$ of heavy-tailed queues, and a light-tailed queue $m>h$.
\begin{itemize}
\item[(a)] If there exists a special $\epsilon$-Lyapunov function for some $\epsilon>0$, then queue $m$ is RDS for all tail exponents $\gamvec \in\Gamma$.
\item[(b)] Suppose that there exists some $\epsilon>0$ such that for all $\gamvec\in\Gamma$, queue $m$ is $\epsilon$-RDS$(\gamvec)$.  Then, there exists an $\epsilon$-special Lyapunov function.

\end{itemize}

\end{corollary}

The proof is provided in Section~\ref{s:cor-pf}. We conjecture that Corollary~\ref{cor:lyap}(b) can be strengthened to provide a converse to part (a). 

\begin{conjecture}\label{con:lyap}
If queue $m$ is RDS for all tail exponents $\gamvec \in\Gamma$, then for some $\epsilon>0$ there exists a special $\epsilon$-Lyapunov function.
\end{conjecture}

If Conjecture~\ref{con:lyap} is true, we will have a tight characterization: a queue will be RDS for all tail exponents $\gamvec \in\Gamma$ if and only if there exists a special $\epsilon$-Lyapunov function that testifies to this. Establishing the conjecture appears to be difficult. Technically it amounts to reversing the order of the quantifiers in the clause
``there exists $\epsilon>0$ such that for all $\gamvec\in\Gamma\ldots$'' in Corollary~\ref{cor:lyap}(b), and showing equivalence with the statement
``for every $\gamvec\in\Gamma$ there exists some $\epsilon>0$...''.

Our Lyapunov-based results 
are relevant to the case where nothing is known about the tail exponents of the heavy-tailed queues, other than the fact that they are positive. 
 On the other hand, 
Lyapunov functions are unlikely to provide useful characterizations of robust delay stability for specific values of the  tail exponents, because there is no apparent way of accounting for the number of jumps through Lyapunov functions.

\section{Proof of the ``if'' direction of Theorem~\ref{th:main} (RJF$\implies$RDS)}\label{sec:proof JF=>RDS}
In this section we provide  the proof of the 
``if'' direction of Theorem~\ref{th:main}, i.e., that the RJF condition implies RDS. 

Throughout this section, we consider a network with $\nn$ nodes, and a set $\S$ of possible service vectors. 
We fix an arrival rate vector $\lamvec^* \succeq {\bf 0}$, a particular queue, $m$, of interest, 
a vector  $\gamvec$ of tail exponents with components in $(0,\infty]$, and some
 $\epsilon>0$ for which the  $\epsilon$-JF$(\gamvec)$ condition holds for $\lamvec^*$ and queue $m$. Our goal is to establish robust delay stability for queue $m$. 
  
The proof is organized in a sequence of lemmas whose proofs are  collected in Appendix~\ref{app:proof JF=>RDS}. However, before proceeding to the formal arguments, it is helpful to provide an overview of the proof.

\subsection{Outline of the proof}\label{s:outline-1} 
Let us  fix an arbitrary time $T$. We aim at upper bounds for the probability $\jP\big(Q_m(T) \geq M\big)$, as $M$ gets large. As long as these bounds are a summable function of $M$, and independent of $T$, it will follow that $\E{Q_m(T)}$ is finite and a bounded function of $T$, which is our goal. Let us now fix some $M$, and keep it fixed throughout, except for the end of the proof. 

The proof relies  on various probabilistic bounds, as well as on  deterministic properties of the MW dynamics. Let us start with the probabilistic part, which is focused on showing that the stochastic system mostly follows the deterministic fluid dynamics, except for certain  ``jumps''  caused by the heavy tails of the arrival processes. We  define a threshold for what constitutes a  jump, and then 
develop a probabilistic bound on the numbers of jumps. 
A difficulty here is that if we use a fixed threshold, and because $T$ is arbitrary, a bound on the number of jumps is not possible. We handle this issue by using a threshold $\theta_t$ that increases almost linearly as we move further to the past, of the form
\begin{equation}\label{eq:def theta t}
	\theta_t=\frac{M+T-t}{{\eta\log (M+T-t)}},\qquad t=0,1,\ldots,T.
\end{equation}
for some positive constant $\eta$ to be defined later. 
At any time $t\le T$ and for any index $j$, we say that $A_j(t)$ is a jump if $A_j(t)>\theta_t$. 
Ignoring logarithmic factors, we show that the jump probability $\jP\big(A_j(t)> \theta_t\big)$ is of order at most $1/(M+T-t)^{1+\gamma_j'}$, where $\gamma_j'$ is slightly smaller than the tail exponent $\gamma_j$. 
By summing over $t$ and 
after some elementary calculations, 
we then obtain that  
$\jP\big(N_j =n_j\big)$ is of order at most $1/M^{n_j \gamma_j'}$,
where $N_j$ is the number of ``jumps'' of the $j$th arrival process during the interval $[0,T]$. 
Then, a further calculation shows that, if $\Nvec=(N_1,\ldots,N_\nn)$, then 
$\jP(\gamvec^T \Nvec>1)$  is of order at most $1/M^{\beta}$, for some  constant $\beta\in(1,2)$, and is therefore a summable function of $M$; see~Lemma~\ref{lem:prob event jump if}.
We then consider 
stochastic fluctuations in the arrival process, other than jumps.  We argue that they average out so that the cumulative arrival process follows its fluid counterpart. 
We refer to this as the ``small fluctuations event,'' and show that it occurs with probability at least 
$1-\nn/M^2$; see \eqref{eq:def event fluc if} and Lemma~\ref{lem:prob event fluc if}.

Having completed the probabilistic analysis, we then switch to  deterministic (sample path) considerations. 
Let $W(\nvec)$ be the set of all points $\qvec$ that can be reached by some $\epsilon$-JF$(\nvec)$ trajectory (see Definition~\ref{def:W}). 
As a first step, we exploit some special properties of the MW dynamics and show that $W(\nvec)$ is $\epsilon$-attracting; that is, any fluid trajectory that starts outside $W(\nvec)$ moves towards that set with rate at least $\epsilon$ (Lemma~\ref{lem:absorb}). We then consider a ``nice'' sample path, that is, a sample path 
for which the small fluctuations event occurs, and 
for which the realized vector of jump counts $\nvec$  satisfies 
$\gamvec^T\nvec\leq 1$. 
(As discussed earlier, ``nice'' sample paths have probability $1-O(M^{-\beta})$ , with $\beta>1$.)   
Our deterministic analysis, outlined in the next paragraph, shows 
that any such sample path stays within O($M$) distance from
the set $W(\nvec)$. Since
$\gamvec^T\nvec\leq 1$, 
the $\epsilon$-JF($\gamvec$) condition implies that any point in $W(\nvec)$ 
satisfies $q_m=0$. Thus, for nice sample paths, we have $Q_m(T)=O(M)$ and therefore, 
$\jP\big( Q_m(T)\ge M  \big)$ is of order at most  $1/ M^{\beta}$, for the constant $\beta\in(1,2)$ mentioned earlier. 
This readily implies a uniform upper bound on $\Exp\big[Q_m(T)\big]$.

The analysis of the dynamics under nice sample paths has two parts. 
We first study the dynamics \emph{between} jumps: we 
rely on the small fluctuations event, and then 
make use of a result from \citep{AlTG19ssc} (Theorem~\ref{th:sensitivity})  to ensure that small fluctuations in the arrivals result into comparably small changes in the resulting stochastic  trajectory; cf.~Lemma~\ref{lem:jump free}. Second, to understand what happens at jump times, we recall that the jump vectors $\nvec$ associated to nice sample paths satisfy $\gamvec^T\nvec\leq 1$. Such vectors $\nvec$  are  allowed in $\epsilon$-JF($\gamvec$) trajectories,
and therefore the jumps cannot take the $\epsilon$-JF($\nvec$) trajectory away from $W(\nvec)$. 
This implies that a ``nice'' sample path stays ``close''  to an $\epsilon$-JF($\nvec$) trajectory, and therefore has a ``small'' $Q_m$.

\subsection{Sensitivity of Max-Weight dynamics}

The proof, for both directions of the theorem, requires fairly precise bounding of the fluctuations of the stochastic trajectories. To this effect, we rely heavily on a
fluctuation bound for the MW dynamics, established in \citep{AlTG19ssc}:

\begin{theorem}[Theorem~2 from \citep{AlTG19ssc}] \label{th:sensitivity-old}
	Fix a network (i.e., the number of nodes and the set $\S$ of possible service vectors), operating under the MW policy, and let $\Qvec(\cdot)$ be the corresponding queue length stochastic process.  
	There exists a (deterministic)   constant $C\ge1$  such that, for any arrival rate vector $\lamvec\succeq {\bf 0}$, any $\qvec(0)\succeq {\bf 0}$,  any $t\ge 0$, and any sample path, if $\qvec(0)= \Qvec(0)$, then
	\begin{equation}
		\big\| \Qvec(t)-\qvec(t)  \big\| \, \le  \, C\left(
		1+\|\lamvec\| +\max_{k<t} \Big\| \sum_{\tau=0}^k \big(\Avec(\tau) -\lamvec   \big)  \Big\|   \right),
	\end{equation}
	where $\qvec(\cdot)$ is the fluid trajectory corresponding to $\lamvec$,  initialized at $\qvec(0)$.
\end{theorem}

We will actually use the following variant  of Theorem~\ref{th:sensitivity-old}, which allows for different initial conditions. The proof is given in Section~\ref{app:sens-pr}.

\begin{theorem}\label{th:sensitivity}
Under the same assumptions as in Theorem~\ref{th:sensitivity-old}, except that we allow for $\Qvec(0)$ and $\qvec(0)$ to be different, and for the same constant $C$, we have
$$
		\big\| \Qvec(t)-\qvec(t)  \big\| \, \le  \, \big\| \Qvec(0)-\qvec(0)\big\|+ C\left(	  
		1+ \|\lamvec\|+\max_{k<t} \Big\| \sum_{\tau=0}^k \big(\Avec(\tau) -\lamvec   \big)  \Big\|   \right).
$$
\end{theorem}

\subsection{Arrival process and jumps}
We now return to the formal proof.
Recall that throughout this section, we fix the network, $\lamvec^*\succeq {\bf 0}$, and the tail exponents
$\gamma_{{j}}\in (0,\infty]$. We assume  that the $\epsilon$-JF$(\gamvec)$ condition holds for $\lamvec^*$, queue $m$, and some $\epsilon>0$. 
We fix some positive integers $M$ and $T$;
these will remain fixed throughout, except for the end of the proof, and except for some additional assumptions that 
$M$ is ``large enough.'' 

We consider an arrival process $\Avec(\cdot)$ in the class $\Aclass_{\delta}(\gamvec; \lamvec^*)$ introduced in Definition~\ref{def:RDS}, with $\delta= {\gammast\epsilon }/{20C}$, 
where $C$ is the  constant
in Theorem~\ref{th:sensitivity} and
\begin{equation} \label{def:gammast}
	\gammast\, \equals\, \min_{{j}=1,\ldots,\ell} \gamma_{{j}}.
\end{equation}
In particular,
\begin{equation} \label{eq:A rate eps2 lamstar}
	\big\| \Exp [ \Avec(t)] -\lamvec^* \big\| \,\le\,\frac{\gammast\epsilon }{20C}, \qquad \forall\  t\geq 0.
\end{equation}
Our goal is to derive  an upper bound on  $\jP(\Qvec_{m}(T)\geq M)$ that holds uniformly for every arrival process $\Avec(\cdot)$ in $\Aclass_{\delta}(\gamvec; \lamvec^*)$, every $T$, and every large enough $M$. 
This will then be used to conclude that the RDS property holds.

Since we have assumed that the tail exponents are nonzero, we have $\gammast>0$. Furthermore, in order to simplify the proof, it is convenient to assume that 
\begin{equation}\label{eq:gam le1}
	\gammast\, \le\,1. 
\end{equation}

 \begin{claim}\label{cl:gamma-small}
The assumption $\gamma\leq 1$ can be made without loss of generality.
\end{claim}
\begin{proof}
Given a system, call it $S$, consider a new system $S'$ in which we add one more queue, queue 0, with $\gamma_0<1$, which ``does not interact'' with the others; for example, for any allowed service vector $\muvec$ in system $S$, we introduce a corresponding service vector $(1,\muvec)$ in system $S'$. This way the dynamics of queues $1,\ldots,\nn$ are the same in the two systems $S$ and $S'$. It follows that for $m\geq 1$, queue $m$ is RDS in system $S$ if and only if it is RDS in system $S'$. Furthermore, because of the lack of interaction, the RJF condition for queue $m$ holds in system $S$ if and only if it holds in system $S'$. 

Once we prove the result for the case $\gamma\leq 1$, we apply it to system $S'$ and obtain the equivalence of RDS and RJF for the latter system. Based on the above discussion, this also establishes the equivalence of RDS and RJF for the original system $S$.
\end{proof}

We define some more constants:
\begin{equation} \label{eq:mumax}
	\mumax = 1+\max_{\muvec\in\S} \|\muvec\|+ \|\lamvec^*\|+\epsilon,
\end{equation}
and 
\begin{equation}\label{eq:def eta}
	\eta\,\equals\, \frac{8000 \,C^2 \nn^2\mumax}{(\gamma\epsilon)^2}.
\end{equation}
Having fixed $M$, $T$, and $\eta$, we finally define $\theta_t$ as in \eqref{eq:def theta t}.

For $j=1,\ldots,\nn$ and $t=0,\ldots,T-1$, we say that $t$ is a jump time for $A_j(\cdot)$, if $A_j(t)>\theta_t$.
For any $j$ and any $\tau\in [0,T)$, we let $N_j(\tau)$ be the (random) number of jumps in $A_j(\cdot)$ that occur during $[0,\tau]$, and also define the corresponding vector 
$\Nvec(\tau)=\big(N_1(\tau),\ldots,N_\nn(\tau)\big)$.
To simplify notation, and as long as $T$ is fixed, 
we use $\Nvec$ and $N_j$ to refer to $\Nvec(T-1)$ and $N_j(T-1)$, respectively.  Note that with this definition, $\Nvec=\Nvec(T-1)$ includes all jumps that can affect $\Qvec(T)$.

We 
consider the event
\begin{equation}\label{eq:def event jump if}
	\event^{\textrm{jump}}(T,M) =\big\{ \gamvec^T \Nvec \le 1\big\}.
\end{equation}
We will argue that it occurs with high probability.

Let $F$ be the set of all $\vn\succeq {\bf 0}$  such that $1< \vg^T \vn \leq 2$. If $F$ is empty, then whenever $\vg^T \vn >1$, we must have
$\vg^T \vn >2$. If $F$ is nonempty, it has finitely many elements, which implies that the minimum $\min_{\vn\in F}  \vg^T \vn$ is attained and its value is greater than 1. In either case, we obtain that 
there exists a constant $\beta>1$ such that
\begin{equation} \label{def:beta}
\mbox{if  } \gamvec^T\vn>1, \mbox{  then  } \vg^T \vn \geq \beta^3.
\end{equation}
Without loss of generality, we can take $\beta$ to satisfy $1<\beta<2$.
In the next lemma we show that the probability that $\gamvec^T\vN>1$
 decays at least as fast as $1/M^{\beta}$. 

\begin{lemma} \label{lem:prob event jump if}
	There exists a constant  $M_1\ge 0$, independent of $T$, such that if $M\ge M_1$, 
then 
$$\Pr\big(\event^{\textrm{jump}}(T,M)\big)
=\Pr(\gamvec^T \Nvec \le 1)
 \geq 1- M^{-\beta}.$$
\end{lemma}
The proof is given in Appendix~\ref{appsec:proof lem jump if}, and relies on the intuition that the probability of a jump in the $j$th  arrival process 
scales at most as
$M^{-\gamma_j}$ (approximately), which then implies that the probability of the event $\{\Nvec=\nvec\}$  scales at most as  $M^{-\gamvec^T \nvec}$ (approximately).

\subsection{Fluctuations of the arrival process}
We now study the remaining fluctuations of the cumulative arrival processes, after we exclude the jumps.  
We begin with a concentration inequality for the sum of independent random variables. The proof of our next lemma is given in Appendix~\ref{appsec:proof Bernstein} and is essentially a reformulation of the Bernstein inequality; see, e.g., (1.21) in Appendix~1 of \citep{AnthB09}.
\begin{lemma} \label{lem:Bernstein if}
Suppose that $X_1,\ldots,X_n$ are independent random variables that satisfy
$X_i\in[0,b]$ and $\E{X_i}\le\lambarmax$, for some $b,\lambarmax >0$. Let $Y=X_1+\cdots+X_n$. Then, for any $z\ge 0$,
	\begin{equation} \label{eq:lem bern}
		\Pr\Big( \big|Y-\E{Y} \big| \, >\, z \Big)\, \le\, 2\exp\left(-\frac{z^2}{2b\,\left(\lambarmax n+ z/3\right)}  \right).
	\end{equation}
\end{lemma}

We consider the truncated process $\Avec^*(t)=\min\big\{ \Avec(t), \theta_t\big\}$, where the minimum is taken componentwise.  
We define the ``small fluctuations'' event
\begin{equation}\label{eq:def event fluc if}
	\event^{\textrm{fluc}}(T,M) = \Big\{ \Big\|  \sum_{\tau=t_0}^{t} \big(\Avec^*(\tau)-\lamvec^*\big)  \Big\|\le \frac{\gamma\epsilon}{10 C} \, \big(M+T-t_0\big), \ \  
	\mbox{for } 0\leq t_0\leq t<T \Big\},
\end{equation}
where $C$ is the constant in Theorem~\ref{th:sensitivity}.

\begin{lemma} \label{lem:prob event fluc if}
	There exists a constant  $M_2\ge 0$ independent of $T$, such that if $M\ge M_2$, then $\Pr\big(\event^{\textrm{fluc}}(T,M)\big)  \geq 1- \nn\, M^{-2} $.
\end{lemma}
The proof is somewhat long but straightforward. It relies on the concentration inequality in  Lemma~\ref{lem:Bernstein if}, and is given in Appendix~\ref{appsec:proof lem fluc if}.

\subsection{Deterministic analysis of the dynamics}

We start with some definitions. It is useful to recall here that $\epsilon$-JF trajectories start at zero just before time zero, 
but can become nonzero, at time zero or later, due to jumps or unstable drifts.

\begin{definition}\label{def:W} 
For any nonnegative integer  vector $\nvec$, we define
	$W(\nvec)$ as the set of all points in $\R^\nn_+$ that can be reached by some $\epsilon$-JF$(\nvec)$ trajectory. 
	\end{definition}

Because $\epsilon$-JF trajectories are by definition nonnegative, $W(\nvec)$ is a subset of $\R_+^{\nn}$.
Note that the set $W(\nvec)$ depends on $\epsilon$, but since $\epsilon$ is held fixed, we suppress this dependence from our notation.

\begin{definition}[$\epsilon$-attracting and $\epsilon$-invariant sets] \label{def:absorb}
	\ 
	\begin{itemize}
		\item[(a)]
		A subset $W$ of $\R^\nn_{+}$ is \emph{$\epsilon$-attracting} if   every 
		fluid trajectory $\qvec(\cdot)$  corresponding to $\lamvec^*$ and initialized at some arbitrary $\qvec(0)\succeq {\bf 0}$, satisfies
		\begin{equation}
			\frac{d}{dt} d\big(\qvec(t), W\big) \le -\epsilon,
		\end{equation}
		whenever $\qvec(t)\not\in W$. 	
		\item[(b)]
		A subset $W$ of $\R^\nn_{+}$ is \emph{$\epsilon$-invariant} if for any  $\lamvec \succeq {\bf 0}$ that 
		satisfies $\|\lamvec-\lamvec^*\|\le\epsilon$, and any fluid trajectory $\qvec(\cdot)$ corresponding to $\lamvec$ and 	initialized at some arbitrary	 $\qvec(0)\in W$, we have $\qvec(\tau)\in W$, for all $\tau\ge 0$.
	\end{itemize}
\end{definition}

We observe that if $W$ is  $\epsilon$-attracting, 
and if $0\leq t_0<t_1$, then every fluid trajectory satisfies
\begin{equation}\label{eq:absorb to W}
	d\big(\qvec(t_1),\,W\big) \,\le\, \max\Big\{0,\, d\big(\qvec(t_0),\,W\big)-(t_1-t_0)\epsilon\Big\}.
\end{equation}
It can be shown  that  an $\epsilon$-attracting set is $\epsilon$-invariant, but we do not need this fact. We are interested instead in the converse statement, that
every $\epsilon$-invariant set is $\epsilon$-attracting,  which we then apply to the set $W(\nvec)$; this is the subject of the next lemma.

\begin{lemma}\label{lem:absorb}
\begin{itemize} \item[(i)]
	Every $\epsilon$-invariant set is $\epsilon$-attracting. 
	\item[(ii)] For any $\nvec\in\Z_+^\nn$, the set $W(\nvec)$ 
	in Definition~\ref{def:W} is $\epsilon$-invariant and, a fortiori, $\epsilon$-attracting.
\end{itemize}
\end{lemma}
The proof is given in Appendix~\ref{appsec:proof lem absorb} and relies on the intuition that for an $\epsilon$-invariant set $W$, $\epsilon$-perturbations of $\lamvec^*$ cannot take the trajectories away from $W$. This means that there must be a drift towards $W$ that locally counteracts such perturbations. The non-expansive  property of the Max-Weight dynamics then enables us to extend this local counteraction result into the desired global attraction property.

Let us fix a sample path of the arrival process. For any $t\in [0,T)$, 
let $\nvec(t)$ be the realized value of $\Nvec(t)$, i.e., $\nvec(t)$ is the vector with the realized number of jumps in the process $\Avec(\cdot)$ up to time $t$ (see the paragraph preceding \eqref{eq:def event jump if}). With a bit of abuse of notation, we let
\begin{equation}
	W(t) \,=\, W\big( \nvec(t-1) \big).
\end{equation}
(The reason for the $t-1$ term on the right-hand side is that  we want to compare
$W(\cdot)$ and $\Qvec(\cdot)$, but   $\Qvec(t)$ is only affected by jumps that happen before time $t$.) 

\subsection{Some more intuition}
The general idea is to show that $\Qvec(\cdot)$ stays close to the set $W(\cdot)$ so that we can  ultimately exploit the fact that $q_m=0$, for every $\qvec\in W(T)$. There are two parts to the argument: 
\begin{itemize}
\item[(i)] If at a certain time, $Q_j(t)$ has a jump (i.e., a large increase), the set $W(t)$ expands along the $j$th coordinate and so the distance between $\Qvec(t)$ and 
$W(t)$ does not increase.
\item[(ii)] In between jumps, $\Qvec(t)$ follows the fluid trajectory, plus some fluctuations, within the range allowed by Lemma~\ref{lem:prob event fluc if}. These fluctuations get ``eliminated'' because the fluid trajectory is attracted to $W(t)$ (Lemma~\ref{lem:absorb}).
Note that~\eqref{eq:def event fluc if} allows for larger 
fluctuations in the far past (see the term $M+T-t_0$); however, for fluctuations in the far past, the motion in the direction of $W(t)$  happens for a longer time period, enough  to eliminate them. This explains the choice of the threshold $\theta_t$
 in \eqref{eq:def theta t}; the logarithmic term in the denominator is included for technical reasons.
\end{itemize}


\subsection{The distance from the invariant set}
Given times that satisfy $0 \leq t_0<t_1\le T$, we say that the interval $(t_0,t_1)$ is jump-free if  $A_{j}(\tau)\le \theta_\tau$, for all $j$ and all $\tau\in (t_0,t_1)$. Note that 
the initial time $t_0$ and the end time $t_1$ are allowed to
be jump times. Let $\Mthree\geq 1$ be a constant, independent of $T$, such that for any $M'\ge \Mthree$,
\begin{equation} \label{eq:q:def Mthree}
	\frac{\gammast\epsilon \,M' }{10} + \frac{\nn\,M'}{\eta\log M'} + \epsilon+ (\|\lamvec^*\|+1)C + \mumax \,\le\, \frac{\gammast\epsilon \,M' }{6}.
\end{equation}

\begin{lemma}\label{lem:jump free}
Suppose that  $M\ge \Mthree$  and fix times that satisfy $0\leq t_0<t_1\le T$. Consider a sample path under which the event 
$\event^{\textrm{fluc}}(T,M)$ occurs,  and the interval $(t_0,t_1)$ is jump-free. Then, 
	\begin{equation}
		d\big(\Qvec(t_1),\, W(t_1)  \big) \,\le\, \max \Big\{ 0,\,  d\big(\Qvec(t_0),\, W(t_0)  \big) - (t_1-t_0) \epsilon    \Big\} \,+\,  \frac{\epsilon\gamma}6\,  (M+T-t_0). 
	\end{equation}
\end{lemma}
Note that the lemma also applies when $t_1=t_0+1$, so that the set of integers in $(t_0,t_1)$ is empty. 
The proof, which is given in Appendix~\ref{appsec:proof lem jump free}, relies on the sensitivity bound in Theorem~\ref{th:sensitivity} and the fact that $W(t_0)$
is $\epsilon$-attractive.
The first term on the right-hand side, reflects the fact that a fluid trajectory is attracted to $W(\cdot)$, during the jump-free interval; the second term reflects the effect of  the smaller fluctuations during the interval $(t_0,t_1)$.

We then apply Lemma~\ref{lem:jump free} and use strong induction on $t$, for $t\leq T$, to establish the next lemma.
Its proof is given in Appendix~\ref{appsec:proof lem deterministic dQW}.

\begin{lemma}\label{lem:deterministic dQW}
Suppose that $M\ge \Mthree$. 
Consider a sample path under which the events 
 $\event^{\textrm{jump}}(T,M)$ and $\event^{\textrm{fluc}}(T,M)$ occur. Then, 
 $d\big(\Qvec(T), 
 W(T) 
 \big)\le M\epsilon/2$.
\end{lemma}

\subsection{Bounding $\Exp[\Qvec_m(t)]$}
Let $\Mb \equals \max\{M_1,M_2,\Mthree\}$, 
where $M_1$, $M_2$, and $\Mthree$ are the constants in  Lemma~\ref{lem:prob event jump if}, Lemma~\ref{lem:prob event fluc if}, and \eqref{eq:q:def Mthree}, respectively. 
Since these constants are independent of $T$, the constant $\Mb$ is also independent of $T$. 

Let us consider some
$M\ge \Mb$,  and a sample path under which the events $\event^{\textrm{jump}}(T,M)$ and $\event^{\textrm{fluc}}(T,M)$ occur.
In particular, we have $\gamvec^T \nvec \leq1$,  
where $\nvec$ is the realized value of $\Nvec$, i.e., the vector with the number of jumps until time $T-1$, for that particular sample path. The $\epsilon$-JF($\gamvec$) condition, which we have assumed to hold, implies that 
 every $\epsilon$-JF($\nvec$) trajectory with $\gamvec^T \nvec \leq 1$, 
 keeps $q_m$ at zero, and therefore, every vector $\qvec\in W (T)$ has $q_m=0$. Consequently, for every sample path
in  $\event^{\textrm{jump}}(T,M) \cap \event^{\textrm{fluc}}(T,M)$, we have
\begin{equation}\label{eq:Qm le M}
	Q_m(T)\,\le\, d\big(\Qvec(T), W(T)\big)\,\le\, 
	\frac{M\epsilon}{2},
\end{equation}
where the second inequality follows from Lemma~\ref{lem:deterministic dQW}.
Therefore, for any $M\ge \Mb$, we have
\begin{equation}\label{eq:Qm bound}
	\begin{split}
		\Pr\left(Q_m(T) > \frac{M\epsilon}2 \right) \,&\le\,  \Big(1 - \Pr\left( \event^{\textrm{jump}}(T,M) \right)   \Big) + \Big(1 - \Pr\left( \event^{\textrm{fluc}}(T,M) \right)   \Big)\\
		&\le\,  M^{-\beta}+ \nn M^{-2} ,
	\end{split}
\end{equation}
where the first inequality follows from \eqref{eq:Qm le M} and the union bound, and the second inequality is due to Lemmas~\ref{lem:prob event jump if} and~\ref{lem:prob event fluc if}.
Since $\beta$ is by definition greater than 1, the formula
$\E{Q_m(T)}\,=\, \int_{0}^{\infty} \Pr\big( Q_m(t) > M\big) \,dM $
implies that $\E{Q_m(T)}$ is bounded above by a constant that does not depend on $T$.
Furthermore note that this bound applies uniformly to all processes in the class
$\Aclass_{\delta}(\gamvec; \lamvec^*)$, for  $\delta={\gammast\epsilon }/{20C}$
(see \eqref{eq:A rate eps2 lamstar}). This shows that

queue $m$ is robustly delay stable and completes the proof of the first direction of Theorem~\ref{th:main}.


\medskip

\section{Proof of the reverse direction of Theorem~\ref{th:main} (RDS$\implies$
RJF)} \label{sec:proof RDS=>JF}
In this section we prove the 
reverse (``only if'') direction of Theorem~\ref{th:main}, i.e., that RDS implies that the $\epsilon$-JF$(\gamvec)$ condition holds for some $\epsilon>0$. The proof is organized in a sequence of lemmas whose proofs  are collected in Appendix~\ref{app:proof RDS=>JF}. 

We actually prove the contrapositive and start by assuming
that for every $\epsilon>0$, there exists  an $\epsilon$-JF$(\nvec)$ trajectory  $\qvec^{\epsilon}(\cdot) $, 
with $\gamvec^T\nvec\leq 1$, and 
such that $\qvec^\epsilon_m(t)>0$, at some  positive time $t$. 

We keep $\lamvec^*$, $\gamvec$, and $\epsilon$ fixed throughout the proof, and show that there exists an 
arrival processes in the class $\Aclass_{\epsilon}(\gamvec; \lamvec^*)$ for which queue $m$ is \emph{not} delay stable. Since $\epsilon$ can be arbitrarily small, this will imply that there exists no $\delta$ such that queue $m$ is delay stable for all arrival processes in the class $\Aclass_{\delta}(\gamvec; \lamvec^*)$, and, therefore, queue $m$ is not RDS.
However, before proceeding to the formal arguments, we overview informally the key ideas in the proof.

\subsection{Outline of the proof}
\label{s:outline-2}
The main idea is to construct a certain arrival process $\Avec(\cdot)$ 
in the class $\Aclass_{\epsilon}(\gamvec; \lamvec^*)$, whose arrival rate is a time-scaled  version of the piecewise constant rate $\lamvec(\cdot)$ associated with the $\epsilon$-JF($\nvec$) trajectory. We then use the bounded sensitivity property of the MW dynamics (Theorem~\ref{th:sensitivity}) to
show that 
the resulting process $\Qvec(\cdot)$ tracks a suitably scaled (by a factor of $T$) version of 
 the $\epsilon$-JF trajectory $\qvec^{\epsilon}(\cdot)$, with substantial probability.
 In particular, 
$\Qvec(T)$ will be (with substantial probability) comparable to $ T\qvec^{\epsilon}(1)$,  leading to a large value of $\Exp\big[Q_m(T)\big]$. 
This part of the argument capitalizes on 
the fact that the number of jumps of the $\epsilon$-JF($\nvec$) trajectory is limited by the condition $\gamvec^T\nvec\leq 1$.

On the technical side, the tracking result involves two separate arguments:
\begin{itemize}
\item[(a)]
Whenever the $\epsilon$-JF($\nvec$) trajectory has a jump, at some time $\tau$,  
there is substantial probability that the stochastic process also has a jump at some time near the scaled counterpart,  $T\tau$, of $\tau$; 
see Lemma~\ref{lem:bound episode jump}. 
\item[(b)] In between jump times of the $\epsilon$-JF($\nvec$) trajectory, we use concentration inequalities to show that there is a fairly large probability that the stochastic process stays close to its fluid counterpart.
\end{itemize}

\subsection{Jumps  of the $\epsilon$-JF($\nvec$) trajectory}
We fix  a network with $\nn$ nodes,  a set $\S$ of possible service vectors, an
arrival rate vector $\lamvec^* \succeq {\bf 0}$, a particular queue, $m$, of interest, 
and a vector  $\gamvec$ of tail exponents with components in $(0,\infty]$.
We also fix some $\epsilon>0$ and assume that  the $\epsilon$-JF($\gamvec$) condition fails to hold.

We start with a few elementary observations, namely, that the $\epsilon$-JF($\nvec$) trajectories of interest can be taken, without loss of generality, through scaling and perturbations, to have some convenient properties that allow us to simplify subsequent notation and arguments. The proof is given in Appendix~\ref{ap:convenient}.

\begin{lemma}\label{l:convenient}
If the $\epsilon$-JF($\gamvec$) condition fails to hold, then
there exists some $\nvec\succeq {\bf 0}$ with   $\gamvec^T\nvec\leq 1$, and an 
$\epsilon$-JF($\nvec$) trajectory $\qvec^{\epsilon}(\cdot) $ with the following properties:
\begin{itemize}
\item[(a)] $q_m^{\epsilon}(1)>0$;
\item[(b)] The times at which $\qvec^{\epsilon}(\cdot) $ is discontinuous (the ``jump times'') all belong to the open interval $(0,1)$.
\item[(c)] At each jump time, exactly one of the components of $\qvec^{\epsilon}(\cdot) $ is discontinuous.
\item[(d)] The arrival rate associated to  $\qvec^{\epsilon}(\cdot) $ satisfies $\inf_t \lambda_j(t)>0$, for all $j$.
\end{itemize}
\end{lemma}

For the rest of the proof, we fix an $\epsilon$-JF($\nvec$) trajectory with the properties in Lemma~\ref{l:convenient}, together with the associated vector $\nvec$ and rate function $\lamvec(\cdot)$. We define
$n=n_1+\cdots+n_\nn$, which is  the total number of jumps.\footnote{We note that $n$ may be equal to zero. For example, for a single unstable queue, an $\epsilon$-JF trajectory becomes positive, even in the absence of jumps. Such a system is not RDS, consistent with our result.}

We define $\Theta_0=0$, $\Theta_{n+1}=1$, and for $k=1,\ldots,n$, let $\Theta_k$ be the $k$th jump time. In particular,
\begin{equation}\label{eq:Tets}
	0=\Theta_0<\Theta_1<\cdots <\Theta_n<\Theta_{n+1}=1.
\end{equation}
We will use $j_k$ to refer to the queue at which the $k$th jump occurs, and $a_k$ to refer to the size of the $k$th jump. In particular,
at time $\Theta_k$, for $k=1,\ldots,n$, $q_j(\cdot)$ is continuous for every $j\neq j_k$, and
$$
q^{\epsilon}_{j_k}(\Theta_k) = q^{\epsilon}_{j_k}(\Theta_k^-) +a_k.$$

\subsection{Defining certain constants}\label{s:constants}

As in \eqref{eq:mumax},
we define
\begin{equation} \label{eq:mumax2}
\mumax = 1+\max_{\muvec\in\S} \|\muvec\|+ \|\lamvec^*\|+\epsilon.
\end{equation}
Moreover, similar to \eqref{def:gammast}, we let
\begin{equation} \label{def:gammast 2}
	\gammast = \min_{j=1,\ldots,\nn} \gamma_j >0.
\end{equation}
By arguing as in Claim~\ref{cl:gamma-small}, we can and will assume, without loss of generality, that 
$\gamma\leq 1$. 
We then define a positive constant 
\begin{equation}\label{eq:def delta =q eps 1}
	c=q^\epsilon_m(1), 
\end{equation}
and also let 
\begin{equation}\label{eq:delta'}
	\dd  \,\equals \,  \frac12\,\min\left\{\frac{\gamma c}{4(1+4\mumax)},\quad \min_{\kappa=0,\ldots,n}\big\{\Tet_{\kappa+1}-\Tet_\kappa\big\},\quad \frac{\min_{\kappa=1,\ldots,n} a_\kappa}{1+2\mumax} \right\}.
\end{equation}
(In case $n=0$, the last term inside the brackets, is taken to be zero.)

\subsection{Defining the stochastic arrivals}\label{s:def-stoch-arr} 

Let us fix some constants $t_0\geq 0$ and $T> 0$, and keep them fixed until the end of the proof, in Section~\ref{s:beyond}. 
In this subsection, we define a stochastic arrival process  over an interval of the form $[t_0,t_0+T)$, thus constructing what we call an \emph{episode} of the overall process. Later, in Section~\ref{s:beyond}, we concatenate multiple episodes, to construct the stochastic process over the entire timeline $[0,\infty)$.

Consider the arrival rate function $\lamvec(\cdot)$ of the
$\epsilon$-JF($\nvec$) trajectory; in particular,  
$\| \lamvec(t)-\lamvec^* \|\leq\epsilon$. 
The arrival rate vector for the stochastic process during the episode is 
a time-scaled (by a factor of $T$) and shifted (by $t_0$) 
version of $\lamvec(\cdot)$. 
More concretely, we let
\begin{equation}\label{eq:lbar}
\bar{\lamvec}(t) =\lamvec\Big(\frac{t-t_0}{T}\Big),\qquad t\in[t_0,t_0+T).
\end{equation}
Clearly,
\begin{equation} \label{eq:lldiff}
	\big\|\bar{\lamvec}(t)-\lamvec^*\big\|\, \le\,\epsilon,\qquad \forall\ t\in [t_0,t_0+T). \end{equation}
Furthermore, because of Lemma \ref{l:convenient}(d), 
$\inf_t \bar{\lambda}_j(t)>0$, for every $j$. 

We now digress to introduce certain constants that will be used to specify the exact form of the distribution of the arrival processes. 
For any $\alpha>0$, we let
\begin{equation} \label{eqa:def sigma}
	\sigma(\alpha) \,\equals\, \int_{\mumax}^{\infty} x^{-(1+\alpha)}\log(x+1)\,dx,
\end{equation}
where $\mumax$ was defined in \eqref{eq:mumax2}. 
Then, for any $\alpha>0$, we have
\begin{equation}\label{eq:sig/sig}
	\frac{\sigma(\alpha)}{\sigma(1+\alpha)} \,=\, \frac{\int_{\mumax}^{\infty}x\cdot x^{-(2+\alpha)}\log(x+1)\,dx}{\int_{\mumax}^{\infty} x^{-(2+\alpha)}\log(x+1)\,dx}
	\,\ge\,\frac{\mumax\,\int_{\mumax}^{\infty} x^{-(2+\alpha)}\log(x+1)\,dx}{\int_{\mumax}^{\infty} x^{-(2+\alpha)}\log(x+1)\,dx} \,=\, \mumax.
\end{equation}

\begin{definition}[Arrivals during an episode] \label{def:episode}
	Given $T>0$ and $t_0\ge0$, 
	we define the arrival process over the interval $t\in[t_0,t_0+T)$ (which we call an \emph{episode}) as follows:
	\begin{itemize}
	\item[(a)] If $\gamma_j=\infty$, then $A_j(t)=\bar\lambda_j(t)$ (deterministically).
	\item[(b)] If $\gamma_j<\infty$, then $A_j(t)$ is a random variable with probability density function
		\begin{equation}\label{eq:def pdf A}
		f_{A_j(t)}(x) \,= \,  \frac{{\bar\lambda_j}(t)}{\sigma(\gamma_j)}\cdot x^{-(2+\gamma_j)}\log(x+1)\,\one\left(x\ge\mumax\right)  \, + \,  \left(1- \frac{{\bar\lambda_j(t)}\sigma(1+\gamma_j)}{\sigma(\gamma_j) }  \right) \delta(x),
	\end{equation} 
where	$\one(\cdot)$ denotes the indicator function, $\delta(\cdot)$ is Dirac's delta function, and $\sigma(\cdot)$ is as defined in \eqref{eqa:def sigma}.
	\item[(c)] The $A_j(t)$, for different $j$ and $t$, are independent.
	\end{itemize}
We refer to $\big\{\Avec(t): t\in[t_0,t_0+T)\big\}$ as an \emph{episode-adjusted} arrival process.
\end{definition}

From \eqref{eq:sig/sig}, we obtain $\sigma(\gamma_j)/\sigma(1+\gamma_j) \ge \mumax\ge \bar\lambda_j(t)$, where the last inequality is due to \eqref{eq:mumax2} and \eqref{eq:lldiff}. 
Therefore,  the coefficient of  the delta function is nonnegative, and we have a well-defined distribution. 
It is also easy to verify that 
$\int_0^\infty f_{A_j(t)}(x)\,dx=1$.
Furthermore, from~\eqref{eq:lldiff}, $\bar{\lambda_j}(t)$ is bounded above. It follows that 
each $A_j(t)$ is dominated by a random variable $\bar{A}_j$ with tail exponent $\gamma_j$, and consequently the process also has tail exponent $\gamma_j$. 
Moreover, $\E{A_j(t)} = \int_0^\infty x\, f_{A_j(t)}(x)\,dx=\bar{\lambda}_j(t)$. 
Therefore, using again~\eqref{eq:lldiff}, we have $\big\|\E{\Avec(t)}-\lamvec^*\big\|\, \le \,\epsilon$, for all $t\in[t_0,t_0+T)$. In particular, the process $\Avec(\cdot)$ belongs to the class $\Aclass_{\epsilon}(\gamvec; \lamvec^*)$, as desired.

\subsection{Probabilistic analysis}\label{s:prob-converse}
We aim to show that during an episode, and with significant probability, the queue vector process $\Qvec(\cdot)$ stays close to a scaled version of the $\epsilon$-JF trajectory, i.e., that 
$$\Qvec(t) \approx T\qvec^{\epsilon}\Big(\frac{t-t_0}{T}\Big),\qquad\forall\  t\in[t_0,t_0+T).$$
 To accomplish this 
we consider each interval of the form $[t_0+\Tet_\kappa T,t_0+\Tet_{\kappa+1} T)$, for $k=0,1,\ldots,n$,  and show that there is a substantial probability that the arrival process we have defined has the following properties: 
(a)  as long as $k>0$, it has a jump in a small segment in the beginning of the interval, and 
(b)  it has small fluctuations in the rest of the interval.
Below, we define events that capture the above two properties.

For $\kappa=1,\ldots,n$, let $\Bvec_\kappa$ be the cumulative arrival vector over the interval $[t_0+\Tet_\kappa T, t_0+\Tet_\kappa T+\dd T )$, i.e.,\footnote{To avoid notation clutter, we present the proof as if $\Theta_k T$ or $\Theta_k T+dT$ were  integer, which is not necessarily the case. Everything goes through, with occasional trivial modifications, if a sum of the form $\sum_{t=a}^{b} c_t$ is interpreted as $\sum_{t=\lfloor a \rfloor}^{\lfloor b \rfloor} c_t$.
}
\begin{equation}\label{eq:B-cumul}
	\Bvec_\kappa \,=\, \sum_{t = t_0+ \Tet_\kappa T}^{t_0+\Tet_\kappa T +\dd T -1} \Avec(t).
\end{equation}
We let $\event_\kappa^{\textrm{jump}}$ be the event that $\Bvec_\kappa$ emulates the jump in $\qvec^{\epsilon}(\cdot)$ at time $\Theta_\kappa$,
scaled by $T$, i.e.,
\begin{equation}\label{eq:E-jump}
	\event_\kappa^{\textrm{jump}} = \Big\{  \big\| \Bvec_\kappa - T a_\kappa \evec_{j_\kappa}  \big\|\,\le\, \dd T (1+2\mumax)\Big\},
\end{equation}
where 
$j_\kappa$ has been defined as the index of the queue at which the $\kappa$th jump takes place,  
and $\dd$ is the constant defined in~\eqref{eq:delta'}.

\begin{lemma}\label{lem:bound episode jump}
	There exist $\psi \in(0,1)$ 
	and $\bar{T}_1>0$, such that if $T\ge \bar{T_1}$, then, for $\kappa=1,\ldots,n$, and for the episode-adjusted arrival process over an episode $[t_0,t_0+T)$, we have  \ $\Pr(\event_\kappa^{\textrm{jump}})\ge \psi T^{-\gamma_{j_\kappa}} \log T$.
\end{lemma}
The proof is given in Appendix~\ref{appsec:proof lem bound episode jump}.

Recall now that the function $\lamvec(\cdot)$ driving the trajectory $\qvec^{\epsilon}(\cdot)$ is piecewise constant, with a finite number of pieces. We use $r$ to denote the number of such pieces. We then proceed to define 
certain ``small fluctuations'' events. 
For $\kappa=0,1,\ldots,n$, we let $\event_\kappa^{\textrm{fluc}}$ be the event that the cumulative fluctuations of $\Avec(\cdot)$ over the interval $[t_0+\Tet_\kappa T+\dd T , t_0+\Tet_{\kappa+1} T )$ are small, i.e.,
\begin{equation}\label{eq:fluct-kappa}
	\event_\kappa^{\textrm{fluc}} = \Big\{
	 \Big\| \sum_{\tau=t_0+\Tet_\kappa T+\dd T}^{t}\big(\Avec(\tau)-\bar{\lamvec}(\tau)\big)   \Big\|\,\le\ \frac{\gamma c T}{32C\rr},\quad
	 \forall\ t\in 
	 [t_0+\Tet_\kappa T+\dd T,t_0+\Tet_{\kappa+1}T)
	\Big\}.
\end{equation}

\begin{lemma}\label{lem:bound episode fluc} 
	There exists a $\bar{T}_2\ge 0$ such that if $T\ge \bar{T}_2$, then
	for $\kappa=0,1,\ldots,n$, and for the episode-adjusted process over an   episode $[t_0,t_0+T)$,
	 we have
 $\Pr(\event_\kappa^{\textrm{fluc}})\ge 1/2$.
\end{lemma}
The proof is given in Appendix~\ref{appsec:proof lem bound episode fluc}.

Note that each one of the events $\event_\kappa^{\textrm{jump}}$, for $k=1,\ldots,n$,  and $\event_\kappa^{\textrm{fluc}}$, for  $k=0,\ldots,n$,  is determined by the arrival process during a particular interval, and that 
all of these intervals are disjoint. Thus, the independence assumption on the arrival process implies that all of these events are independent.
Therefore, if $T\ge \max\{\bar{T}_1,\bar{T}_2\}$, then
\begin{equation} \label{eq:prob all event jump fluc}
	\begin{split}
		\Pr\Big(\event_{1}^{\textrm{jump}},\ldots, \event_{n}^{\textrm{jump}}, \event_{0}^{\textrm{fluc}},\ldots, \event_{n}^{\textrm{fluc}}    \Big)   
		\,&=\, \prod_{\kappa=1}^n \Pr(\event_\kappa^{\textrm{jump}}) \,\cdot\, \prod_{\kappa=0}^n \Pr(\event_\kappa^{\textrm{fluc}})\\
		&\geq \prod_{\kappa=1}^n \psi T^{-\gamma_{j_\kappa}} \log T \,\cdot\, \prod_{\kappa=0}^n \frac12\\
		\,&=\, \frac{1}{2}\,\left(\frac{\psi\log T}2\right)^n \, T^{-\sum_{\kappa=1}^{n}\gamma_{j_\kappa}}\\
		\,&=\,  \frac{1}{2}\, \left(\frac{\psi\log T}2\right)^n \, T^{-\gamvec^T\nvec}\\
		\,&\geq\, \frac{1}{2}\,
		\left(\frac{\psi}2\right)^n \, T^{-\gamvec^T\nvec}\,\log T\\
		\,&\geq\, 
		\frac{1}{2}\, \left(\frac{\psi}2\right)^{1/\gamma}
		 \, T^{-\gamvec^T\nvec}\,\log T\\
		\,&\geq\, 
		\frac{1}{2}\, \left(\frac{\psi}2\right)^{1/\gamma}
		\, T^{-1}\log T,
	\end{split}
\end{equation}
where the first inequality follows from Lemmas~\ref{lem:bound episode jump} and~\ref{lem:bound episode fluc}.  
The second inequality is because $n\ge 1$ and (without loss of generality) $\log T\geq 1$.
The third inequality is due to
$n\gamma = \sum_{j=1}^{\nn}n_j\gamma \le \sum_{j=1}^{\nn}n_j\gamma_{j}\le 1$, 
so that $n\leq 1/\gamma$, together with the fact $\psi\leq 1$ (see Lemma~\ref{lem:bound episode jump}).
The last inequality is again because
 $\gamvec^T\nvec\le 1$. 

\subsection{$\Exp[Q_m]$ is large at the end of an episode}\label{s:Qlarge}

We are now ready to argue that if the 
events $\event_1^{\textrm{jump}},\ldots,\event_n^{\textrm{jump}}$ and $\event_{0}^{\textrm{fluc}},\ldots,\event_n^{\textrm{fluc}}$ occur, then, over an episode $[t_0,t_0+T)$, the  process $\Qvec(\cdot)$ stays close  to the suitably scaled 
$\epsilon$-JF trajectory. 
As a consequence, the value of $Q_m$ at the end of the episode becomes of order $cT$, where $c=q^{\epsilon}_m(1)>0$, as in \eqref{eq:def delta =q eps 1}.  
This argument is
entirely deterministic. It is carried out in the course of the proof  of the next lemma (in Appendix~\ref{appsec:proof lem episode deterministic}),   
and relies  on the sensitivity bound in Theorem~\ref{th:sensitivity}.

\begin{lemma}\label{lem:bound episode deterministic}
There exists a
$\bar{T}_3\ge 0$ such that if $T\ge \bar{T}_3$, then  the following holds.	
If a
 sample path of the episode-adjusted process over the episode $[t_0,t_0+T)$ satisfies the events 
  $\event_1^{\textrm{jump}},\ldots,\event_n^{\textrm{jump}}$ and $\event_{0}^{\textrm{fluc}},\ldots,\event_n^{\textrm{fluc}}$, 
 and if $\big\| \Qvec(t_0) \big\| \le c T/5$,  then, $Q_m(t_0+T)\ge {c T }/{2}$.
\end{lemma}

Let $\rho=\left({\psi}/2\right)^{1/\gamma}\big /4 $
 and  $\bar{T}\equals \max\{2,\bar{T}_1, \bar{T}_2, \bar{T}_3\}$, where $\bar{T}_1$, $\bar{T}_2$, and $\bar T_3$ are the constants in Lemmas~\ref{lem:bound episode jump},~\ref{lem:bound episode fluc}, and~\ref{lem:bound episode deterministic}, respectively. We note that all of these constants, $\bar{T}_1$, $\bar{T}_2$, $\bar{T}_3$, and therefore $\bar{T}$ as well, are defined in terms of general parameters and properties of the particular $\epsilon$-JF trajectory, and are deterministic. 
Lemma~\ref{lem:bound episode deterministic} and \eqref{eq:prob all event jump fluc} imply that  for an episode  $(t_0,t_0+T)$,  with $T\ge\bar{T}$, and initialized so that $\| \Qvec(t_0) \|\le c T/5$, we have
\begin{equation}\label{eq:pr Q_m 1/T}
	\Pr\left( Q_m(t_0+T)\ge \frac{c T}2 \right) 
	\,\ge  \,  \Pr\Big(\event_{1}^{\textrm{jump}},\ldots, \event_{n}^{\textrm{jump}}, \event_{0}^{\textrm{fluc}},\ldots, \event_{n}^{\textrm{fluc}}    \Big) 
	\,\ge\, 2\rho \, T^{-1}\log T,
\end{equation}
which implies that
\begin{equation}\label{eq:at the end of an episode}
	\Exp[Q_m(t_0+T)]\,\ge\, \frac{c T }{2} \,\Pr\left(Q_m(t_0+T)\ge \frac{c T }{2}\right) \, \ge\, \frac{c T }{2}\cdot 2\rho T^{-1}\log T \, =\, 
	{\rho c \log T}.
\end{equation}


\subsection{Concatenating episodes, over the entire timeline}\label{s:beyond}
So far, we have defined and studied an arrival process over an episode $[t_0,t_0+T)$.  We now concatenate a sequence of such episodes, of increasing duration, 
which defines an arrival process over an infinite timeline. 

We define times $T_0,T_1,T_2,\ldots,$ 
and arrival processes for the intervals $[T_i,T_{i+1})$, recursively,
as follows. We let $T_0=0$ and $T_1=\bar{T}$, where $\bar T$ was defined in the last paragraph  of Section \ref{s:Qlarge}. We also let  $\Avec(\cdot)$, for $t\in[T_0,T_1)$ be the corresponding episode-adjusted process, as in Definition~\ref{def:episode}.
Suppose now that we have defined $T_i$, for some $i\geq 1$, as well as the arrival process for $t\in[0,T_i)$. We then let
\begin{equation}\label{eq:def Ti inductive}
	T_{i+1} = T_i + \max\left\{T_i, \frac{10 \,\Exp\big[\|\Qvec(T_i)\|\big]}{c}\right\}. 
\end{equation}
Finally, we define the arrival process over the episode $[T_i,T_{i+1})$ to be the corresponding episode-adjusted process. With this recursion, the arrival process is now well-defined for all times $t\geq 0$. 

Note that $\Exp\big[\|\Qvec(T_i)\|\big]\le \Exp\left[\sum_{t=0}^{T_i-1}\|\Avec(t)\|\right]= \sum_{t=0}^{T_i-1}\Exp\big[\|\Avec(t)\|\big]  
\le  T_i\nn(\lamax) <\infty$. 
This guarantees that all $T_i$ are finite, and that we have an infinite number of episodes. 
Moreover, note that for $i=1,2,\ldots$, we have $T_{i+1}\ge 2T_i$. As a result,  $T_{i}\ge 2^{i-1} \bar{T}$ and also,
\begin{equation}\label{eq:log T-T}
	\log (T_{i+1}-T_i) \,\ge\, \log T_i  \,\ge\,  i \log2,
\end{equation}
where the last inequality is because $T_i\geq 2^{i-1}\bar{T}\geq 2^{i}$.

We define the event 
$$\Bev_i=\Big\{ 
\|\Qvec(T_i)\| \le \frac{(T_{i+1}-T_{i})c}{5}
\Big\},$$ and use the
 Markov inequality, to obtain
\begin{equation}\label{eq:onehalf}
 \Pr(\Bev_i)
\, \ge\, 1- \frac{\Exp\big[\|\Qvec(T_i)\|\big]}{(T_{i+1}-T_{i})c/{5}}
	\,\ge\, 1- \frac{\Exp\big[\|\Qvec(T_i)\|\big]}{\big(10\Exp\big[\|\Qvec(T_i)\|\big]\big)/5} \,=\,\frac12,
\end{equation}
where the second inequality is due to \eqref{eq:def Ti inductive}.

Recall now that the inequality~\eqref{eq:at the end of an episode}, which was about an episode of length $T$, made use of the assumption 
$\| \Qvec(t_0) \|\le c T/5$, where $t_0$ is the start time of the episode. 
According to~\eqref{eq:onehalf}, this assumption is satisfied at the start time of the episode $[T_i,T_{i-1})$, with probability at least 1/2.
By  interpreting~\eqref{eq:at the end of an episode} as a statement about conditional expectations, and with $t_0$ and $t_0+T$ replaced by $T_i$ and $T_{i+1}$, respectively,
we obtain
\begin{equation*}
\Exp\big[Q_m(T_{i+1})\big] \, \geq\, \Pr(\Bev_i) \cdot \Exp\big[Q_m(T_{i+1}) \mid \Bev_i\big]
\, \geq\, \frac{1}{2}\cdot \rho c \log (T_{i+1}-T_i)
\, \geq\,  \frac{\rho c \,i \log2 }{2},
\end{equation*}
where the last inequality follows from \eqref{eq:log T-T}.
Therefore,  $\Exp\big[Q_m(T_i)\big]$ grows unbounded as $i$ increases. Consequently, 
under the arrival process that we constructed, queue $m$ is not delay stable.
This conclusion is obtained for any positive choice of $\epsilon$, no matter how small, and establishes that 
queue $m$ is not RDS.  This completes the proof of the second direction of Theorem~\ref{th:main}.


\medskip
\medskip

\section{Proof of Theorem~\ref{th:lyap}} \label{sec:proof lyap}

\subsection{Proof of the  first direction} Let us fix some $\epsilon>0$. 
To establish one direction of the result, we assume that the $\epsilon$-JF$(\gamvec)$ condition holds for every  
$\gamvec\in\Gamma$. We will show that there exists a special $\epsilon$-Lyapunov function.

Let $\N$ be the set of all nonnegative integer vectors $\nvec$ such that $n_j=0$ for $j>h$; that is, we allow arbitrarily many jumps at the heavy-tailed queues and no jumps at the light-tailed ones. 
As in
Definition~\ref{def:W}, 
for any nonnegative integer vector $\nvec$, let $W(\nvec)$ be the set of all points in $\R_+^\nn$ that are reachable by $\epsilon$-JF$(\nvec)$ trajectories. 
Let $W = \bigcup_{\nvec\in\N} W(\nvec)$, and 
consider a Lyapunov function $V(\cdot)$ equal to the distance from $W$, i.e., $V(\xvec)\equals d\big(\xvec,W\big)$, for any $\xvec\in\R_+^{\nn}$. 
We will show that this Lyapunov function  has the desired properties.

The distance function is clearly Lipschitz continuous, with a Lipschitz constant equal to 1, which implies the first property in the definition of special $\epsilon$-Lyapunov functions. 

For the second property, Lemma~\ref{lem:absorb}(b) applies and shows that each set $W(\nvec)$ is $\epsilon$-invariant. It can be seen that the union, $W$, of the $\epsilon$-invariant sets $W(\nvec)$ is also $\epsilon$-invariant. It then follows from  Lemma~\ref{lem:absorb}(a) 
 that $W$ is $\epsilon$-attracting. This proves the second property in Definition~\ref{def:nice lyap}.

Note that every $\nvec\in\N$ satisfies the inequality $\gamvec^T \nvec\leq 1$ for \emph{some} $\gamvec\in \Gamma$. Since  the $\epsilon$-JF$(\gamvec)$ condition holds for every $\gamvec\in\Gamma$, it follows that every $\epsilon$-JF($\nvec$) trajectory, with $\nvec\in\N$, satisfies $q_m(t)=0$, for all $t$. Hence   $q_m=0$, for all $\qvec\in W$.
Furthermore, since ${\bf 0}\in W$, we have $V({\bf 0})=0$. This 
establishes the third property in  Definition~\ref{def:nice lyap}.

Finally, $W$ is closed under  jumps along coordinates associated with heavy-tailed 
arrivals.   Therefore, $V(\cdot)$ is nonincreasing
 along those directions, and the fourth property in Definition~\ref{def:nice lyap} follows.
Thus, $V(\cdot)$ has all the required properties of special $\epsilon$-Lyapunov functions. This completes the proof of one direction of the theorem.

\subsection{Proof of the reverse direction}
We continue with the proof of the reverse direction. 
We fix some $\epsilon>0$ and assume
that there exists a special $\epsilon$-Lyapunov function $V(\cdot)$, and let $W=\big\{\xvec\in\R_+^{\nn} \mid V(\xvec)=0 \big\}$. 
The argument rests on the  $\epsilon$-invariance of $W$ which, in turn, relies on some properties of the MW dynamics that we discuss next.

For any $\lamvec,\xvec\in\R_+^\nn$, let
\begin{equation} 
	\xivec_{\lamvec}(\xvec) \equals \dot{\qvec}(0),
\end{equation}
where $\qvec(\cdot)$ is the fluid trajectory corresponding to arrival rate $\lamvec$ and initialized with $\qvec(0)=\xvec$. In view of \eqref{eq:fluid evolution}, we have $\xivec_{\lamvec}(\xvec)\in\Drift_\lamvec(\xvec)$. Moreover, 
it is shown in Lemma~2(a) of \citep{AlTG19sensitivity} that $\xivec_{\lamvec}(\xvec)$
 has the minimum norm among all vectors in $\Drift_\lamvec\big({\xvec}\big)$, i.e.,
\begin{equation} \label{eq:min norm in S}
	\xivec_{\lamvec}(\xvec)\,=\, \argmin{\nuvec \in \Drift_\lamvec(\xvec)} \|\nuvec\|, \qquad \forall\  \xvec\in\R_+^\nn.
\end{equation}
with the minimizer being unique. 

Given a  closed and convex set $\mathcal{A}\subset \R^{\nn} $ and a point $\xvec\in\R^\nn$, we denote by $\pi_\mathcal{A}(\xvec)$ the projection of $\xvec$ on $\mathcal{A}$, defined as  the point in $\mathcal{A}$ which is closest  to $\xvec$.
With this  terminology, $\xivec_{\lamvec}(\xvec)$
is the projection $\pi_{\Drift_{\lamvec}{(\xvec)}}({\bf 0})$ of the zero vector on the set  $\Drift_{\lamvec}{(\xvec)}$. 

In what follows, we also make use of an elementary  property of projections: if
$\mathcal{A}$ is a closed convex set, $\bvec$ is some vector, and 
$\mathcal{B}=\mathcal{A}+\bvec$, then  
\begin{equation}\label{eq:geo}
\|\pi_\mathcal{A}(\xvec)- \pi_\mathcal{B}(\xvec)\|\,\leq\, \|\bvec\|.
\end{equation}

As a consequence of the above, 
for any $\lamvec_1$, $\lamvec_2$, and $\xvec$ in $\R_+^\nn$,
\begin{equation} \label{eq:drift diff bound}
	\begin{split}
	\big\|\xivec_{\lamvec_1}(\xvec)-\xivec_{\lamvec_2}(\xvec)\big\| 
	&\,=\,\big\|\pi_{\Drift_{\lamvec_1}{(\xvec)}}(\mathbf{0})- \pi_{\Drift_{\lamvec_2}{(\xvec)}}(\mathbf{0})\big\| \\
	&\,=\, \big\|\pi_{\Drift_{\lamvec_1}(\xvec)}(\mathbf{0})- \pi_{\Drift_{\lamvec_1}(\xvec)+\lamvec_2-\lamvec_1}(\mathbf{0}){\big\|}\\
	&\,\le\, \big\|\lamvec_1-\lamvec_2\big\|,
	\end{split}
\end{equation}
where the second equality is because 
$$\Drift_{\lamvec_2}(\xvec) = \lamvec_2 - \Sbar(\xvec) = \lamvec_1 - \Sbar(\xvec) + \lamvec_2-\lamvec_1
 = \Drift_{\lamvec_1}(\xvec) +  \lamvec_2-\lamvec_1,$$
 and the inequality  follows from ~\eqref{eq:geo}.

\begin{lemma} \label{l:W-inv}
The set $W$ is $\epsilon$-invariant.
\end{lemma}
\begin{proof}
Since $V$ is a special $\epsilon$-Lyapunov function, it is Lipschitz continuous with Lipschitz constant $1$.
Let $\lamvec\in\R_+^\nn$ be such that $\|\lamvec-\lamvec^*\|\leq\epsilon$, and consider   fluid trajectories $\qvec(\cdot)$ and $\pvec(\cdot)$ corresponding to arrival rates $\lamvec$ and $\lamvec^*$, respectively, initialized with the same nonnegative vector $\qvec(0)=\pvec(0)\not\in W$. Then,
\begin{equation}
	\begin{split}
		\dot{V}\big(\qvec(t)\big)\big|_{t=0} \,&=\, \lim_{\delta\downarrow0} \frac{V\big(\qvec(\delta)\big)- V\big(\qvec(0)\big)}{\delta}\\
		&\,=\,
		\lim_{\delta\downarrow0} \frac{V\big(\pvec(\delta)\big) + V\big(\qvec(\delta)\big) -V\big(\pvec(\delta)\big)  - V\big(\qvec(0)\big)}{\delta}\\
				&\le\, \lim_{\delta\downarrow0} \frac{\big[V\big(\pvec(\delta)\big) + \| \qvec(\delta) -\pvec(\delta)  \|\big]- V\big(\qvec(0)\big)}{\delta}\\
		&=\, \lim_{\delta\downarrow0}\left( \frac{V\big(\pvec(\delta)\big) - V\big(\pvec(0)\big)}{\delta} \,+\,\frac{\| \qvec(\delta) -\pvec(\delta)\|}{\delta}\right)\\
		&\le\, \lim_{\delta\downarrow0}\left( \frac{V\big(\pvec(\delta)\big) - V\big(\pvec(0)\big)}{\delta} \,+\,\| \lamvec - \lamvec^*\| \right)\\
		&=\,  \dot{V}\big(\pvec(t)\big)\big|_{t=0} \,+\,\| \lamvec - \lamvec^*\|\\
		&\leq\, -\epsilon +\epsilon\\
		&=\, 0,
	\end{split}
\end{equation}
where the first inequality is because $V$ has a Lipschitz constant equal to $1$, the third equality is due to $\qvec(0)=\pvec(0)$, and the second inequality follows from \eqref{eq:drift diff bound}. 
The last inequality follows from the second property of special $\epsilon$-Lyapunov functions in Definition~\ref{def:nice lyap} and 
the assumption $\|\lamvec-\lamvec^*\|\leq\epsilon$.

The above argument shows 
that, for any $\lamvec\in\R_+^\nn$ with $\|\lamvec-\lamvec^*\|\leq\epsilon$, and for any fluid trajectory $\qvec(\cdot)$ corresponding to arrival rate $\lamvec$, the distance from the set $W$ cannot increase. In particular, if $\qvec(\cdot)$  is 
initialized with $\qvec(0)\in W$, it must stay in $W$. 
 Therefore, using the terminology in Definition~\ref{def:absorb}, $W$ is $\epsilon$-invariant. 
\end{proof}

From the third property in the definition of special $\epsilon$-Lyapunov functions, we have $V({\bf 0})=0$ and, therefore, ${\bf 0}\in W$. Thus, every $\epsilon$-JF($\nvec$) trajectory starts in $W$. 
 From the fourth property in the definition of special $\epsilon$-Lyapunov functions, $W$ is  closed with respect to positive jumps along 
the coordinates associated with heavy-tailed arrivals. 
Using also Lemma~\ref{l:W-inv} for the times between jumps, we see that 
for any $\nvec\in\N$, every $\epsilon$-JF($\nvec$) trajectory stays in $W$. 
Equivalently, for any $\gamvec\in\Gamma$, every $\epsilon$-JF($\gamvec$) trajectory stays in $W$.
Finally, employing again the third property of special $\epsilon$-Lyapunov functions, we have $q_m=0$, for all $\qvec\in W$. This implies that $q_m(t)=0$, for  all $\epsilon$-JF$(\gamvec)$ trajectories $\qvec(\cdot)$ with $\gamvec\in \Gamma$, and all $t\ge0$.
This establishes the $\epsilon$-JF$(\gamvec)$ condition for all 
$\gamvec\in \Gamma$, 
 and completes the proof of Theorem~\ref{th:lyap}
 
\subsection{Proof of Corollary~\ref{cor:lyap}}\label{s:cor-pf} 

For part (a), fix some $\epsilon>0$, and suppose that there exists a special $\epsilon$-Lyapunov function.  Theorem~\ref{th:lyap} implies that the $\epsilon$-JF($\gamvec$) condition holds for every $\gamvec\in\Gamma$. It then follows from Theorem~\ref{th:main} that queue $m$ is RDS for every $\gamvec\in\Gamma$.

For part (b), we fix some $\epsilon>0$, and assume that queue $m$ is $\epsilon$-RDS, for all $\gamvec\in\Gamma$. In particular, queue $m$ is RDS and Theorem~\ref{th:main} implies that the $\epsilon'$-JF($\gamvec$) condition holds for some $\epsilon'>0$. However, 
a close inspection of the proof of the reverse part of Theorem~\ref{th:main} reveals that we can in fact choose $\epsilon'$ to be the same as $\epsilon$. Thus, the $\epsilon$-JF($\gamvec$) condition holds and Theorem~\ref{th:lyap} implies that  there exists a special $\epsilon$-Lyapunov function.


\section{Discussion}\label{sec:discuss}
In this section, we summarize some key points and conclude with a few open questions.
 
\subsection{Framing and results}
We  have addressed the problem of delay stability for a class of queueing networks that operate under the Max-Weight scheduling policy, when some arrival processes are heavy-tailed and some are light-tailed. The overall purpose was to develop conditions for delay stability in terms of fluid-like models. However, as  illustrated by the example in Section~\ref{sec:example}, delay instability can be the result of multiple coordinated large jumps. The probabilities of such large  jumps are, in turn, affected 
by the tail exponents of the arrival processes. Given that 
traditional fluid models are oblivious to the tail exponents, we had to introduce  JF (jumping-fluid) models, a  generalization that  allows for jumps along the coordinates associated with heavy-tailed flows, subject to a budget on the number of jumps, with the budget being determined by the tail exponents.
 
At the same time, it became clear that 
 tight conditions for delay stability  
that do not depend on the details of the arrival distributions are only possible under a suitable ``robust'' formulation with respect to both the arrival rates and the arrival process distributions.
With a careful choice of definitions, we were finally able to establish 
necessary and sufficient conditions for robust delay stability, in terms of 
$\epsilon$-JF models.

In Section~\ref{sec:example}, we also discussed a related, so-called ZF condition. The ZF condition essentially examines fluid trajectories that start at zero and involve a single jump, and leads to a necessary condition for delay stability~\citep{MarkMT18}, but the question whether it can also form the basis for a sufficient condition was  open. Our results show that in order to obtain necessary \emph{and\/} sufficient conditions, we need to examine a richer set of trajectories, that involve multiple jumps.

Finally, 
earlier works \citep{Mark13, MarkMT18} had shown that Lyapunov functions with certain structural properties could yield sufficient conditions for  delay stability. 
But it was not clear if and when delay stability is equivalent to the existence of such Lyapunov functions. Our results 
make progress towards establishing the completeness of such a Lyapunov-based methodology, for the regime where the heavy-tailed flows can have arbitrarily small tail exponents. 

Our RJF condition 
 is difficult to test for general networks. In some sense, this  reflects the intrinsic complexity of the (robust) delay stability problem. The Lyapunov-based condition also appears to be hard to test, for general networks.

\subsection{Alternative formulations}
Given the complexity of the RJF condition, it is natural to inquire about simpler alternatives. 
  For example, is it possible to obtain tight delay stability results (without robustness) if we consider 
JF models with a constant rate $\lamvec(\cdot)$? In the same spirit, could we restrict to the case where all jumps in  $\epsilon$-JF models take place at the same time? Might it be easier to consider concrete arrival processes, instead of focusing on delay stability for all arrival proceses with given exponents?

For all three of the above questions, the answer is negative. We discuss such variations and related (counter)examples in  
 Appendix~\ref{app:new}.

\subsection{Open problems}
We collect here a few open problems and possible future research directions.
\begin{itemize}
\item[(a)] Can the results be generalized  to other scheduling policies, e.g., MW-$\alpha$ policies \citep{Neel10},  an extension of the MW policy considered in this paper, or more generally, to other stochastic networks whose stability has been studied using fluid models? One obstacle here is that 
our main result relies heavily on a particular fluctuation bound, which has been established specifically for the MW dynamics  \citep{AlTG19ssc}.   
However, progress  may be possible if we rely on alternative stochastic bounding techniques. 
\item[(b)] Can we identify some special classes of networks for which our criteria (either the RJF condition or the Luyapunov-based condition) can be tested in polynomial time?
\item[(c)] Is the RJF condition, which involves time-varying $\lamvec(t)$, with $\lamvec(t)\approx \lamvec^*$, equivalent to a similar condition in which we only consider $\epsilon$-JF trajectories
with time-invariant $\lamvec(t)=\lamvec$, with $\lamvec\approx \lamvec^*$? See Appendix~\ref{app:new} for further discussion and some conjectures.

\end{itemize}

\section*{Acknowledgment}
We are grateful to Dr. Bert Zwart for providing us with pointers to relevant works in the literature and for informative technical discussions.

\ifIEEE
\bibliographystyle{ieeetr}
\fi
\ifSIAM
\bibliographystyle{siamplain}
\else 

\bibliographystyle{informs2014}

\fi

\bibliography{schbib}

\ifIEEE
\newpage
\appendices
\fi

\ifSIAM
\appendix
\medskip

\else 
\newpage
{\huge Appendices}
\begin{APPENDICES}
	
\hspace{3.5cm}

\fi

\section{Exploring alternative formulations} \label{app:new}
Our formulation involves rate-robustness (robustness with respect to variations in the arrival rate), as well as distributional robustness (by considering the worst-case over all distributions with given tail exponents). It would have been preferable to develop conditions that characterize delay stability for specific systems (with fixed arrival rates and arrival distributions). However, this seems to be impossible, for reasons that will become clearer in this appendix. In particular, 
the distributional robustness aspect appears to be inevitable, as
as long as we are aiming at conditions that are both necessary and sufficient; see Section~\ref{app:details}. For this reason, most of this appendix is devoted to exploring variants of rate robustness. In the interest of brevity, we keep the discussion informal, without rigorous proofs.

\subsection{Variations of our definitions}\label{app:variations}
In this subsection, we present a number of variations to our definitions of RDS and of the RJF condition. In later subsections, we will elaborate on their relations. 
Throughout this appendix,  we assume that the tail exponent vector $\gamvec$ has been fixed, with every $\gamma_j$ in $(0,\infty]$. We also fix some $\lamvec^*>0$.
The various definitions that we offer   differ only with respect to the choice of allowed functions 
$\lamvec(\cdot)$.

Let $\F$ be a class of  discrete-time functions $\lamvec(\cdot)$.
We say that queue $m$ is $\F$-RDS if it obeys Definition~\ref{def:RDS}, except that 
the allowed arrival rates $\E{\Avec(t)}$ are also required to belong to $\F$. 

We consider the following choices for $\F$, leading to three alternative definitions of robust delay stability, namely G-RDS, C-RDS, and 0-RDS:
\begin{itemize}
\item[G:] (General) Here we impose no additional restrictions on $\E{\Avec(t)}$. Thus,  G-RDS is identical to the RDS condition that we have studied.
\item[C:] (Constant) Here we require $\E{\Avec(t)}$ to be constant. Effectively, we are considering small but constant perturbations of $\lamvec^*$. 
\item[0:] (Zero) Here, we require  $\E{\Avec(t)}$ to be equal to $\lamvec^*$, for all times $t$.
\end{itemize}

We continue similarly, to define variants of the RJF condition. Let $\F$ be a class of  continuous-time functions $\lamvec(\cdot)$. The $\F$-RJF condition is defined exactly as in Definition~\ref{def:eps jf cond}, except that we only consider $\epsilon$-JF trajectories for which the rate function $\lamvec(\cdot)$ is also  required to belong to the class $\F$. We consider four possible choices for $\F$, leading to four variants of the RJF condition, namely, UC-RJF, PC-RJF, C-RJF, 0-RJF:
\begin{itemize}
\item[UC:] (Uniformly continuous) Here we remove the requirement in Definition~\ref{def:eJF} that $\lamvec(\cdot)$ be piecewise constant.
Instead, we require $\lamvec(\cdot)$ to be (i) piecewise continuous, with a finite number of discontinuities, and (ii) uniformly continuous on any interval in which it is continuous.
\item[PC:] (Piecewise constant) Here, $\lamvec(\cdot)$ is exactly as in Definition~\ref{def:eJF}, and in particular, piecewise constant. Thus, the PC-RJF condition coincides with the RJF condition we have been studying.
\item[C:] (Constant) Here we require $\lamvec(\cdot)$ to be constant. Effectively, we are considering small but constant perturbations of $\lamvec^*$. 
\item[0:] (Zero) Here, we require $\lamvec(\cdot)$ to be equal to $\lamvec^*$, for all times $t$. 
\end{itemize}

\subsection{Relations between alternative definitions}\label{app:relations}
In this section, we explore the relation between  $\F$-RDS and
$\F$-RJF 
 conditions for different choices of F; see Figure~\ref{f:diagram} for a visual summary.

\begin{figure} 
	\begin{center}
		{
			\ifSIAM
			\includegraphics[width = .95\textwidth]{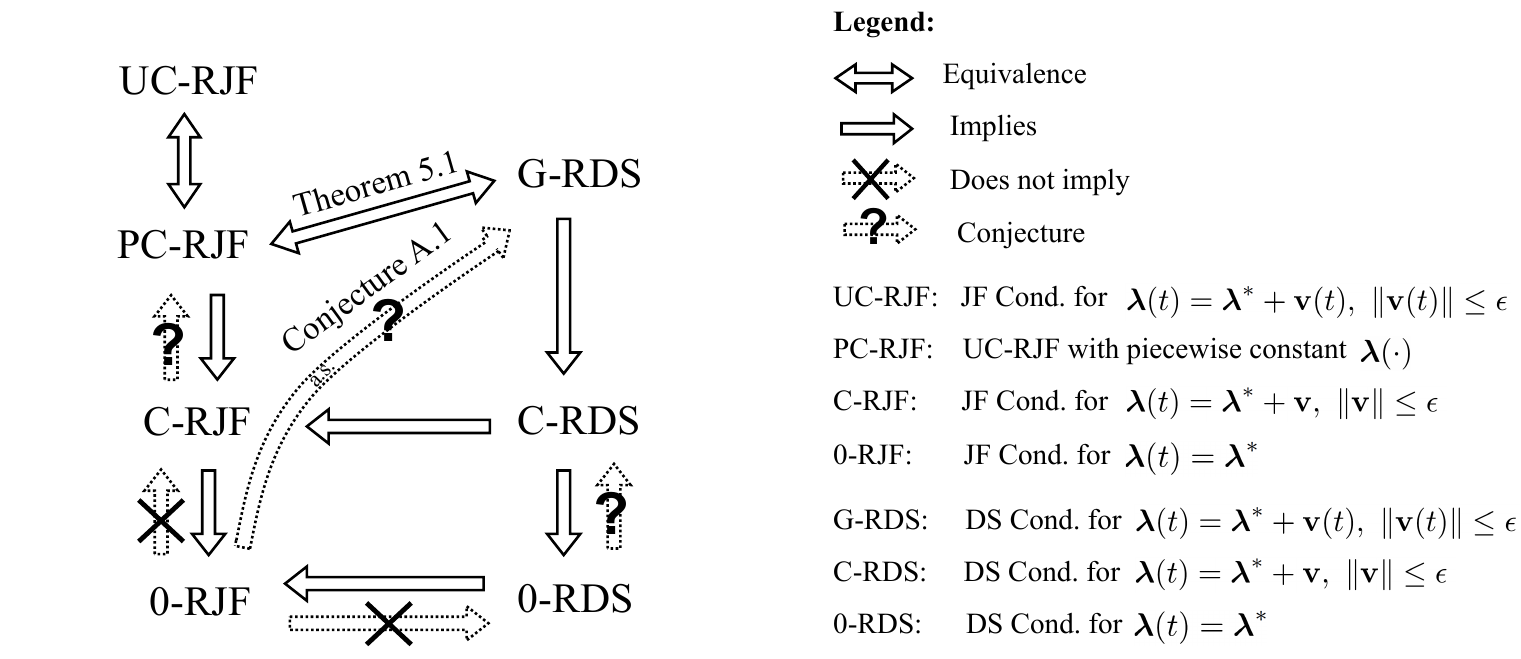}
			\else 
			\includegraphics[width = .95\textwidth]{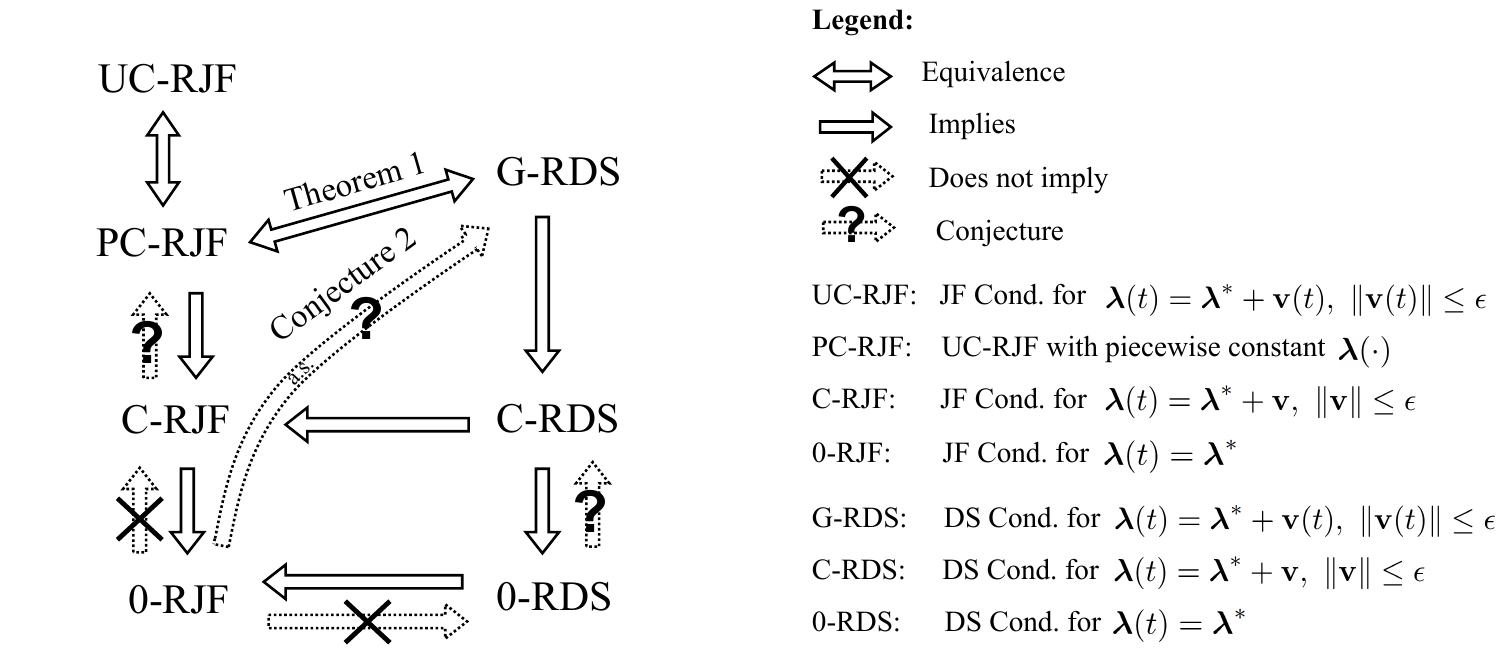}
			\fi}
		\vspace{0cm}
	\end{center}
	\caption{Relation between the various  conditions.} 
	\label{f:diagram}
\end{figure}

It is clear that when we restrict to a smaller class, the RDS or RJF conditions are easier to satisfy. Thus, 
$$\mbox{G-RDS} \implies \mbox{C-RDS} \implies \mbox{0-RDS},$$
and
$$ \mbox{UC-RJF} \implies \mbox{PC-RJF} \implies \mbox{C-RJF} \implies \mbox{0-RJF}.$$ Furthermore, Theorem~\ref{th:main} has established that G-RDS is equivalent to PC-RJF.

\vspace{5pt}
\noindent
{{\bf PC-RJF $\implies$ UC-RJF.}}\,\,
An arbitrary (continuous-time) function $\lamvec(\cdot)$ in the class UC can be approximated by a piecewise constant function with finitely many pieces, uniformly over a compact set. Furthermore, it can be shown that if we perturb by $\epsilon$ the vector $\lamvec(\cdot)$ that drives a JF-trajectory, the resulting trajectory is perturbed by at most $\epsilon$ over a time interval of length 1.  It follows that if the UC-RJF condition fails, we can construct piecewise-constant approximations of $\lamvec(\cdot)$ that demonstrate that the PC-RJF condition also fails. Therefore,  PC-RJF$\implies$UC-RJF. 

Taking Theorem~\ref{th:main} also into account, we see that all three conditions, G-RDS, PC-RJF, and UC-RJF, are equivalent. An alternative path to the same conclusion consists of modifying the proof in Section~\ref{sec:proof RDS=>JF}, and showing  that G-RDS$\implies$UC-RJF. This is possible, but quite tedious.

\vspace{5pt}
\noindent
{{\bf (C-RDS $\implies$ C-RJF) and (0-RDS $\implies $ 0-RJF).}}\,\,
These two implications are true because the proof in Section~\ref{sec:proof RDS=>JF} applies verbatim. Indeed, if we assume that C-RJF fails to hold, we start with a  trajectory $\qvec(\cdot)$ that is driven by a constant rate $\lamvec$ and drives queue $m$ to a positive 
value. The construction of the arrival process in Section~\ref{s:def-stoch-arr} yields a process with a constant rate $\bar\lamvec$. Thus, the same proof  establishes that failure of the C-RJF condition leads to failure of C-RDS; equivalently, C-RDS implies C-RJF. The argument that 0-RDS$\implies$0-RJF is the same.

\vspace{5pt}
\noindent
{{\bf 0-RJF $\notimply$ 0-RDS.}}\,\,
This fact exemplifies the difficulty of obtaining necessary and sufficient conditions in the absence of robustness considerations with respect to the arrival rates.

The argument is simple. Consider a single queue that is served at unit rate, and let $\lambda^*=1$. Suppose that the tail exponent is larger than 1, so that no jumps are allowed. In that case, there is only one possible JF trajectory, which obeys $\dot q = 1-1=0$. When initialized at zero, the JF trajectory stays at zero. Thus, the 0-RJF condition holds for the single queue of interest. On the other hand, as long as the arrivals are not deterministic, the stochastic system is marginally unstable, the expected queue length grows to infinity, and the 0-RDS condition does not hold.

The above example involves a system operating at the boundary of its capacity region (marginally unstable). We can also construct simple examples (involving two queues) in which the system operates in the interior of the stability region,  is stable, 
satisfies the 0-RJF condition, 
but is not 0-RDS. Such an example (which we omit) involves a system that operates  at the threshold between robust delay stability and robust delay instability.

\vspace{5pt}
\noindent
{{\bf 0-RJF $ \notimply $ C-RJF.}} \,\, 
 This is again a simple observation. Consider the same single-queue system as in the previous paragraph, with $\lambda^*=1$. As long as the rate is fixed at 1, the JF trajectory stays at zero, and the 0-RJF condition holds. On the other hand, a small constant perturbation that results in $\lambda>1$ yields a divergent JF trajectory, and therefore the C-RJF condition does not hold.
 
\subsection{Conjectures and open problems}

We list here a number of questions and conjectures. 

\vspace{5pt}
\noindent
{\bf 0-RDS $\overset{?}{\implies}$ C-RDS}
We conjecture that when $\lamvec^*\succ {\bf 0}$,\footnote{
The reason for the condition $\lamvec^*\succ {\bf 0}$ is that if $\lamvec^*={\bf 0}$, then a system is trivially 0-RDS, but a small perturbation that leads to positive arrival rates can result in a delay unstable system.} 0-RDS implies C-RDS. Ultimately, this amounts to showing that the set of positive arrival rate 
vectors $\lambda^*$ for which the system is delay stable (robustly, over all distributions with given tail exponents) is open. The rationale behind this conjecture is that in more standard settings (ordinary stability) the 
set of positive vectors $\lamvec^*$ that lead to a stable system is open.

\vspace{5pt}
\noindent
{\bf C-RJF $\overset{?}{\implies}$  PC-RJF} We conjecture that this implication is true, although we do not see how to establish it. If it is true, it would follow from the diagram in Figure~\ref{f:diagram} that C-RJF and C-RDS are equivalent to G-RDS, UC-RJF, and PC-RJF. 

An indirect approach to establishing the conjecture would be to show that (i) C-RJF $\implies$ C-RDS, and (ii) C-RDS $\implies$ G-RDS. However, this appears to be difficult. Our proof that PC-RJF implies G-RDS involves the set $W$ of points reachable by $\epsilon$-JF trajectories; see Definition~\ref{def:W}. However, when we restrict $\lamvec(\cdot)$ to be constant, this set is no longer $\epsilon$-invariant, and Lemma~\ref{lem:absorb}(ii) fails to go through. 

\vspace{5pt}
\noindent
{\bf Generic considerations.} A fundamental reason behind the mismatch between 0-RJF and G-RDS is that, at least for simple examples, the set of nonnegative nominal rates $\lamvec^*$ for which 0-RJF holds is closed 
whereas the set of positive nominal rates $\lamvec^*$ for which G-RDS holds is open. It is conceivable, however, that one set is the closure of the other, and that the difference between the two sets is just a lower-dimensional boundary. This leads us to the  conjecture that 
0-RJF and G-RDS are generically equivalent.

\begin{conjecture}
	Let us fix a network and some $\gamvec$. The set of nonnegative nominal arrival vectors $\lamvec^*$ for which the 0-RJF condition holds  but G-RDS does not hold has zero Lebesgue measure.
\end{conjecture}

\subsection{The details of the arrival distribution may matter}\label{app:details}
Our discussion so far has been about distributionally robust results, dealing with delay stability for all arrival distributions with the given tail exponents $\gamvec$. 
The reason for this was that 
JF models cannot take into account any further properties of these distributions.

Once we start inquiring about delay stability for a fixed, fully-specified system, the situation is  more complex: necessary and sufficient conditions for delay stability appear to be impossible. We illustrate the situation by stating a positive result  
and disciussing the obstacles in establishing a converse.

\vspace{5pt}
\noindent
\subsubsection{\bf Delay stability implies  the 0-RJF condition, under a regularity assumption} Suppose that a particular system (with a constant arrival rate $\lamvec^*$ and given, i.i.d.~arrival distributions)  is delay stable. Suppose furthermore, that the distribution of each $A_j(t)$ satisfies \eqref{eq:def gamma}, with $\gamma$ replaced by the appropriate $\gamma_j$. Then, it can be shown that the 0-RJF condition holds. 
The argument involves similar ideas as the 
proof in Section~\ref{sec:proof RDS=>JF}. That is, we can show that the stochastic system can track an $\epsilon$-JF trajectory with significant probability.

Note, however, that the 0-RJF condition does not imply delay stability, even under such a regularity assumption. The argument is the same as in our earlier example that showed that the 0-RJF condition does not imply the 0-RDS condition.

\subsubsection{Without a regularity assumption, delay stability need not  imply  the 0-RJF condition} 
In contrast to the above mentioned result, we have strong reasons to conjecture that there  exist systems that are  delay stable and yet, the 0-RJF condition fails to hold. The intuition behind this conjectur is as follows.

Consider a system with two heavy-tailed arrival streams, together with some light-tailed ones.
Suppose that the tail exponents of the heavy-tailed arrivals are larger than 1/3. We can arrange the system so that a JF trajectory drives the light-tailed queue of interest to a positive value if and only if we have one jump at each heavy-tailed queue,  the two jump times are approximately equal, and the jump sizes are comparable (within a constant factor of each other). Such a system will not satisfy the 0-RJF condition.

As in the proof in Section~\ref{sec:proof RDS=>JF}, we might expect that the stochastic system can track this JF trajectory. However, we can arrange the 
arrival process distributions for the two heavy-tailed queues to be such that their supports are wide apart. For example, one distribution may be supported on integers of the form $10^{2i}$ and the other on integers of the form $10^{2i+1}$. 
In that case, equal-size jumps are essentially impossible. As a consequence, the stochastic system should be unable to emulate the JF trajectory,  the instability mechanism suggested by the JF trajectory need not be present, and the queue of interest may turn out to be delay stable.


\subsection{The timing of the jumps}\label{app:timing}
The definition of $\epsilon$-JF trajectories allows for jumps at different times. On the other hand, our examples so far rely on jumps that happen simultaneously. This raises the question whether the RJF condition is equivalent to an analogous condition in which we only consider trajectories with simultaneous jumps.  
It turns out that this is not possible. We give an example with four queues in which an $\epsilon$-JF trajectory drives a certain queue to a positive value, but this is only possible if we allow jumps to occur at different times. 

Consider the system in Fig.~\ref{fig:timing}.  The first queue receives heavy-tailed arrivals, with $\gamma_1=1/2$, while the  three other queues receive light-tailed arrivals. There are three possible service vectors as shown in the figure, and the  arrival rate vector is $\lamvec^*=(1,2,1,1)$. 
Note that the condition $\gamvec^T\nvec\leq 1$ allows up to two jumps at queue 1. If we  restrict to simultaneous jumps, this  essentially limits us to a single jump at queue 1.

Suppose that $q_1$ has a jump of size $27$ at time $0$. Then, the $0$-JF trajectory is piecewise linear,  with breakpoints $\qvec(0)=(27,0,0,0)$, $\qvec(3)=(6,6,0,0)$, $\qvec(5)=(0,2,2,0)$, and $\qvec(9)=\mathbf{0}$. This can be easily verified by noticing that the MW policy chooses service vector $\muvec^1$ for $t\in[0,3)$,  service vector $\muvec^2$ for $t\in(3,5)$; and the service capacity is split between $\muvec^2$ and $\muvec^3$ with ratios $5/8$ and $3/8$, for $t\in(5,9)$. 
It then follows from the form of the above  piecewise linear fluid trajectory that given a single jump at time $0$, $q_4$ will stay at zero for all subsequent times. 
We now argue that this will not be the case if  $q_1$ undergoes two jumps at different times. 

Suppose that $q_1$ has a jump of size $27$ at time $0$ and a jump of size $2$ at time $5$. 
Let $\pvec(\cdot)$ be the associated jumping fluid trajectory. 
Then, right before time $5$ we have $\pvec(5^-)=\qvec(5^-)=(0,2,2,0)$. 
Therefore, after the second jump we have $\pvec(5)=(2,2,2,0)$. 
In this case, $\muvec^3$ will be the dominant service vector for some positive time interval starting from time $5$. 
Since $q_4$ receives no service under $\muvec^3$, it will start to build up and become positive.
 
In this example, the 0-RJF condition fails to hold and the system is not 0-RDS.
On the other hand, if we were to restrict to simultaneous jumps, we would not be able to tell that this is the case. Finally, using the same example, we see  that we should also consider non-simultaneous jumps when examining the C-RJF or PC-RJF conditions.

\begin{figure} 
	\begin{center}
		{\includegraphics[width = .8\textwidth]{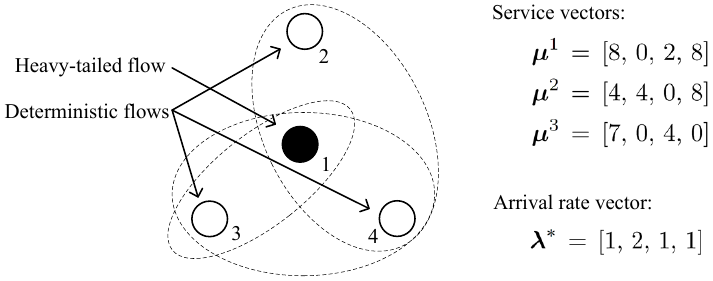}}
		\vspace{0cm}
	\end{center}
	\caption{A network with three light-tailed queues and one heavy-tailed queue, which demonstrates the importance of the timing of multiple jumps. 
	In this example, if $q_1$ undergoes a single jump at time $0$, $q_4$  stays at $0$ . 
	However, two jumps in $q_1$ at suitably arranged times can result in a positive \ $q_4$. 
	}
	\label{fig:timing}
\end{figure}


\medskip
\medskip

\section{Proofs of lemmas for the first direction of Theorem~\ref{th:main} ($\epsilon$-JF $\implies$RDS)} \label{app:proof JF=>RDS}

\subsection{Proof of Theorem~\ref{th:sensitivity}} \label{app:sens-pr}

 We compare the fluid trajectory $\qvec(\cdot)$, which is initialized with $\qvec(0)\neq \Qvec(0)$,  with another fluid trajectory $\tilde\qvec(\cdot)$, initialized  with $\tilde\qvec(0)=\Qvec(0)$.
From the triangle inequality,
$$\big\| \Qvec(t)-\qvec(t)  \big\| \leq \big\| \Qvec(t)-\tilde\qvec(t)  \big\| +
\big\| \tilde\qvec(t)-\qvec(t)  \big\| .
$$
We 
 apply Theorem~\ref{th:sensitivity-old} to bound the first term on the right-hand side. 
 Because of the nonexpansive property of the MW dynamics, we also have
$$\big\| \tilde\qvec(t)-\qvec(t)\big| \leq
\big\| \tilde\qvec(0)-\qvec(0)\big\|=
\big\| \Qvec(0) - \qvec(0)\big\|,$$
and the result follows.

\subsection{Proof of Lemma~\ref{lem:prob event jump if}} \label{appsec:proof lem jump if} 
Let us fix $T$ throughout this proof. Recall the constant 
$\beta\in(1,2)$ introduced in the context of \eqref{def:beta}. 
For $j=1,\ldots,\nn$, we define
$$ 
	\gamma_j' \,=\, \begin{cases} {\gamma_j}/{\beta} \qquad& 
						\mbox{if } \gamma_j< \beta^3, \\ \beta^2 \qquad& 
						\mbox{if }\gamma_j \ge \beta^3.
		\end{cases}.
$$
We then let $\gamvec'=(\gamma'_1,\ldots,\gamma_{\nn}')$.
 and 
$\gamma' \equals \min_{j} \gamma'_j$.
Note that in all cases, we have $\gamma_j'\leq \gamma_j/\beta$, so that $\gamma'\leq \gamma/\beta$. As argued in Claim~\ref{cl:gamma-small}, we can and will (without loss of generality) assume that $\gamma=\min_j \gamma_j\leq 1$ and thus $\gamma'<1$. 
Finally, 
in view of \eqref{def:beta}, it can be seen that for any nonnegative integer vector $\nvec$,
\begin{equation}\label{eq:gamm iff beta}
\mbox{if}\quad  \gamvec^T \nvec >1, \quad \mbox{then}\quad  (\gamvec')^T \nvec \ge\beta^2 >1.
\end{equation}

For every $j$, and according to our definition  \eqref{eq:def gamma process} of the tail exponent of an arrival process, 
there is a random variable  $\Au_j$  that dominates $A_j(t)$, for all $t\ge0$,
and for which all moments of order less than $1+\gamma_j$ are finite; see \eqref{eq:def gamma}.
We  define
\begin{equation}\label{eq:def Gam}
	\Gamma_j \equals \E{{\Au_j} ^{1+\gamma_j'}}, \qquad j=1,\ldots,\nn,
\end{equation}
which is finite because 
$\gamma_j'<\gamma_j$.

For $t=0,\ldots,T-1$, and $j=1,\ldots,\nn$, 
let
$$p_{j,t}=\Pr\big( A_j(t)>\theta_t\big).$$
For  any $j$ and $t\le T-1$, the Markov inequality yields
\begin{equation} \label{eq:pjt le ..}
	\begin{split}
		p_{j,t} &\,=\, \Pr\left(A_j(t)>\theta_t \right) \\
		&\,\le\, \Pr\left(\Au_j>\theta_t\right)\\
		&\,=\, \Pr\left({\Au_j} ^{1+\gamma_j'} >\theta_t^{1+\gamma_j'} \right)\\
		&\,\le\, \frac{\E{{\Au_j} ^{1+\gamma_j'}}}{\theta_t^{1+\gamma_j'} }\\
		&\,=\, \frac{\Gamma_j\,\eta^{1+\gamma_j'}\log^{1+\gamma_j'} (M+T-t)}{(M+T-t)^{1+\gamma_j'} },
	\end{split}
\end{equation}
where the last equality is due to the definitions of $\theta_t$ and $\Gamma_j$ in \eqref{eq:def theta t} and \eqref{eq:def Gam}, respectively.

Let $\phi=1-2^{-1/\nn}$, and note that $0<\phi<1$. 
	Since 
	$0<(1-1/\beta)\gamma' <\gamma_j'$, for every $j$, 
	 there exists some
$M_1\ge 1$ such that if $M\ge M_1$, then 
\begin{equation}\label{eq:def M1 1}
	\Gamma_j \,\eta^{1+\gamma_j'}\, \log^{1+\gamma_j'}M \, \le\, 
	\phi\,\cdot\,\big(\gamma_j' - (1-1/\beta)\gamma'\big)\,\cdot\,  M^{(1-1/\beta)\gamma'}, \quad j=1,\ldots,\nn.
\end{equation}
We fix such an $M_1$. Then, for  $M\ge M_1$, we have
\begin{equation}\label{eq:upper sum p 2}
	\begin{split}
		\sum_{t=0}^{T-1} p_{j,t}\,&\le\, \Gamma_j\,\eta^{1+\gamma_j'} \sum_{t=0}^{T-1} \frac{\log^{1+\gamma_j'} (M+T-t)}{(M+T-t)^{1+\gamma_j'} }\\
		&\leq \, \Gamma_j\,\eta^{1+\gamma_j'} \sum_{\tau=1}^{\infty} \frac{\log^{1+\gamma_j'} (M+\tau)}{(M+\tau)^{1+\gamma_j'} }\\
		&\le\, \phi\big(\gamma_j' - (1-1/\beta)\gamma'\big) \sum_{\tau=1}^{\infty} \frac{ (M+\tau)^{(1-1/\beta)\gamma'}}{(M+\tau)^{1+\gamma_j'} }\\
		&\le\,\phi\big(\gamma_j' - (1-1/\beta)\gamma'\big)\int_{M}^{\infty} \frac{ x^{(1-1/\beta)\gamma'}}{x^{1+\gamma_j'} } \,dx\\
		&=\, \phi\,M^{(1-1/\beta)\gamma' -\gamma_j'},
	\end{split}
\end{equation}
where the first inequality is due to \eqref{eq:pjt le ..} and the third inequality follows from \eqref{eq:def M1 1}.

Let $X_{j,t}$ be the event  $\big\{A_j(t)>\theta_t\big\}$. 
Recall that $N_j$ stands for the number of jumps at queue $j$ during the interval $[0,T)$. 
For any $j$ and any nonnegative integer $n$, we have
\begin{equation}\label{eq:Pr(N_j=n) 2}
	\begin{split}
		\Pr(N_j=n) \,&\le\, \sum_{0\le\tau_1<\cdots<\tau_n\le T-1} \Pr\big(X_{j,\tau_1}\cap\cdots \cap  X_{j,\tau_n}\big)\\
		& =\, \sum_{0\le\tau_1<\cdots<\tau_n\le T-1}  p_{j,\tau_1}\cdots p_{j,\tau_n} \\
		&\le\, \left(\sum_{t=0}^{T-1} p_{j,t}\right)^n\\
		&\le\, \phi^n\,M^{(1-1/\beta)\gamma' n-\gamma_j'n}.
	\end{split}
\end{equation}
where the first inequality follows from the union bound, the  equality is from the independence of the events $X_{j,\tau_1},\ldots, X_{j,\tau_n}$, and the last inequality is due to \eqref{eq:upper sum p 2}. 
Therefore, for any $\nvec\in\Z_+^\nn$, 
\begin{equation}\label{eq:Nvec nvec 2}
	\begin{split}
	\Pr(\Nvec=\nvec) \,&=\, \prod_{j=1}^{\nn} \Pr(N_j=n_j)\\
	\, &\le\, \prod_{j=1}^{\nn} \phi^{n_j}\, M^{(1-1/\beta)\gamma'
	 n_j} \,M^{-\gamma_j'n_j} \\
	\, &=\, \phi^{|\nvec|}\, M^{(1-1/\beta)\gamma'|\nvec|-\nvec^T \gamvec'},
	\end{split}
\end{equation}
where $|\nvec|=n_1+\cdots+n_{\nn}.$
If  $ \gamvec^T \nvec>1$, then \eqref{eq:gamm iff beta} asserts that $\nvec^T\gamvec' \ge\beta^2$. As a result, and since $\gamma'$ is the smallest component of $\gamvec'$, 
\begin{equation}\label{eq:prime not prime ineq 2}
	\left(1-\frac{1}\beta\right)\, \gamma'|\nvec| - \nvec^T\gamvec' \,\le\, \left(1-\frac{1}\beta\right)\, \nvec^T\gamvec'  - \nvec^T\gamvec' 
	\,=\, -\frac{\nvec^T\gamvec' }\beta \, \le\, -\beta.
\end{equation}
Hence,
\begin{equation}
	\begin{split}
		\Pr\left( \gamvec^T \Nvec> 1  \right) \,&=\, \sum_{\nvec:\,\gamvec^T \nvec > 1 }\Pr\left( \Nvec=\nvec  \right)\\
		&\le\, \sum_{\nvec:\, \gamvec^T \nvec > 1 } \phi^{|\nvec|}\, M^{(1-1/\beta)\gamma'|\nvec|-\nvec^T \gamvec'}\\
		&\le\, M^{-\beta} \sum_{\nvec:\, \gamvec^T \nvec > 1 } \phi^{|\nvec|}\\
		&\le \, M^{-\beta} \left(\left(\sum_{n_1,\ldots,n_\nn\ge 0 }   \phi^{|\nvec|}\right)-1\right)\\
		&= \, M^{-\beta} \left(\left(\prod_{j=1}^\nn\sum_{n_j=0}^\infty   \phi^{ n_j}\right)-1\right)\\
		&= \, M^{-\beta} \left(\frac{1}{\left(1-\phi\right)^{\nn}}-1\right)\\
		&= \, M^{-\beta},
	\end{split}
\end{equation}
where the first inequality is from \eqref{eq:Nvec nvec 2} and the second inequality follows from \eqref{eq:prime not prime ineq 2}.
The last equality is because $1-\phi=2^{-1/\nn}$.
This completes the proof of Lemma~\ref{lem:prob event jump if}.



\medskip

\subsection{Proof of Lemma~\ref{lem:Bernstein if}} \label{appsec:proof Bernstein}

The Bernstein inequality (see, e.g.,  (1.21) in Appendix 1 of \citep{AnthB09}) asserts that for any $z\geq 0$, we have
\begin{equation}\label{eq:bern0}
	\Pr\big( |Y-\E{Y}|  \, >\, z \big)\, \le\, 2\exp\left(-\frac{z^2/2}{\sum_{i=1}^n \E{(X_i-\E{X_i})^2} + bz/3 } \right).
\end{equation}
We note that $\E{(X_i-\E{X_i})^2} \le \E{X_i^2} \le \E{X_i b} \le b\lambarmax$. Plugging this into \eqref{eq:bern0},  we obtain
\begin{equation}
	\Pr\big( |Y-\E{Y}|  \, >\, z \big) 
	\,\le\, 2\exp\left(-\frac{z^2/2}{n b\lambarmax + bz/3  }\right). 
\end{equation}
This implies \eqref{eq:lem bern} and completes the proof of Lemma~\ref{lem:Bernstein if}.


\medskip

\subsection{Proof of Lemma~\ref{lem:prob event fluc if}} \label{appsec:proof lem fluc if}

For every $j$, and according to our definition  \eqref{eq:def gamma process} of the tail exponent of an arrival process, 
there is a random variable $\bar{A}_j$  that  dominates $A_j(t)$, for all $t\ge0$,
and for which all moments of order less than $1+\gamma_j$ are finite; see \eqref{eq:def gamma}.
Since $\gamma_j>0$, we have 
$\int_{0}^{\infty}  \Pr\big(\bar{A}_j> x\big)\,dx =\E{\bar{A}_j}<\infty$. 
Therefore, there exists a constant $M_2>0$ such that for any $M> M_2$ and every $j$,
\begin{equation}\label{eq:M2 large}
	\int_{M/\eta\log M}^{\infty} \Pr\big(\bar{A}_j> x\big)\,dx \, \le\, \frac{\gamma\epsilon}{60C\sqrt{\nn}}.
\end{equation}
Note that the choice of $M_2$ is independent of $T$. For the rest of the proof, we fix $M_2$ and assume that $M> M_2$.

Recall that $\Avec^*(\tau)=\min\big\{ \Avec(\tau), \theta_\tau\big\}$, so that 
\begin{equation} \label{eq:tr-eta}
A_j(\tau)-A_j^*(\tau)=\max\big\{0,A_j(\tau)-\theta_\tau\big\}.
\end{equation}
Therefore, 
\begin{equation}\label{eq:truncate-tails}
	\begin{split}
		\E{A_j(\tau)}-\E{A_j^*(\tau)} \,&=\,  \Exp\big[\max\big\{0,A_j(\tau)-\theta_\tau\big\}   \big]\\
		&=\, \int_0^{\infty} \Pr\big( \max\big\{0,A_j(\tau)-\theta_\tau\big\} >x\big)\, dx\\
			&= \, \int_0^{\infty} \Pr\big( A_j(\tau) - \theta_{\tau} >x\big)\, dx\\
		&=\, \int_{\theta_\tau}^{\infty} \Pr\big(A_j(\tau)> x\big) \,dx\\
		&\le\, \int_{M/\eta\log M}^{\infty} \Pr\big(A_j(\tau)> x\big) \,dx\\
		&\le\, \int_{M/\eta\log M}^{\infty} \Pr\big(\bar{A}_j> x\big) \,dx\\
		&\le\, \frac{\gamma\epsilon}{60C\sqrt{\nn}},
	\end{split}
\end{equation} 
where the fourth equality uses a change of variables from $x+\theta_\tau$
to $x$, 
the first inequality uses the fact $\theta_\tau\le M/\eta\log M$ for all $\tau\ge0$ (see the definition \eqref{eq:def theta t} of $\theta_t$), 
the second inequality is because $\bar{A}_j$ dominates  $A_j(\tau)$,
and the last inequality follows from \eqref{eq:M2 large}.
Therefore, for any  $\tau\ge0$,	
\begin{equation}\label{eq:AA*}
	\big\|\E{\Avec(\tau)}-\E{\Avec^*(\tau)}\big\| \,\le\, \sqrt{\nn}\, \max_{j=1,\ldots,\nn} 
	\Big\{ \E{A_j(\tau)}-\E{A_j^*(\tau)} \Big\} \,\le\, \frac{\gamma\epsilon}{60C}.
\end{equation}

We now introduce some ``simpler'' events whose occurrence will be shown to imply the event $\event^{\textrm{fluc}}(T,M)$ that was introduced in~\eqref{eq:def event fluc if}:
\begin{equation}
	\event^* =\Big\{ \Big\| \sum_{\tau=t_0}^{t} \big(\Avec^*(\tau)-\E{\Avec^*(\tau)}  \big) \Big\|  \le \frac{\gamma\epsilon}{30C} \big(M+T-t_0\big),\quad \mbox{for } 0\leq t_0\leq t<T \Big\},	
\end{equation}
and for every $j$,
\begin{equation}
	\event^*_j=\Big\{ \Big| \sum_{\tau=t_0}^{t} \Big(A^*_j(\tau)-\E{A^*_j(\tau)}  \big) \Big|\le \frac{\gamma\epsilon}{30C\nn} \big(M+T-t_0\big),\quad 
	\mbox{for } 0\leq t_0\leq t<T \Big\},
\end{equation}
We observe that the events $\event^*_1,\ldots,\event^*_\nn$ imply the event $\event^*$.
Then,  the union bound, applied to the complements of these events, implies that
\begin{equation}\label{eq:event event*}
	\Pr\big(\event^*\big) \,\ge\, 1- \sum_{j=1}^\nn \big(1-\Pr(\event^*_j)  \big).
\end{equation}

We now argue that the event $\event^*$ implies the event $\event^{\textrm{fluc}}(T,M)$, defined in \eqref{eq:def event fluc if}. Indeed, 
suppose  that event 
$\event^*$ occurs. Then, as long as  $0\le t_0 \leq t < T$, we obtain
\begin{equation}
	\begin{split}
		 &\Big\| \sum_{\tau=t_0}^{t} \big(\Avec^*(\tau)-\lamvec^*  \big) \Big\|  \\
		 &\,=\,  \Big\| \sum_{\tau=t_0}^{t} \big(\Avec^*(\tau)-\E{\Avec^*(\tau)}  \big)\, + \, \sum_{\tau=t_0}^{t} \big(\E{\Avec^*(\tau)} - \E{\Avec(\tau)} \big) \,+\, \sum_{\tau=t_0}^{t} \big(\E{\Avec(\tau)}-\lamvec^*  \big) \Big\|\\
		 &\,\le\, \Big\| \sum_{\tau=t_0}^{t} \big(\Avec^*(\tau)-\E{\Avec^*(\tau)}  \big)\Big\| \, + \, \sum_{\tau=t_0}^{t} \big\|\E{\Avec^*(\tau)} - \E{\Avec(\tau)} \big\| \,+\, \sum_{\tau=t_0}^{t} \big\|\E{\Avec(\tau)}-\lamvec^*   \big\|\\
		 &\,\le\, \frac{\gamma\epsilon}{30C }\, \big(M+T-t_0\big) \, + \, \frac{\gamma\epsilon}{60C }\,(t-t_0) \,+\, \frac{\gamma\epsilon}{20C}\,(t-t_0)\\
		 &\, \le\, \frac{\gamma\epsilon}{10C}\, \big(M+T-t_0\big),
	\end{split}
\end{equation}
where in the second inequality, we used the definition of the event $\event^*$ to bound the first term, 
\eqref{eq:AA*} to bound the second term, and
\eqref{eq:A rate eps2 lamstar} to bound the third term.
The last inequality is because  $t<T$.
Thus, the event $\event^*$ implies the event $\event^{\textrm{fluc}}(T,M)$,  and
\begin{equation}\label{eq:p fluc * *j}
	 \Pr\left( \event^{\textrm{fluc}}(T,M) \right)  \,\ge\,   \Pr\left( \event^* \right) \,\ge\, 1- \sum_{j=1}^\nn \big(1-\Pr(\event^*_j)  \big),
\end{equation}
where the second inequality follows from \eqref{eq:event event*}.

To complete the proof, we derive an upper bound for $1-\Pr(\event^*_j)$. As a first step we obtain a relation between various constants,  which reflects the fact that $\eta$ has been chosen ``large enough''. 
We have 
\begin{equation} \label{eq:eta fractional ineq}
	\begin{split}
		\frac{\gamma^2\epsilon^2\eta }{2(30C\nn)^2 \mumax + 20\gamma\epsilon  C \nn}\,&\ge \,\frac{\gamma^2\epsilon^2 \eta}{1800C^2\nn^2 \mumax + 20\mumax C\nn }\\
		&\ge\, \frac{\gamma^2\epsilon^2}{1820C^2\nn^2 \mumax} \,\cdot\, \eta  \\
		&=\, \frac{\gamma^2\epsilon^2}{1820C^2\nn^2 \mumax} \,\cdot \, \frac{8000 C^2 \nn^2\mumax}{\gamma^2\epsilon^2}\\
		&>\,4,
	\end{split}
\end{equation}
where the first inequality is due to the assumptions $\gamma\le 1$  
and the fact $\epsilon<\mumax$, which is evident 
from the definition 
\eqref{eq:mumax} of $\mumax$; the second inequality is because $C\ge 1$, 
and the equality follows from the definition of $\eta$ in \eqref{eq:def eta}.

Finally, we note that
for any $\tau\ge 0$ and every $j$,
\begin{equation}\label{eq:bound max A*j}
	\E{A^*_j(\tau)} \, \le\, \E{A_j(\tau)} \,\le\, \big\|\E{\Avec(\tau)}\big\| \,\le\, \|\lamvec^*\|+\epsilon \,\le\, \mumax.
\end{equation}

For any $ t_0$ and $t$ with $0\leq t_0\leq t < T$,  
using the fact that $A_j^*(\tau)$ is bounded above by $\theta_{\tau}$,
we have
\begin{equation}\label{eq:prob event*j single}
	\begin{split}
		\Pr&\left(\Big| \sum_{\tau=t_0}^{t} \big(A^*_j(\tau)-\E{A^*_j(\tau)}  \big) \Big|> \frac{\gamma\epsilon}{30C\nn} \big(M+T-t_0\big)\right)\\
		\,&\qquad\le\, 2\exp\left( - \frac{\left(\frac{\gamma\epsilon}{30C\nn} (M+T-t_0)\right)^2}{ 2 \frac{M+T-t_0}{\eta \log(M+T-t_0)} \cdot \left( \mumax(t-t_0+1) + \frac{\gamma\epsilon}{90C\nn} (M+T-t_0)\right)}  \right)\\
		\,&\qquad\le\, 2\exp\left( - \frac{\gamma^2\epsilon^2 \eta \log(M+T-t_0)}{ 2 (30C\nn) ^2 \mumax +  20\gamma\epsilon C\nn }  \right)\\
		\,&\qquad\le\, 2\exp\big(  -4 \log(M+t-t_0)  \big)\\
		&\qquad=\, \frac{2}{(M+T-t_0)^4 },
	\end{split}
\end{equation}
where the first inequality follows from Lemma~\ref{lem:Bernstein if} with  $z=({\gamma\epsilon}/{30C\nn})(M+T-t_0)$, $b=(M+T-t_0)/\eta \log(M+T-t_0)=\theta_{t_0}\ge \theta_\tau$, $\lambarmax = \mumax$ (see \eqref{eq:bound max A*j}), and $n=t-t_0+1$; the second inequality is because $t-t_0+1< M+ T-t_0$, and the third inequality is due to \eqref{eq:eta fractional ineq}.
Therefore, for every $j$,
\begin{equation}\label{eq:1-pe*j}
	\begin{split}
		1-\Pr(\event^*_j) \,&\le\, \sum_{t_0=0}^{T-1}\sum_{t=t_0}^{T-1}  \Pr\left(\Big| \sum_{\tau=t_0}^{t} \big(A^*_j(\tau)-\E{A^*_j(\tau)}  \big) \Big|> \frac{\gamma\epsilon}{30C\nn} \big(M+T-t_0\big)\right)\\ 
		&\le\, \sum_{t_0=0}^{T-1}\sum_{t=t_0}^{T-1}  \frac{2}{(M+T-t_0)^4 }\\ 
		&=\, \sum_{t_0=0}^{T-1}  \frac{2(T-t_0)}{(M+T-t_0)^4 }\\ 
		&\le\, \sum_{t_0=0}^{T-1}  \frac{2}{(M+T-t_0)^3 }\\ 
		&\le\, \sum_{\tau=1}^{\infty}  \frac{2}{(M+\tau)^3 }\\ 
		&\le\, \int_{0}^{\infty}  \frac{2}{(M+x)^3 } \,dx\\ 
		&=\,M^{-2 } ,\\ 
	\end{split}
\end{equation}
where the first inequality is from the union bound and the second inequality is due to \eqref{eq:prob event*j single}. 
Plugging \eqref{eq:1-pe*j} into \eqref{eq:p fluc * *j}, we obtain $\Pr\left( \event^{\textrm{fluc}}(T,M) \right) \geq 1-\nn \, M^{-2} $. This completes the proof of Lemma~\ref{lem:prob event fluc if}.


\medskip

\subsection{Proof of Lemma~\ref{lem:absorb}} \label{appsec:proof lem absorb}
We start with the proof of the first part,  and assume that $W$ is $\epsilon$-invariant. We will prove the result under the additional assumption that the set $W$ is closed. This 
is without  loss of generality, for the following reason. Given an $\epsilon$-invariant set $W$, let $\overline{W}$ be its closure.
Because fluid trajectories under the MW policy are continuous functions of their initial conditions, it is not hard to see that 	$\overline{W}$ is 
also $\epsilon$-invariant. Once we show the result for closed sets, we  will have established that 	$\overline{W}$ is $\epsilon$-attracting. Finally, since $d(\xvec,W)=d(\xvec,\overline{W})$, for all $\xvec$, we can  conclude that ${W}$ is also $\epsilon$-attracting.

Having assumed that $W$ is closed, 
we now consider 
 a fluid trajectory $\qvec(\cdot)$ initialized  with $\qvec(0)=\qvec_0\succeq {\bf 0}$, with $\qvec_0\not\in W$. Then,
there exists some $\xvec_0\in W$ which is closest to $\qvec_0$. 
Let $\xvec(\cdot)$ be a fluid trajectory  initialized at $\xvec(0)=\xvec_0$ and corresponding to the rate vector 
\begin{equation}\label{eq:def lam e lam*}
	\lamvec \equals \lamvec^* + \epsilon\,\frac{\qvec_0-\xvec_0}{\|\qvec_0-\xvec_0 \|}.
\end{equation} 
Since $W$ is $\epsilon$-invariant, 
and since $\|\lamvec - \lamvec^*\|=\epsilon$, 
we have $\xvec(t)\in W$ for all $t\ge 0$. 
In particular, $d\big (\qvec(t),W\big) \leq \big\|\qvec(t) - \xvec(t)\big\|$, for all $t\geq 0$. Furthermore, equality holds at time $t=0$. This implies that
\begin{equation}\label{eq:dqw in terms of dqx}
	\frac{d}{dt} d\big(\qvec(t),\, W\big)\Big|_{t=0} \,\le\, \frac{d}{dt} \big\|\qvec(t) - \xvec(t)\big\| \Big|_{t=0}.
\end{equation}

From the fluid equations (see Definition~\ref{def:fluid model}), we have $\dot{\xvec}(0)\in \Drift_{\lamvec}(\xvec_0) =  \lamvec -\Sbar(\xvec_0) $.
Equivalently, there exist coefficients $\alpha_\muvec \ge0$, for $\muvec\in \S (\xvec_0)$, such that $\sum_{\muvec\in\S (\xvec_0)} \alpha_\muvec=1$ and  

\begin{equation}\label{eq:dot x}
	\dot{\xvec}(0) = \lamvec - \sum_{\muvec\in\S(\xvec_0)} \alpha_\muvec \muvec.
\end{equation} 
Similarly, let $\beta_\nuvec \ge0$ for 
$\nuvec\in \S(\qvec_0)$
 be a set of coefficients, such that $\sum_{\nuvec\in\S(\qvec_0)} \beta_\nuvec=1$ and 
\begin{equation} \label{eq:dot q}
	\dot{\qvec}(0) = \lamvec^* - \sum_{\nuvec\in\S(\qvec_0)} \beta_\nuvec \nuvec.
\end{equation}

For any $\muvec\in \S(\xvec_0)$ and any $\nuvec\in \S(\qvec_0)$, since by definition $\muvec$ is a maximizer of $\uvec^T\xvec_0$ over all $\uvec\in\S$, we have $(\muvec-\nuvec)^T \xvec_0\ge 0$. Similarly, $(\nuvec-\muvec)^T\qvec_0\ge 0$. Combining these two inequalities, we obtain
\begin{equation}
	( \muvec-\nuvec )^T \big(\xvec_0-\qvec_0\big) \,=\, ( \muvec-\nuvec )^T \xvec_0 \,+\, ( \nuvec-\muvec )^T \qvec_0 \, \ge\,0,\qquad \forall\ \muvec\in\S(\xvec_0),\ 
	\nuvec\in\S(\qvec_0).
\end{equation}
Since $\sum_{\muvec\in\S(\xvec_0)} \alpha_\muvec=\sum_{\nuvec\in\S(\qvec_0)} \beta_\nuvec=1$, it follows that
\begin{equation}\label{eq:MW dir ineq}
	\Big( \sum_{\muvec\in\S(\xvec_0)} \alpha_\muvec\muvec\,-\,  \sum_{\nuvec\in\S(\qvec_0)} \beta_\nuvec \nuvec \Big)^T \big(\xvec_0-\qvec_0\big) \, 
	=\, \sum_{\muvec\in\S(\xvec_0)}\sum_{\nuvec\in\S(\qvec_0)} \alpha_\muvec  \beta_\nuvec
	(\muvec-\nuvec)^T(\xvec_0-\qvec_0)
	\ge\,0.
\end{equation}

Using \eqref{eq:dqw in terms of dqx},  we have
\begin{equation}
	\begin{split}
		\frac{d}{dt} d\big(\qvec(t),\, W\big)\Big|_{t=0} \,&\le\, \frac{d}{dt} \big\|\qvec(t) - \xvec(t)\big\| \Big|_{t=0}\\
		&=\, \frac{\big(\dot{\xvec}(0) - \dot{\qvec}(0)\big)^T \big(\xvec_0 - \qvec_0\big)}{\big\|\xvec_0 - \qvec_0\big\| }\\
		&=\, \frac{1}{\big\|\xvec_0 - \qvec_0\big\| } \Big( \lamvec -  \sum_{\muvec\in\S(\xvec_0)} \alpha_\muvec \muvec   \,-\, \lamvec^* + \sum_{\nuvec\in\S(\qvec_0)} \beta_\nuvec \nuvec  \Big)^T \big(\xvec_0 - \qvec_0\big)  \\
		&=\, \frac{(\lamvec-\lamvec^*)^T \big(\xvec_0 - \qvec_0\big)}{\big\|\xvec_0 - \qvec_0\big\| } \, -\,  \Big(\sum_{\muvec\in\S(\xvec_0)} \alpha_\muvec \muvec   \,-\, \sum_{\nuvec\in\S(\qvec_0)} \beta_\nuvec \nuvec  \Big)^T \frac{\xvec_0 - \qvec_0}{\big\|\xvec_0 - \qvec_0\big\| }  \\
		&\le\, \frac{(\lamvec-\lamvec^*)^T \big(\xvec_0 - \qvec_0\big)}{\big\|\xvec_0 - \qvec_0\big\| } \\
		&=\, \left( \frac{\epsilon\big(\qvec_0-\xvec_0\big)}{\big\|\xvec_0 - \qvec_0\big\|} \right)^T \,\frac{ \big(\xvec_0 - \qvec_0\big) }{\big\|\xvec_0 - \qvec_0\big\| } \\
		&=\, -\epsilon.
	\end{split}
\end{equation}
where the second equality follows from \eqref{eq:dot x} and \eqref{eq:dot q}, the second inequality is due to \eqref{eq:MW dir ineq}, and the fourth equality uses  the definition of $\lamvec$ in \eqref{eq:def lam e lam*}. 
Thus, $W$ is $\epsilon$-attracting. 

For the proof of the second part, 
we fix some $\nvec$ and consider  the set $W(\nvec)$, as in Definition~\ref{def:W}. Let $\xvec$ be an element of $W(\nvec)$. Then, there exists an $\epsilon$-JF($\nvec$) trajectory $\xvec(\cdot)$ that reaches $\xvec$. Consider now a fluid trajectory $\qvec(\cdot)$, corresponding to $\lamvec$, for some $\lamvec$ with $\|\lamvec-\lamvec^*\|\leq \epsilon$, and initialized at $\xvec$. The concatenation of the trajectories $\xvec(\cdot)$ and $\qvec(\cdot)$ is  an $\epsilon$-JF($\nvec$) trajectory, and therefore any point that it can reach also belongs to $W(\nvec)$. Thus, $W(\nvec)$ is $\epsilon$-invariant, as claimed.


\medskip

\subsection{Proof of Lemma~\ref{lem:jump free}} \label{appsec:proof lem jump free}
As in the statement of the lemma, we fix some $M\ge \Mthree$,  and times that satisfy $0\leq t_0<t_1\le T$. We also fix a sample path under which the event 
$\event^{\textrm{fluc}}(T,M)$ occurs,  and the interval $(t_0,t_1)$ is jump-free.

Let $\qvec(\cdot)$ be a fluid trajectory corresponding to the arrival rate $\lamvec^*$ and initialized with $\qvec(t_0+1) = \Qvec(t_0+1)$. 
Then,
\begin{equation} \label{eq:Qq bound 1}
	\begin{split}
		\big\| \Qvec(t_1)-\qvec(t_1)  \big\| 
		&\,  \le  \, C\left(1+\|\lamvec^*\| +\max_{t\in(t_0,t_1)} \big\| \sum_{\tau=t_0+1}^{t} \big(\Avec(\tau) - \lamvec^*\big)  \big\| \right) \\
		&\,  =  \, C\left(1+\|\lamvec^*\| +\max_{t\in(t_0,t_1)} \big\| \sum_{\tau=t_0+1}^{t} \big(\Avec^*(\tau) - \lamvec^*\big)  \big\|  \right) \\
		&\,  \le  \, C\left(1+\|\lamvec^*\| +  \frac{\gammast\epsilon }{10C}\, \big(M+T-t_0\big)   \right) \\
		&\,= \, \frac{\gammast\epsilon  }{10}\,\big(M+T-t_0\big) + (\|\lamvec^*\|+1)C,
	\end{split}
\end{equation}
where the first inequality follows from Theorem~\ref{th:sensitivity}, the first equality is because $(t_0,t_1)$ is jump-free and as a result $\Avec^*(\tau) = \Avec(\tau)$ for all $\tau\in(t_0,t_1)$, and the second inequality is because of the occurrence of the event $\event^{\textrm{fluc}}(T,M)$.\footnote{In case $t_1=t_0+1$, the interval $(t_0,t_1)$ is empty, and the term involving a maximum over $t\in(t_0,t_1)$ is interpreted as zero.}

Moreover, from Lemma~\ref{lem:absorb}, $W(t_0+1)$ is $\epsilon$-attracting. 
Furthermore, because $(t_0,t_1)$ is a jump-free interval, we have $W(t_1) = W(t_0+1)$. 
Therefore,  
\begin{equation}\label{eq:dqW bound 1}
	\begin{split}
	d\big(\qvec(t_1),\,W(t_1)\big) &\, =\,  d\big(\qvec(t_1),\,W(t_0+1)\big)\\
	&\,\le\,  \max\Big\{0,\, d\big(\qvec(t_0+1),\,W(t_0+1)\big)-(t_1-t_0-1)\epsilon\Big\},
	\end{split}
\end{equation}
where  the inequality rests on \eqref{eq:absorb to W}.

We now consider the effect of possible jumps at time $t_0$. 
Let $k\ge 0$ be the number of jumps that occur at time $t_0$, in different entries of $\Avec(t_0)$. Without loss of generality, suppose that $A_1,\ldots,A_k$ undergo a jump at time $t_0$. Let $\Jvec$ be an $\nn$-dimensional vector with the first $k$ entries equal to $A_1(t_0),\ldots,A_k(t_0)$ and all other entries equal to zero.
Then, $\Avec(t_0)-\Jvec$ is an $\nn$-dimensional vector whose first $k$ entries are zero;  all other entries are jump-free and are therefore bounded by $\theta_{t_0}=(M+T-t_0)/\big(\eta\log (M+T-t_0)\big)$.
Thus,
\begin{equation}\label{eq:A J}
	\big\| \Avec(t_0)-\Jvec\big\| \, \le\, \frac{(M+T-t_0)\nn}{\eta\log (M+T-t_0)}.
\end{equation}

A key consequence of our definition of the sets $W(t)$ is that 
if $t_0$ is a jump time, then some components of the vector $\Nvec(t_0)$ 
are larger than those of $\Nvec(t_0-1)$, so that the set $W(t_0+1)$ is larger than the set $W(t_0)$. 
In particular, if 
$\xvec\in W(t_0)$, then $\xvec+\Jvec\in W(t_0+1)$.
Now, let $\xvec$ be a point in the closure of $W(t_0)$ which is 
closest to $\Qvec(t_0)$, i.e.,
$$\xvec \in \argmin{\yvec\in \mbox{\small closure}(W(t_0))} \| \yvec - \Qvec(t_0)\|.$$ 
Since $\xvec\in \mbox{closure}\big(W(t_0)\big)$, it follows 
that $\xvec+\Jvec\in \mbox{closure}\big(W(t_0+1)\big)$. 
Therefore, 
\begin{equation}\label{eq:d qW t0+1}
	\begin{split}
		d\big(\qvec(t_0+1),\,W(t_0+1)\big)\, &\le\, d\big(\qvec(t_0+1),\,\xvec+\Jvec\big)\\
		&=\, d\Big( \big[\Qvec(t_0)-\muvec(t_0)\big]^+ + \Avec(t_0),\, \xvec+\Jvec\Big)\\
		&\le\,  d\Big( \big[\Qvec(t_0)-\muvec(t_0)\big]^+ + \Avec(t_0),\, \Qvec(t_0) + \Avec(t_0)\Big)\\
		&\quad+\, d\Big(\Qvec(t_0) + \Avec(t_0),\, \Qvec(t_0) + \Jvec\Big)\,+\,  d\Big(\Qvec(t_0) + \Jvec, \, \xvec+\Jvec\Big)\\
		&=\,  d\Big( \big[\Qvec(t_0)-\muvec(t_0)\big]^+ ,\, \Qvec(t_0) \Big) \,+\, d\big( \Avec(t_0),\,  \Jvec\big)\,+\,  d\big(\Qvec(t_0) , \, \xvec\big)\\
		&\le\,  \|\muvec(t_0)\| \,+\, d\big( \Avec(t_0),\,  \Jvec\big)\,+\,  d\big(\Qvec(t_0) , \, \xvec\big)\\
		&\le\,   \mumax \,+\, \frac{(M+T-t_0)\nn}{\eta\log (M+T-t_0)} \,+\,  d\big(\Qvec(t_0) , \, \xvec\big)\\
		&=\,   \mumax \,+\, \frac{(M+T-t_0)\nn}{\eta\log (M+T-t_0)} \,+\,  d\big(\Qvec(t_0) , \, W(t_0)\big),
	\end{split}
\end{equation}
where the first inequality is because $\xvec+\Jvec\in \mbox{closure}\big(W(t_0+1)\big)$, the first equality is from the evolution formula for the MW dynamics in \eqref{eq:evolution MW} and the initialization $\qvec(t_0+1)=\Qvec(t_0+1)$ for $\qvec(\cdot)$, the last inequality follows from \eqref{eq:A J}, and the last equality is from the definition of $\xvec$.

Combining \eqref{eq:Qq bound 1}, \eqref{eq:dqW bound 1}, and \eqref{eq:d qW t0+1}, we obtain for $M\ge \Mthree$,
\begin{equation}
	\begin{split}
		d\big(\Qvec(t_1),\, W(t_1) \big)\,
		&\le\, d\big(\Qvec(t_1),\, \qvec(t_1) \big)\,+\,  d\big(\qvec(t_1),\, W(t_1) \big)\\
		&\le\, \frac{\gammast\epsilon \,(M+T-t_0) }{10} + (\|\lamvec^*\|+1)C\,+\,  d\big(\qvec(t_1),\, W(t_1) \big)\\
		&\le\, \frac{\gammast\epsilon \,(M+T-t_0) }{10} + (\|\lamvec^*\|+1)C \,+\,  \max\Big\{0,\, d\big(\qvec(t_0+1),\,W(t_0+1)\big)-(t_1-t_0-1)\epsilon\Big\}\\
		&\le\, \frac{\gammast\epsilon \,(M+T-t_0) }{10} + (\|\lamvec^*\|+1)C \\
		&\quad+\,  \max\left\{0,\, \left[d\big(\Qvec(t_0) , \, W(t_0)\big)+ \mumax + \frac{(M+T-t_0)\nn}{\eta\log (M+T-t_0)} \right] -(t_1-t_0-1)\epsilon\right\}\\
		&\le \, \max\Big\{0,\, d\big(\Qvec(t_0) , \, W(t_0)\big) - (t_1-t_0-1)\epsilon \Big\}\\
		&\quad+ \,   \mumax + \frac{(M+T-t_0)\nn}{\eta\log (M+T-t_0)} + \frac{\gammast\epsilon \,(M+T-t_0)}{10} + (\|\lamvec^*\|+1)C \\
		&\le \, \max\Big\{0,\, d\big(\Qvec(t_0) , \, W(t_0)\big) - (t_1-t_0)\epsilon \Big\}\, + \, \epsilon\\
		&\quad +  \frac{(M+T-t_0)\nn}{\eta\log (M+T-t_0)} + \frac{\gammast\epsilon \,(M+T-t_0) }{10} + (\|\lamvec^*\|+1)C + \mumax \\
		&\le \, \max\Big\{0,\, d\big(\Qvec(t_0) , \, W(t_0)\big) - (t_1-t_0)\epsilon \Big\}\,+ \,  \frac{\gammast\epsilon \,(M+T-t_0) }{6}
	\end{split}
\end{equation}
where the second inequality is from \eqref{eq:Qq bound 1}, the third inequality is due to \eqref{eq:dqW bound 1}, the fourth inequality follows from \eqref{eq:d qW t0+1}, and the last inequality is because $M$ has been chosen large enough, as in \eqref{eq:q:def Mthree}.
This completes the proof of Lemma~\ref{lem:jump free}.

\medskip

\subsection{Proof of Lemma~\ref{lem:deterministic dQW}} \label{appsec:proof lem deterministic dQW}
Let $N(t)$ be the cardinality of $\Nvec(t)$, that is,  the number of jumps up to time $t$. We also use the convention $N(-1)=0$.
  When the sample path is such that the event $\event^{\textrm{jump}}(T,M)$ occurs, then $N(t) \le 1/\gamma$, for all $t< T$. Thus, for any $t\in[0,T)$,
\begin{equation}\label{eq:bet<eps}
	\frac{\big(N(t)+2\big) \,\gamma\epsilon}{3} \,\le\, \frac{\big(1/\gamma+2\big) \,\gamma\epsilon}{3} \, 
	= \, \frac{(1+2\gamma)\epsilon}{3}\, \leq\, \epsilon,
\end{equation}
where we have used our assumption that $\gamma\leq 1$.

We fix a sample path of the arrival process $\Avec(\cdot)$ under which both events  $\event^{\textrm{jump}}(T,M)$ and $\event^{\textrm{fluc}}(T,M)$ occur. We will use  strong induction to prove that for any $t\in [0,T]$,
\begin{equation}\label{eq:strong induction}
	d\big(\Qvec(t),W(t)\big)\, \le\, \frac{\big(N(t-1)+2\big) \,\gamma\epsilon}{6}\, (M+T-t).
\end{equation}

To establish the base case of the induction, we use the assumption $\Qvec(0)={\bf 0}$ and the fact that ${\bf 0}\in W(0)$ (because $\epsilon$-JF trajectories are zero for negative times). Thus, $d\big(\Qvec(0),W(0)\big)=0$, and therefore~\eqref{eq:strong induction} holds for $t=0$.

For the induction step, we consider a time $t_1\in(0,T]$ and assume that \eqref{eq:strong induction} holds for all $t<t_1$. We will show that \eqref{eq:strong induction} also holds for time $t_1$. Let 
\begin{equation}\label{eq:def t0}
	{t}_0 \,\equals\, \max\big\{0, T - (2(T-t_1)+M)\big\} \,=\, \max\big\{0, 2t_1-T-M\big\}.
\end{equation}
Note that either $t_0=0<t_1$ or $t_0=2t_1-T-M\leq t_1-M<t_1$, so that we always have $t_0<t_1$.
We consider two cases:

\vspace{5pt} \noindent
{\bf Case 1.} For the first case, we assume that the interval $({t}_0,t_1)$ is jump-free.\footnote{This includes the case where $t_1=t_0+1$, so that the interval $(t_0,t_1)$ is empty.}
We consider two subcases. 
If $t_0=0$, then
\begin{equation}\label{eq:case1 1}
	\begin{split}
		d\big(\Qvec(t_1),\, W(t_1)  \big) \,&\le\, \max \Big\{ 0,\,  d\big(\Qvec(t_0),\, W(t_0)  \big) - (t_1-t_0) \epsilon    \Big\} \,+\,  \frac{\epsilon\gamma}6\,  (M+T-t_0)\\
		&=\, \frac{\epsilon\gamma}6\,  (M+T-t_0)\\
		&\le\, \frac{\epsilon\gamma}6\,  \big(M+T-(2t_1-T-M)\big)\\
		&=\, \frac{2\epsilon\gamma}6\,  (M+T-t_1),\\
	\end{split}
\end{equation}
where the first inequality is due to Lemma~\ref{lem:jump free}, the first equality is because of the assumptions $t_0=0$ and $\Qvec(0)={\bf 0}$, together with the observation that ${\bf 0}\in W(0)$. The second inequality follows from \eqref{eq:def t0}, which implies that $0=t_0\ge 2t_1-T-M$. In particular, \eqref{eq:strong induction} holds for $t=t_1$.

For the second subcase, we assume that
 $t_0>0$, in which case, $t_0=2t_1-T-M$. Then, 
\begin{equation}\label{eq:case1 2}
	\begin{split}
		& d\big(\Qvec(t_1),\, W(t_1)  \big) \, \\
		&\le\, \max \Big\{ 0,\,  d\big(\Qvec(t_0),\, W(t_0)  \big) - (t_1-t_0) \epsilon    \Big\} \,+\,  \frac{\epsilon\gamma}6\,  (M+T-t_0)\\
		&\le\, \max \left\{ 0,\,  \frac{\big(N(t_0-1)+2\big) \,\gamma\epsilon}{6}\, (M+T-t_0)   - (t_1-t_0) \epsilon  \right\} \,+\,  \frac{\epsilon\gamma}6\,  (M+T-t_0)\\
		&=\, \max \left\{ 0,\,  \frac{\big(N(t_0-1)+2\big) \,\gamma\epsilon}{6}\, \big(2M+2T-2t_1\big)   - (M+T-t_1) \epsilon  \right\} \,+\,  \frac{\epsilon\gamma}6\,  (2M+2T-2t_1)\\
		&=\, \max \left\{ 0,\,  \left(\frac{\big(N(t_0-1)+2\big) \,\gamma\epsilon}{3}-\epsilon\right)\,(M+T-t_1)   \right\} \,+\,  \frac{2\epsilon\gamma}6\,  (M+T-t_1)\\
		&=\, \frac{2\epsilon\gamma}6\,  (M+T-t_1),
	\end{split}
\end{equation}
where the first inequality follows from  Lemma~\ref{lem:jump free}, the second inequality is due to the induction hypothesis \eqref{eq:strong induction},  the first equality uses the substitution $t_0 = 2t_1-T-M$, and the last equality is due to \eqref{eq:bet<eps}. 
Thus, \eqref{eq:strong induction} again holds for $t=t_1$.
This completes the induction step for Case~1.

\vspace{5pt}\noindent
{\bf Case 2.} For the second case, we assume that there is a jump in the interval $({t}_0,t_1)$. Let $\hat{t}_0$ be the last jump time in the interval $({t}_0,t_1)$, so that $(\hat{t}_0,t_1)$ is jump free. Since $\hat{t}_0$ is the last jump time, we have 
\begin{equation}\label{eq:Nt1 Nt0}
N(t_1-1) \,=\, N(\hat{t}_0)  \,\ge\, N(\hat{t}_0-1)+1,
\end{equation}
where the inequality is because there is at least one jump at $\hat{t}_0$ (we say ``at least'' because 
at time $\hat{t}_0$, we could have jumps at multiple components of $\Avec(\cdot)$.
Consequently,
\begin{equation}\label{eq:case2}
	\begin{split}
		& \hspace{-30pt} d\big(\Qvec(t_1),\, W(t_1)  \big) \\
		&\le\, \max \Big\{ 0,\,  d\big(\Qvec(\hat{t}_0),\, W(\hat{t}_0)  \big) - (t_1-\hat{t}_0) \epsilon    \Big\} \,+\,  \frac{\epsilon\gamma}6\,  (M+T-\hat{t}_0)\\
		&\le\, \max \left\{ 0,\,  \frac{\big(N(\hat{t}_0-1)+2\big) \,\gamma\epsilon}{6}\, (M+T-\hat{t}_0)   - (t_1-\hat{t}_0) \epsilon  \right\} \,+\,  \frac{\epsilon\gamma}6\,  (M+T-\hat{t}_0)\\
		&=\, \max \left\{ \frac{\epsilon\gamma}6\,  (M+T-\hat{t}_0),\,  \frac{\big(N(\hat{t}_0-1)+3\big) \,\gamma\epsilon}{6}\, (M+T-\hat{t}_0)   - (t_1-\hat{t}_0) \epsilon  \right\} \\
		&\le\, \max \left( \frac{\epsilon\gamma}6\,  (M+T-\hat{t}_0),\,  \frac{\big(N(t_1-1)+2\big) \,\gamma\epsilon}{6}\, (M+T-\hat{t}_0)   - (t_1-\hat{t}_0) \epsilon  \right) \\
		&\le\, \max \left\{ \frac{\epsilon\gamma}6\,  (2M+2T-2t_1),\,  \frac{\big(N(t_1-1)+2\big) \,\gamma\epsilon}{6}\, (M+T-t_1)  \right. \\
		&\left. \hspace{45pt}+ \left(\frac{\big(N(t_1-1)+2\big) \,\gamma\epsilon}{6} -\epsilon\right) (t_1-\hat{t}_0)  \right\} \\
		&\le\, \max \left\{ \frac{2\epsilon\gamma}6\,  (M+T-t_1),\,  \frac{\big(N(t_1-1)+2\big) \,\gamma\epsilon}{6}\, (M+T-t_1)   \right\} \\
		&=\,\frac{\big(N(t_1-1)+2\big) \,\gamma\epsilon}{6}\, (M+T-t_1) ,\\
	\end{split}
\end{equation}
where the first inequality follows from  Lemma~\ref{lem:jump free} and the assumption that $(\hat{t}_0,t_1)$ is jump-free, the second inequality is due to the induction hypothesis \eqref{eq:strong induction}, the third inequality is from \eqref{eq:Nt1 Nt0}, the fourth inequality is because $\hat{t}_0\ge t_0\ge 2t_1-T-M$, and  the last inequality  is due  to \eqref{eq:bet<eps}.  
Therefore, the induction step goes through for Case~2 as well. This completes the proof of the induction, and implies \eqref{eq:strong induction}, for all $t\in [0,T]$.

Finally, letting $t=T$,  \eqref{eq:strong induction} becomes
\begin{equation}
	d\big(\Qvec(T),W(T)\big)\, \le\, \frac{\big(N(T-1)+2\big) \,\gamma\epsilon}{6}\, M \,\le\, \frac{M\epsilon}2,
\end{equation}
where the last inequality is due to \eqref{eq:bet<eps}. This completes the proof of Lemma~\ref{lem:deterministic dQW}.

\medskip

\section{Proof of  lemmas for  the second direction of Theorem~\ref{th:main}: RDS$\implies\epsilon$-JF condition}\label{app:proof RDS=>JF}

\subsection{Proof of Lemma~\ref{l:convenient}} \label{ap:convenient} 
Suppose that the $\epsilon$-JF($\gamvec$) condition fails to hold. Then, there exists 
a nonnegative integer  vector $\nvec'$, with $\gamvec^T\nvec'\leq 1$, an 
$\epsilon$-JF($\nvec'$) trajectory $\qvec'(\cdot)$, and some time $T$ such that
$q'_m(T)>0$. If $T=0$, we can use right-continuity of trajectories to see that, without loss of generality, we can take $T$ to be positive. 
We then define a new trajectory $\qvec(\cdot)$ by letting
$\qvec(t)=\qvec'(tT)/T$, for all $t\ge0$.
It is not hard to verify that $\qvec(\cdot)$ is also an $\epsilon$-JF($\nvec'$) trajectory, and satisfies
$q_m(1)=q'_m(T)/T>0$, so that the first property is satisfied.

Suppose now that some of the jumps of $\qvec(\cdot)$ happen after time 1. Let
$\nvec$ be the vector that counts the number of jumps that take place until time 1.  Starting with $\qvec(\cdot)$, we eliminate the jumps that happen after time 1, to obtain an $\epsilon$-JF($\nvec$) trajectory, with 
$\gamvec^T \nvec  \leq \gamvec^T \nvec' \leq 1$,  all jumps in $[0,1]$, and $q_m(1)>0$.
By slightly perturbing the jump times, and using  a continuity argument, we 
can ensure  that no two components have simultaneous jumps, and also  that
all jump times  belong to $(0,1)$, so that properties (b) and (c) in Lemma~\ref{l:convenient} are satisfied.

Finally, we can replace the arrival rates $\lambda_j(t)$ that drive the $\epsilon$-JF trajectory by $\max\{\lambda_j(t),\epsilon'\}$, where $\epsilon'<\epsilon$ is a small positive constant. This ensures that $\inf_t \lambda_j(t)>0$. Furthermore, using a continuity argument, and as long as $\epsilon'$ is small enough, the property $q_m(1)>0$ is preserved. This proves condition (d) in Lemma~\ref{l:convenient}, and concludes the proof of the lemma.

\medskip

\subsection{Proof of Lemma~\ref{lem:bound episode jump}} \label{appsec:proof lem bound episode jump}
For simplicity, and without loss of generality, we present the proof for the case where  $t_0=0$. 
We let $\bar{T}_1$ be a large enough constant such that for all $T\geq \bar{T}_1$ and all $k$, we have
\begin{equation} \label{eq:adT mumax}
(a_\kappa-\dd )T\,\ge\, \max\{\mumax, \sqrt{T}\}.
\end{equation}
 This is possible because according to the definition of $\dd $ in \eqref{eq:delta'}, we have $\min_\kappa a_\kappa>\dd $.

We fix some $T\ge\bar{T}_1$, as well as  some $\kappa\in\{1,\ldots,n\}$. We
 aim to show that the process has a substantial probability of a jump of size approximately $a_kT$ during the interval $\big[\Theta_k T, (\Theta_k+d)T\big)$. 
Within the proof of this lemma, we use the symbol $j$ (instead of $j_k$) to denote 
 the index of the queue at which the $k$th jump took place.

From \eqref{eq:adT mumax}, we have 
$\log \big((a_k-d) T\big) \geq (1/2)\log T$.
 We then obtain,
for any $t$, any $j$, and any $x\in \big[(a_\kappa-\dd )T, a_\kappa T\big]$, 
\begin{equation}\label{eq:fA decreasing}
	\begin{split}
	f_{A_{j}(t)}(x) \,&=\, \frac{\bar\lambda_j(t)}{\sigma(\gamma_j)}\cdot x^{-(2+\gamma_j)}\log(x+1)\,\one\left(x\ge\mumax\right) \\
	&=\, \frac{\bar\lambda_j(t)}{\sigma(\gamma_j)}\cdot x^{-(2+\gamma_j)}\log(x+1) \\
	&\ge\, \frac{\bar\lambda_j(t)}{\sigma(\gamma_j)}\cdot (a_\kappa T)^{-(2+\gamma_j)}\log\big((a_\kappa-d) T\big)\\
	&\ge\, \zeta T^{-(2+\gamma_j)}\log T,
	\end{split}
\end{equation}
where the second equality is due to \eqref{eq:adT mumax}, and where 
$\zeta$ is a positive constant chosen so that $\zeta \leq  \bar\lambda_i(t)/2 \sigma(\gamma_i)a_\kappa ^{(2+\gamma_{j})}$, for every $i$,  $k$, and $t$. Note that $\zeta$ can be taken positive because, according to Lemma~\ref{l:convenient}, we can assume that 
$\inf_t \bar\lambda_j(t)>0$. 

    We define $\phi=\zeta d$. Then,
\begin{equation}\label{eq:prob A {s alpha}(t)}
	\begin{split}
		\Pr\Big(A_{j}(t)\in \big[(a_\kappa-\dd )T, a_\kappa T\big]\Big) & \,=\, \int_{(a_\kappa-\dd )T}^{a_\kappa T} f_{A_{j}(t)}(x)\,dx\\
		&\ge\, (dT)\cdot \zeta T^{-(2+\gamma_j)}\log T\\
		&=\, \phi\, T^{-(1+\gamma_j)}\log T.
	\end{split}
\end{equation}

As in \eqref{eq:B-cumul}, but with $t_0=0$, we define
\begin{equation}\label{eq:B-cumul-new}
	\Bvec_\kappa \,=\, \sum_{t =  \Tet_\kappa T}^{\Tet_\kappa T +\dd T -1} \Avec(t),
\end{equation}
and, for $t\in \big[\Tet_\kappa T, (\Tet_\kappa+\dd )T\big)$,
\begin{equation}\label{eq:def Ut}
	\Uvec_t \equals 
	\Bvec_\kappa -  A_{j}(t)\evec_{j},
\end{equation}
and note that $\|\Uvec_t\| \leq \|\Bvec_k\|$, for every $k$ and $t$. 
We have
\begin{equation}\label{eq:bound prob sum A 1/4}
	\begin{split}
		\Pr\Big(\big\|\Uvec_t\big\|  \,
		\ge\, 2\dd T\mumax \Big)\,
		&\le\, \frac{\Exp\big[\|\Uvec_t\|\big]}{2\dd T\mumax}\\
		&\le\, \frac{\sum_{\tau=\Tet_\kappa T}^{(\Tet_\kappa+\dd )T-1} \E{\|\Avec(\tau)\|}}{2\dd T\mumax}\\
		&\le\, \frac{\dd T \big(\lamax\big)} {2\dd T\mumax}\\
		&\leq\,\frac{1}2.
	\end{split}
\end{equation}
where the first step made use of  the Markov inequality, and the last step used the fact $\lamax\leq\mumax$.

For $t\in \big[\Tet_\kappa T, (\Tet_\kappa+\dd )T\big)$, let $\event_t$ be the event that $A_{j}(t)\in \big[(a_\kappa-\dd )T, a_\kappa T\big]$ and  $\|\Uvec_t\|  \,\le\, 2\dd T \mumax$. 
Note that the term $A_j(t)$ is omitted from $U_t$ and, therefore, $A_j(t)$ and $U_t$ are independent. Thus,
using \eqref{eq:prob A {s alpha}(t)} and \eqref{eq:bound prob sum A 1/4}, we obtain 
\begin{equation}\label{eq:prob event t}
	\Pr\big(\event_t\big)\,\ge\, \frac{\phi}{2 }\, T^{-({1+\gamma_{j}})}\log T, \quad \forall \ t \in  \big[\Tet_\kappa T, (\Tet_\kappa+\dd )T\big).
\end{equation}
In light of the definition of $\dd $ in \eqref{eq:delta'}, we have
$\dd (1+2\mumax)< a_\kappa$. Therefore,
\begin{equation}\label{eq:aux ineq}
	2\dd T\mumax < (a_\kappa-\dd )T.
\end{equation}
Thus, for any $t,\tau\in \big[\Tet_\kappa T, (\Tet_\kappa+\dd )T\big)$ with $\tau\ne t$, if $\|\Uvec_t\|  \,\le\, 2\dd T\mumax$, then
\begin{equation}\label{eq:akappa Askappa ineq}
	A_{j}(\tau) \,\le\, \|\Uvec_t\|\,\le\,  2\dd T\mumax \,<\, (a_\kappa-\dd )T.
\end{equation}
In the above, 
 the first inequality follows because if $\tau\neq t$, then $A_j(\tau)$ is one of the summands in the
definition \eqref{eq:def Ut} of $\Uvec_t$, and the last inequality is due to \eqref{eq:aux ineq}.
Consequently, if $\|\Uvec_t\|  \,\le\, 2\dd T\mumax$, then $\event_\tau$ holds for no $\tau\in \big[\Tet_\kappa T, (\Tet_\kappa+\dd )T\big)$ with $\tau\ne t$, i.e., for $t\ne\tau$, the events $\event_t$ and $\event_\tau$ are disjoint. 

We now claim  that  for any $t\in \big[\Tet_\kappa T, \Tet_\kappa T+\dd T\big)$, the event $\event_t$ implies the event $\event_\kappa^{\textrm{jump}}$, defined in~\eqref{eq:E-jump}. Indeed, 
when event $\event_t$ occurs, we obtain
\begin{equation}
	\begin{split}
	\big\| \Bvec_\kappa - a_\kappa T \evec_{j} \big\|  
	&\,=\,	\big\| \Uvec_t + A_{j}(t) \evec_{j}  -a_\kappa T \evec_{j}    \big\|\\
	&\,\le\,	\| \Uvec_t\| + \big\|A_{j}(t) \evec_{j}  -a_\kappa T \evec_{j}    \big\|\\
	&\,=\,\| \Uvec_t\| + |A_{j}(t)  -a_\kappa T\big|\\
	&\,\le\, 2\dd T \mumax + \dd T,
	\end{split}
\end{equation}
where the last inequality follows from the definition of $\event_t$. This shows that the event $\event_\kappa^{\textrm{jump}}$ occurs, as claimed. Let $\psi=\min\{1,\dd\phi/2\}$. We then have 
\begin{equation}
	\begin{split}
	\Pr\big(\event_\kappa^{\textrm{jump}}\big) 
	&\,\ge\, \Pr\left(\bigcup_{t\in [\Tet_\kappa T,\Tet_\kappa T +\dd  T)} \event_t\right)\\
	&\,=\,\sum_{t=\Tet_\kappa T}^{(\Tet_\kappa +\dd )T-1} \Pr(\event_t)\\
	&\,\ge\, \dd T\cdot  \frac{\phi T^{-({1+\gamma_{j}})} \log T}{2 } \\
	&\, \geq\, \psi T^{-\gamma_{j}}\log T,
	\end{split}
\end{equation}
where the first equality is because, for $t\ne\tau$, the events $\event_t$ and $\event_\tau$ are disjoint, the second inequality is due to \eqref{eq:prob event t}, and the last one uses the definition of $\psi$.
This completes the proof of Lemma~\ref{lem:bound episode jump}.

\medskip

\subsection{Proof of Lemma~\ref{lem:bound episode fluc}} \label{appsec:proof lem bound episode fluc}

This proof is similar to the proof of Lemma~\ref{lem:prob event fluc if} in Appendix~\ref{appsec:proof lem fluc if}, although some of the details are different.

Recall that $r$ is the number of piecewise constant pieces in the trajectory $\qvec^{\epsilon}(\cdot)$. 
We fix some $\kappa\in\{0,1,\ldots,n\}$  and let 
\begin{equation}\label{eq:def alpha  gam delta}
	\alpha\,\equals\, \frac{\gamma c}{32C \rr},
\end{equation}
and
\begin{equation}\label{eq:eta eta alpha}
	\eta\,\equals \, \frac{{8}\nn \big(6\nn \mumax+\alpha\big)}{3\alpha^2}.
\end{equation}

\begin{claim}\label{cl:t-two}
	There exists a  $\bar{T}_2\ge  8\nn$ such that for any $T\ge\bar{T}_2$, any $j$, and any $t \in \big[\Tet_\kappa T+\dd T,\Tet_{\kappa+1}T\big)$, we have 
	\begin{equation}\label{eq:M2 large reverse}
	 	\int_{T/\eta\log T}^{\infty} \Pr\big(A_j(t)\ge x\big)\,dx \, \le\, \frac{\alpha}{2 \sqrt{\nn}},
	\end{equation}
	and 
	\begin{equation}\label{eq:p1lt}
		\Pr\left(A_j(t)> \frac{T}{\eta\log T} \right)\le \frac{1}{4\nn T}.
	\end{equation}
\end{claim}
\begin{proof}
Since all $\gamma_j$ are positive, we can fix some $\delta$ such that $0<\delta<\gamma_j$, for every $j$. 
Then, the density $f_{A_j(t)} (\cdot)$ in Definition~\ref{def:episode} decays at least as fast as $x^{-(2+\delta)}$. More concretely, there exists a constant $\chi$ such that for all $j$ and $t$, we have 
$$f_{A_j(t)} (x) \leq \chi x^{-(2+\delta)}, \qquad \forall\ x\geq\mumax.$$
We then have,
	for any time $t\in \big[\Tet_\kappa T+\dd T,\Tet_{\kappa+1}T\big)$, any $y\ge 1$, and any $j$, 
	\begin{equation}\label{eq:cdf bound A}
			\Pr\big(A_j(t)\ge y\big)\, \leq\,\, \chi \int_{y}^\infty x^{-(2+\delta)}\, dx
			=\frac{\chi}{1+\delta}\cdot x^{-(1+\delta)}.
	\end{equation}
	 It then follows that as $T$ goes to infinity, both  of $T\, \Pr\left(A_j(\tau)> {T}/{\eta\log T} \right)$ \, and  $\int_{T/\eta\log T}^{\infty} \Pr\big(A_j(\tau)\ge x\big)\,dx$ go to zero, uniformly over all $j$ and $\tau$.  
	  Therefore, there exists a $\bar{T}_2$ such that for any $T\ge\bar{T}_2$, \eqref{eq:p1lt} and \eqref{eq:M2 large reverse} hold, and the claim follows.
\end{proof}

For the rest of the proof, we fix such a $\bar{T}_2$ and assume that $T\ge\bar{T}_2$.
For any $\tau\in[\Tet_\kappa T+\dd T,\Tet_{\kappa+1}T)$ and every $j$, let
\begin{equation}
	A_j^*(\tau) \,\equals\, \min\left\{A_j(\tau), \frac{T}{\eta \log T}\right\}.
\end{equation}
We define the ``no jumps'' event  $\event^=$ as follows:
\begin{equation}
	\event^= =\Big\{ \Avec(\tau) = \Avec^*(\tau), \quad\textrm{for all } \tau\in \big[\Tet_\kappa T+\dd T,\Tet_{\kappa+1}T\big)\Big\}.
\end{equation}
Then,
\begin{equation} \label{eq:pr event alpha}
	\begin{split}
		1-\Pr(\event^=)\,
		&\le\, \sum_{j=1}^\nn \sum_{\tau=\Tet_\kappa T+\dd T}^{\Tet_{\kappa+1}T}\Pr\left(A_j(\tau)> \frac{T}{\eta\log T} \right)\\
		&\le \sum_{j=1}^\nn \sum_{\tau=\Tet_\kappa T+\dd T}^{\Tet_{\kappa+1}T} \frac{1}{4 \nn T}\\
		&\le\, \frac14,
	\end{split}
\end{equation}
where the first inequality uses the union bound,  the second inequality follows from \eqref{eq:p1lt}, and the last inequality uses the fact that $\dd$, as defined in~\eqref{eq:delta'} is no larger than 1.
Moreover, for any $j$ and  any $\tau\in[\Tet_\kappa T+\dd T,\Tet_{\kappa+1}T)$, and using the same steps as in~\eqref{eq:truncate-tails},
$$	\E{A_j(\tau)}-\E{A_j^*(\tau)} \,
		\leq\, \int_{T/\eta\log T}^{\infty} \Pr\big(A_j(\tau)> x\big) \,dx\\
		\le\, \frac{\alpha}{2 \sqrt{\nn}},$$
where  the last inequality follows from \eqref{eq:M2 large reverse}.
Therefore, for any  $\tau\ge0$,
\begin{equation}\label{eq:AA* reverse}
	\big\|\E{\Avec(\tau)}-\E{\Avec^*(\tau)}\big\| \,\le\, \sqrt{\nn}\, \max_{j} \Big(\E{A_j(\tau)}-\E{A_j^*(\tau)}\Big) \,\le\, \frac{\alpha}{2}.
\end{equation}

Consider now the following  events: 
\begin{equation}
	\event^* =\Big\{ \ \Big\| \sum_{\tau=\Tet_\kappa T+\dd T}^{t} \big(\Avec^*(\tau)-\E{\Avec^*(\tau)}  \big) \Big\|  \le \frac{\alpha T}{2}, \quad \forall\ t\in \big[\Tet_\kappa T+\dd T,\Tet_{\kappa+1}T\big)\Big\},
\end{equation}
and for $j=1,\ldots,\nn$,
\begin{equation}
	\event^*_j =\Big\{ \ \Big| \sum_{\tau=\Tet_\kappa T+\dd T}^{t} \big(A^*_j(\tau)-\E{A^*_j(\tau)}  \big) \Big|\le \frac{\alpha T}{2\nn}, \quad\forall\  t \in \big[\Tet_\kappa T+\dd T,\Tet_{\kappa+1}T\big)\Big\}.
\end{equation}
Note that if all of the events $\event^*_1,\ldots,\event^*_\nn$ occur, then $\event^*$ also occurs.  
By applying the union bound to the complement of these events, we have 
\begin{equation}\label{eq:event event* reverse}
	1-\Pr\big(\event^*\big) \,\le\,\sum_{j=1}^\nn \big(1-\Pr(\event^*_j)  \big).
\end{equation}

Consider a sample path under which
 the events $\event^=$  and $\event^*$ occur. Then, for any $t\in \big[\Tet_\kappa T+\dd T,\Tet_{\kappa+1}T\big)$, 
\begin{equation}
	\begin{split}
		\Big\| \sum_{\tau=\Tet_\kappa T+\dd T}^{t}& \big(\Avec(\tau)-\E{\Avec(\tau)}  \big) \Big\|  
		\,=\,	\Big\| \sum_{\tau=\Tet_\kappa T+\dd T}^{t} \big(\Avec^*(\tau)-\E{\Avec(\tau)}  \big) 
		\Big\| \\
		&\,=\,  \Big\| \sum_{\tau=\Tet_\kappa T+\dd T}^{t} \big(\Avec^*(\tau)-\E{\Avec^*(\tau)}  \big)\, + \, \sum_{\tau=\Tet_\kappa T+\dd T}^{t} \big(\E{\Avec^*(\tau)} - \E{\Avec(\tau)} \big) \Big\|\\
		&\,\le\, \Big\| \sum_{\tau=\Tet_\kappa T+\dd T}^{t} \big(\Avec^*(\tau)-\E{\Avec^*(\tau)}  \big)\Big\| \, + \, \sum_{\tau=\Tet_\kappa T+\dd T}^{t} \big\|\E{\Avec^*(\tau)} - \E{\Avec(\tau)} \big\|\\
		&\,\le\, \Big\| \sum_{\tau=\Tet_\kappa T +\dd T}^{t} \big(\Avec^*(\tau)-\E{\Avec^*(\tau)}  \big)\Big\|  \,+\, \frac{\alpha T}{2}\\
		&\,\le\, \frac{\alpha T}{2} \,+\, \frac{\alpha T}{2}\\
		&\,=\, \frac{\gamma c T}{32C \rr},
	\end{split}
\end{equation}
where the first equality is due to $\event^=$,  the second inequality is due to \eqref{eq:AA* reverse},  the third inequality follows from $\event^*$, and the last equality is from the definition of $\alpha$ in \eqref{eq:def alpha  gam delta}.
Therefore,  the events $\event^=$ and $\event^*$ imply the event $\event^{\textrm{fluc}}_\kappa$. 
Using the union bound on the complements of these events,
\begin{equation}\label{eq:p fluc * *j reverse}
1-
\Pr\left( \event^{\textrm{fluc}}_\kappa \right)  \,\leq\, \big(1-\Pr\left( \event^= \right)\big) + \big(1-\Pr\left( \event^* \right)\big) \,
\leq\, \frac14+
\sum_{j=1}^\nn \big(1-\Pr(\event^*_j)  \big),
\end{equation}
where the last inequality follows from \eqref{eq:pr event alpha} and \eqref{eq:event event* reverse}. 
To complete the proof, we develop
an upper bound on $1-\Pr(\event^*_j)$. 

For any $\tau\ge 0$ and every $j$, and as in~\eqref{eq:bound max A*j}, we have 
\begin{equation}\label{eq:bound max A*j reverse}
	\E{A^*_j(\tau)} \, \le\, \E{A_j(\tau)} \,\le\, \big\|\E{\Avec(\tau)}\big\|  \,\le\, \mumax.
\end{equation}
For any $t\in \big[\Tet_\kappa T+\dd T,\Tet_{\kappa+1}T\big)$,  we have
\begin{equation}\label{eq:prob event*j single reverse}
	\begin{split}
		\Pr&\left(\Big| \sum_{\tau=\Tet_\kappa T+\dd T}^{t} \big(A^*_j(\tau)-\E{A^*_j(\tau)}  \big) \Big|> \frac{\alpha T}{2\nn} \right)\\
		\,&\qquad\le\, 2\exp\left(-  \frac{\left(\alpha T/2\nn\right)^2}{ 2 \frac{T}{\eta \log T} \cdot \big( \mumax(t-\Tet_\kappa T-\dd T+1) + \alpha T / 6  \nn\big)}  \right)\\
		\,&\qquad\le\, 2\exp\left(  - \frac{3\alpha^2 \eta \log T}{4\nn \big(6\nn \mumax+\alpha\big)} \right)\\
		\,&\qquad=\, 2\exp\big(  -2 \log T \big)\\
		&\qquad=\, 2 T^{-2},
	\end{split}
\end{equation}
where the first inequality follows from Lemma~\ref{lem:Bernstein if} with the identifications $z=\alpha T/2\nn$, $b=T/\eta \log T$, $\lambarmax = \mumax$ (see \eqref{eq:bound max A*j reverse}), and $n=t-\Tet_\kappa T-\dd T+1$; the second inequality is because $n\leq T$; the equality follows from  the definition of $\eta$ in~\eqref{eq:eta eta alpha}.
Therefore, for every $j$,
\begin{equation}\label{eq:1-pe*j reverse}
	\begin{split}
		1-\Pr(\event^*_j) \,&\le\, \sum_{t=\Tet_\kappa T+\dd T}^T  \Pr\left(\Big| \sum_{\tau=\Tet_\kappa T+\dd T}^{t} \big(A^*_j(\tau)-\E{A^*_j(\tau)}  \big) \Big|> \frac{\alpha T}{2\nn}\right)\\ 
		&\le\, \sum_{t=1 }^T  \frac{2}{T^2 }\\ 
		&=\,  2T^{-1}\\
		&\le\,  \frac{1}{4\nn},
	\end{split}
\end{equation}
where the first inequality is from the union bound, the second inequality is due to \eqref{eq:prob event*j single reverse}, and the last inequality is from the condition  $T\ge \bar{T}_2\ge 8\nn$ in Claim~\ref{cl:t-two}.
Plugging \eqref{eq:1-pe*j reverse} into \eqref{eq:p fluc * *j reverse}, we obtain $\Pr\left( \event^{\textrm{fluc}}_\kappa \right) \ge 1/2$. This completes the proof of Lemma~\ref{lem:bound episode fluc}.


\medskip

\subsection{Proof of Lemma~\ref{lem:bound episode deterministic}} \label{appsec:proof lem episode deterministic}
The proof of this lemma is essentially a continuity argument, together with an induction that pieces together  different time segments.

For simplicity, and without loss of generality, we assume that 
the constant $t_0$ in the statement of the lemma is equal to zero.
Note, however, that with this convention $\Qvec(0)$ will in generally be nonzero.

For any $\lamvec,\xvec\in\R_+^\nn$, let
\begin{equation}\label{eq:def xi lambda}
	\xivec_{\lamvec}(\xvec) \equals \dot{\qvec}(0 ),
\end{equation}
where $\qvec(\cdot)$ is the fluid trajectory corresponding to arrival rate $\lamvec$ and initialized with $\qvec(0)=\xvec$. In view of \eqref{eq:fluid evolution}, we have $\xivec_{\lamvec}(\xvec)\in\Drift_\lamvec(\xvec)$. From the definition~\eqref{eq:def S lambda} of   the set  $\Drift_\lamvec(\xvec)$ of possible drifts, 
and the standing assumption $\big\|\lamvec(\tau)-\lamvec^*\big\| \leq \epsilon$, for all $\tau$,
we have 
\begin{equation}\label{eq:ugly drift bound}
	\big\|\xivec_{\lamvec(\tau)}\big( \qvec^{\epsilon}(\tau)\big)\big\|\, \,\le\,  \mumax,\qquad \forall \ \tau\geq 0,
\end{equation}
where $\mumax$, was defined in~\eqref{eq:mumax2}.
By taking into account the jump $a_k\evec_{j_k}$ at time $\Theta_k$ and then integrating the drift $\xivec_{\lamvec(\tau)}$, over the jump-free interval  $[\Theta_k,\Theta_k+d]$, we have
$$
\big\|\qvec^{\epsilon}(\Tet_\kappa+\dd ) - \qvec^{\epsilon}(\Tet_\kappa)\big\|\,\le\, \dd\mumax.
$$
Noting also that $\qvec^{\epsilon}(\Theta_k)=\qvec^{\epsilon}(\Theta_k^-)+
a_k\evec_{j_k}$, we obtain	
\begin{equation}\label{eq:del q bound over del period}
\big\|\qvec^{\epsilon}(\Tet_\kappa+\dd ) - \qvec^{\epsilon}(\Tet_\kappa^-) - a_k\evec_{j_k}\big\|\,\le\, \dd\mumax.
\end{equation}

For $\tau\ge 0$, let 
$\bar{\muvec}^{a}(\tau)\equals \min\big\{\muvec(\tau),\Qvec(\tau)\big\}$, where
the minimum is taken componentwise; thus, $\muvec^{a}(\tau)$ is
the actual service received at time $\tau$.
It then follows from \eqref{eq:evolution MW} that 
\begin{equation}\label{eq:Q mubar A}
	\Qvec(\tau+1)=\Qvec(\tau)-\muvec^{a}(\tau)+\Avec(\tau).
\end{equation}

We start by considering the intervals $[\Theta_kT, (\Theta_k+d)T]$ associated with jumps, for $\kappa=1,\ldots,n$.  We are working with a sample path for which the event   $\event_\kappa^{\textrm{jump}}$ occurs,  Therefore,
\begin{equation} \label{eq:consequence of event t}
	\begin{split}
		\big\|\Qvec\big((&\Tet_\kappa+\dd )T\big) - \big(\Qvec(\Tet_\kappa T)+Ta_\kappa \evec_{j_\kappa}\big)\big\| \\
		&= \, \Big\|\Big(\Qvec(\Tet_\kappa T) \,+\,\sum_{\tau=\Tet_\kappa T}^{\Tet_\kappa T+\dd T-1} \big(\Avec(\tau)-\muvec^{a}(\tau)\big) \Big) - \big(\Qvec(\Tet_\kappa T)+Ta_\kappa \evec_{j_\kappa}\big)\Big\| \\
		&\le\, \Big\|\Big(\sum_{\tau=\Tet_\kappa T}^{\Tet_\kappa T+\dd T-1} \Avec(\tau) \Big) - Ta_\kappa \evec_{j_\kappa}\Big\| \,+\,
		\Big\|\sum_{\tau=\Tet_\kappa T}^{\Tet_\kappa T+\dd T-1}\muvec^{a}(\tau) \Big\|
		\\
		&\ \leq \dd T(1+2 \mumax) + \dd T \mumax\\
	&\ = \dd T(1+3 \mumax), 
	\end{split}
\end{equation}
where the first equality is due to \eqref{eq:Q mubar A}. The second inequality follows from 
the definition \eqref{eq:E-jump} of $\event_\kappa^{\textrm{jump}}$ and the fact $\|\muvec^{a}(\tau)\|\le\|\muvec(\tau)\|\le \mumax$, for all $\tau\ge0$.

Combining \eqref{eq:consequence of event t} with \eqref{eq:del q bound over del period}, it follows that for $\kappa=1,\ldots,n$,
\begin{equation*}
	\begin{split}
		\big\|\Qvec&\big((\Tet_\kappa+\dd )T\big) - T\qvec^{\epsilon}(\Tet_\kappa+\dd )\big\|\\
		& \,\le\, \, \big\|\Qvec\big((\Tet_\kappa+\dd )T\big) - \big(\Qvec(\Tet_\kappa T)+Ta_\kappa \evec_{j_\kappa}\big)\big\| \\
		&	\qquad\, +\,
		\big\|\big(\Qvec(\Tet_\kappa T)+Ta_\kappa \evec_{j_\kappa}\big) - T\big(\qvec^{\epsilon}(\Tet_\kappa^-)+ a_\kappa \evec_{j_\kappa}\big)\big\| \\
		&	\qquad	\, +\, \big\| T\big(\qvec^{\epsilon}(\Tet_\kappa^-)+ a_\kappa \evec_{j_\kappa}\big)-T\qvec^{\epsilon}(\Tet_\kappa+\dd ) \big\| \\
		&\le\,   \big\|\Qvec\big((\Tet_\kappa+\dd )T\big) - \big(\Qvec(\Tet_\kappa T)+Ta_\kappa \evec_{j_\kappa}\big)\big\| \,+\, \big\|\Qvec\big(\Tet_\kappa T\big) - T\qvec^{\epsilon}(\Tet_\kappa^-)\big\| 
		\,+\, \dd T \mumax\\
		&\le\, \dd T(1+3\mumax)\,+\,
		 \big\|\Qvec\big(\Tet_\kappa T\big) - T\qvec^{\epsilon}(\Tet_\kappa^-)\big\|  \,+\, \, \dd T \mumax\\
		&=\, \big\|\Qvec\big(\Tet_\kappa T\big) - T\qvec^{\epsilon}(\Tet_\kappa^-)\big\|  \,+\, \dd T(1+4\mumax),
	\end{split}
\end{equation*}
where the second and third inequalities are due to \eqref{eq:del q bound over del period} and \eqref{eq:consequence of event t}, respectively.
Using the definition \eqref{eq:delta'} of $\dd$, we conclude that
\begin{equation}\label{eq:big tri eq}
\big\|\Qvec\big((\Tet_\kappa+\dd )T\big) - T\qvec^{\epsilon}(\Tet_\kappa+\dd )\big\|
\leq 
\big\|\Qvec\big(\Tet_\kappa T\big) - T\qvec^{\epsilon}(\Tet_\kappa^-)\big\|  \,+\, \frac{\gamma c T}{8}.
\end{equation}

We have so far established that if the two trajectories 
$\Qvec(\cdot)$ and $\qvec^{\epsilon}(\cdot)$ are close at the beginning of an interval 
$[\Theta_k T,(\Theta_k+d)T]$, they are also close at the end. We now need to establish a similar conclusion over intervals of the form $[(\Theta_k+d)T, \Theta_{k+1}T]$. We wish to capitalize on Theorem~\ref{th:sensitivity}. That result however refers to fluid models with  constant arrival rates. In contrast, our stochastic process has a piecewise constant arrival rate and the same holds for the associated JF trajectory. We deal with this issue by applying Theorem~\ref{th:sensitivity} repeatedly, for each one of the piecewise constant segments. 

Let us fix some $k\in\{1,\ldots,n\}$ and recall that $\rr$ is an upper bound on the number of subintervals in $[\Tet_\kappa +\dd ,\Tet_{\kappa+1})$ during which $\lamvec(\cdot)$  stays constant. 
Let us then fix some times $\theta_i$, for $i=1,\ldots,\rr+1$, such that 
$$\Theta_k +\dd = \theta_1<\cdots< \theta_\rr< \theta_{\rr+1} =\Theta_{k+1},$$
and such that  $\lamvec(\cdot)$ is constant during each one of the intervals
$(\theta_i,\theta_{i+1})$, for $i=1,\ldots,\rr$.

Under our assumption that the sample path satisfies the event
$\event_\kappa^{\textrm{fluc}}$, we see that during each one of the intervals
$[\theta_iT,\theta_{i+1}T]$, and for $i=1\ldots, \rr$, we have 
\begin{equation}\label{eq:max to max}
	\begin{split}
		\max_{t\in[\tet_i  T,\tet_{i+1}T)}   \Big\| \sum_{\tau=\tet_i T}^{t}\big(\Avec(\tau) -\bar{\lamvec}(\tau)\big)\Big\|
		\,&\le\,2\max_{t\in[\Tet_\kappa T+\dd T,\Tet_{\kappa+1}T)}   \Big\| \sum_{\tau=\Tet_\kappa T+\dd T}^{t}\big(\Avec(\tau) -\bar{\lamvec}(\tau)\big)\Big\|\\
		\,&\le\, \frac{\gamma c T}{16C \rr }.
	\end{split}
\end{equation}

We now note  that the function $\tilde\qvec^{\epsilon}(\cdot)$ defined by $\tilde\qvec^{\epsilon}(t)=T\qvec^{\epsilon}(t/T)$ is also an $\epsilon$-JF trajectory and, in particular, is a fluid trajectory during each interval $[\theta_iT,\theta_{i+1}T)$.  
We apply Theorem~\ref{th:sensitivity} over this interval.
Using also the fact that $1+\big\| \lamvec(\tau)\big\| \leq \mumax$, we obtain
\begin{equation}\label{eq:small-int}
\big\| \Qvec(\theta_{i+1}T)-T\qvec^{\epsilon}(\theta_{i+1}^{-})\|
\leq \big\| \Qvec(\theta_{i}T)-T\qvec^{\epsilon}(\theta_{i}^{-})\big\| + C\mumax 
+C \cdot\frac{\gamma c T}{16C \rr}.
\end{equation}
By summing such inequalities, for $i=1,\ldots,\rr$, 
and using the facts $\theta_{1}=\Theta_{k}+d$, $\theta_{r+1}=\Theta_{k+1}$, and
$\qvec^{\epsilon}\big((\Theta_k+\dd)^-\big) = 
\qvec^{\epsilon}(\Theta_k+\dd)$,  
we obtain
\begin{equation}\label{eq:bigger-int}
\big\| \Qvec(\Theta_{k+1}T)-T\qvec^{\epsilon}(\Theta_{k+1}^-)\big\|
\leq \big\|\Qvec((\Theta_k+\dd) T)-T\qvec^{\epsilon}(\Theta_k+\dd)\big\| + 
rC\mumax +
\frac{\gamma c T}{16 }.
\end{equation}

We now add \eqref{eq:big tri eq} and \eqref{eq:bigger-int} to obtain 
$$\big\| \Qvec(\Theta_{k+1}T)-T\qvec^{\epsilon}(\Theta_{k+1}^-)\big\|
\leq \big\|\Qvec(\Theta_k T)-T\qvec^{\epsilon}(\Theta_k^-)\big\| + rC \mumax+
\frac{3\gamma c T}{16}.$$
We finally sum these bounds, for $k=1,\ldots,n$, and use  the fact 
$\qvec^{\epsilon}(\Theta_{n+1}^-) =\qvec^{\epsilon}(1^-)$, 
to conclude that
\begin{equation}\label{eq:all-times}
\begin{split}
\big\| \Qvec(T)-T\qvec^{\epsilon}(1^-)\big\|
&\leq\, 
\big\|\Qvec(\Theta_1 T)-T\qvec^{\epsilon}(\Theta_1^-)\big\| + nrC \mumax+
\frac{3n\gamma c T}{16}
\\
&\leq \, \big\|\Qvec(\Theta_1 T)-T\qvec^{\epsilon}(\Theta_1^-)\big\| +  nrC \mumax+
\frac{3c T}{16},
\end{split}
\end{equation}
where we also made use of the property $n\gamma\leq \gamvec^T\nvec\leq 1$.

The interval $[0,\Theta_{1})$ is to be treated a little different, as $\Theta_0=0$ is not a jump time. Even so, the argument leading to \eqref{eq:bigger-int} applies verbatim and shows that 
$$\big\| \Qvec(\Theta_{1}T)-T\qvec^{\epsilon}(\Theta_{1}^-)\big\|
\leq \big\|\Qvec(0)-T\qvec^{\epsilon}(0)\big\| + rC\mumax +
\frac{\gamma c T}{16 }\leq \big\|\Qvec(0)\big\|+rC\mumax+\frac{ c T}{16 },
$$
where we made use of the fact that the fluid trajectory is intialized at zero, and the inequality $\gamma\leq 1$.
Combining with~\eqref{eq:all-times}, 
 we finally conclude that
$$\big\| \Qvec(T)-T\qvec^{\epsilon}(1)\big\|
\leq \big\|\Qvec(0)\big\|+
(n+1)rC\mumax+
\frac{ c T}{4 }.$$
We now let $\bar{T}_3$ be large enough so that, for any $T\geq \bar{T}_3$, we have
$(n+1)rC\mumax+(cT/4)  +(cT/5)  \leq cT/2$. 
As long as $T\ge \bar{T}_3$, and using   the assumption $\| \Qvec(0) \|\le cT/5$, we obtain
\begin{equation*}
		\big\| \Qvec(T)-T\qvec^{\epsilon}(1)\big\| \leq    \frac{ c T}{5} + (n+1)rC\mumax+
		\frac{ c T}{4 }  \leq  \frac{ c T}{2 } .
\end{equation*}
Finally, using $c=\qvec^\epsilon_m(1)$, 
we have
\begin{equation}
	\begin{split}
		Q_m(T)\,&=\, T\qvec^\epsilon_m(1)\,-\, \big(T\qvec^\epsilon_m(1)-Q_m(T)\big)\\
		&\ge\,T\qvec^\epsilon_m(1)\,-\, \big\|T\qvec^{\epsilon}(1)-\Qvec(T)\big\|\\
		&\ge\,T\qvec^\epsilon_m(1)\,-\,  \frac{{c} T }{2}\\
		&=\,{c} T\,-\, \frac{{c} T }{2}\\
		&=\,\frac{{c} T }{2},
	\end{split}
\end{equation}
This completes the proof of Lemma~\ref{lem:bound episode deterministic}.

\ifSIAM
\else 
\end{APPENDICES}

\fi

\end{document}